\RequirePackage{fix-cm}
\documentclass{svjour3}                     
\smartqed  
\usepackage{graphicx}

\usepackage{mathptmx}      
\usepackage{latexsym}
\usepackage{amsmath, amssymb,amsfonts}
\usepackage{natbib}
\usepackage{color}

\usepackage{algorithm}
\usepackage{algpseudocode}
\algnewcommand\algorithmicinput{\textbf{Input:}}
\algnewcommand\INPUT{\item[\algorithmicinput]}

\newtheorem{thm}{Theorem}

\newtheorem{corrolary}{Corollary}
\newtheorem{assumption}{Assumption}

\begin{document}

\title{Non-convex sampling for a mixture of locally smooth potentials}


\author{ Dao Nguyen
}


\institute{Dao Nguyen \at
	      Department of Mathematics, University of Mississippi, Oxford, Mississippi, USA\\
	      \email{dxnguyen@go.olemiss.edu}
}

\date{Received: date / Accepted: date}

\maketitle

\begin{abstract}
The purpose of this paper is to examine the sampling problem through Euler discretization, where the potential function is assumed to be a mixture of locally smooth distributions and weakly dissipative. We introduce $\alpha_{G}$-mixture locally smooth and $\alpha_{H}$-mixture locally Hessian smooth, which are novel and typically satisfied with a mixture of distributions. Under our conditions, we prove the convergence in Kullback-Leibler (KL) divergence with the number of iterations to reach $\epsilon$-neighborhood of a target distribution in only polynomial dependence on the dimension.
The convergence rate is improved when the potential is $1$-smooth and $\alpha_{H}$-mixture locally Hessian smooth. Our result for the non-strongly convex outside the ball of radius $R$ is obtained by convexifying the non-convex domains. In addition, we provide some nice theoretical properties of $p$-generalized Gaussian smoothing and prove the convergence in the $L_{\beta}$-Wasserstein distance for stochastic gradients in a general setting.
\end{abstract}

\section{Introduction}

\label{intro} The task of sampling is crucial to a large number of fields, including computational statistics and statistical learning \citep{cesa2006prediction,chen2018fast,kaipio2006statistical,rademacher2008dispersion,robert2013monte}.
Sampling problems often take the form of:
\[
\pi(\mathrm{x})=\mathrm{e}^{-U(x)}/\int_{\mathbb{R}^{d}}\mathrm{e}^{-U(y)}\mathrm{d}y,
\]
where the function $U(\mathrm{x})$, also known as the potential function. There has been an increased interest in sampling from discretized dynamics, which leaves the objective distribution invariant. Here we study the over-damped Langevin diffusion \citep{parisi1981correlation} associated with $U$, assumed to be continuously differentiable:
\begin{equation}
\mathrm{dY}_{t}=-\nabla U(Y_{t})dt+\sqrt{2}\mathrm{d}B_{t},\label{eq:1}
\end{equation}
where $(B_{t})_{t\geq0}$ is a $d$-dimensional Brownian motion and
its Euler discretization of Eq.\eqref{eq:1} defines on the
following updated equation:
\begin{equation}
\mathrm{x}_{k+1}=\mathrm{x}_{k}-\eta_{k}\nabla U(\mathrm{x}_{k})+\sqrt{2\eta_{k}}\xi_{k},\label{eq:2}
\end{equation}
where $(\eta_{k})_{k\geq1}$ is a sequence of step sizes that can remain constant or decrease to $0$, and $\xi_{k}\sim\mathcal{N}(0,\ I_{d\times d})$
are independent Gaussian random vectors. The Euler discretization is sometimes referred to as the Langevin Monte Carlo (LMC) or the unadjusted Langevin algorithm (ULA). Historically, much of the theory of convergence of sampling has focused on asymptotic convergence without examining dimension dependence in detail. Non-asymptotic convergence rates have recently gained attention, especially those involving polynomial dependence on target distribution dimensions. Under the condition that $U$ is strongly convex and gradient Lipschitz, \citet{dalalyan2017theoretical, durmus2017nonasymptotic, durmus2019high} established ULA convergence in Wasserstein distance and in total variation. Since then, non-asymptotic convergence rates of unadjusted Langevin algorithms for log-concave distributions have been extensively studied in \citep{dalalyan2019user, durmus2019high, durmus2017nonasymptotic, cheng2018convergence, brosse2019tamed}. The requirement for strong convexity for the potential $U$ can be relaxed either by assuming convexity at infinity or dissipativity. When the former condition is satisfied, convergence results in the Wasserstein-$1$ distance have been shown by \citet{cheng2018sharp} and \citet{majka2020nonasymptotic} through the contraction property described in \citet{eberle2016reflection}. For certain conditions, \citet{erdogdu2018global} expanded the non-asymptotic analysis of the Langevin diffusion to a wider range of diffusions. Under the latter assumption,  \citet{xu2018global} improved the convergence rate by directly analyzing the ergodicity of the overdamped Langevin Monte Carlo while \citet{raginsky2017non} established a non-asymptotic estimate in the Wasserstein-$2$ distance. Both methods, however, depend on the number of iterations. Using auxiliary continuous processes and the use of contraction results from \citet{eberle2019quantitative} and \citet{chau2021stochastic}, a convergence rate of 1/2 in the Wasserstein-$1$ distance was obtained.

Nevertheless, the Euler discretization of an underlying Langevin dynamics typically requires $U(\mathrm{x})$ to have Lipschitz-continuous gradients (global smoothness). Frequently, this requirement is too strict and prevents many common applications \citep{durmus2018efficient,kaipio2006statistical,marie2019preconditioned}.
Generally speaking, non-globally smooth potentials arise from two sources: super-linear growth at infinity of the gradient, which drives the smoothness constant grow with radius; weakly smooth gradient,
which causes the convexity non-uniform and the Hessian unbounded. It has been shown that Euler's discretization with super-linearly growing coefficients is unstable due to the fact that the moments of the discretization could diverge to infinity at a finite time. This problem is usually addressed by incorporating a taming technique (e.g. see \citep{hutzenthaler2012strong,sabanis2013note,sabanis2016euler,sabanis2019higher, brosse2019tamed,lovas2020taming,lim2021non}).
The latter weakly smooth conditions are less well known, with only a few works to the best of our knowledge. Firstly, \citet{chatterji2019langevin} has established an original approach to dealing with weakly smooth (possibly non-smooth) potential problems through smoothing. This technique relies on results obtained from the optimization community, in which a Gaussian is used to perturb the gradient evaluating point. They do not demand strong assumptions, such as the existence of proximal maps, composite structure \citep{atchade2015moreau,durmus2018efficient}, or strong convexity \citep{hsieh2018mirrored}. However, \citet{chatterji2019langevin} analyzes over-damped Langevin diffusion in the context of convex potential functions while many applications involve sampling in high dimensional spaces have non-convex settings. Secondly, \citet{erdogdu2020convergence} proposed a very elegant result using tail growth for weakly smooth and weakly dissipative potentials. By using degenerated convex and modified log-Sobolev inequality, they prove that LMC gets $\epsilon$-neighborhood of a target distribution in KL divergence with the convergence rate of $\tilde{O}(d^{\frac{1}{\alpha}+\frac{1+\alpha}{\alpha}(\frac{2}{\beta}-\mathrm{1}_{\{\beta\neq1\}})}\epsilon^{\frac{-1}{\alpha}})$ where $\alpha$ and $\beta$ are degrees of weakly smooth and dissipative defined in the next section. In the same vein, \citep{nguyen2022unadjusted} relaxed the degenerated convex at infinity to the Poincar\'{e} inequality and derive similar results as in these cases. \citep{chewi2021analysis} provided result for chi-squared or R\'{e}nyi divergences using Latala- Oleszkiewicz or modified log-Sobolev inequality, which interpolates between the Poincar\'{e} and log-Sobolev inequalities. \citep{balasubramanian2022towards} proved that averaged Langevin Monte Carlo with $\epsilon$-relative Fisher information after $O(L^2d^2/\epsilon^2)$ iterations using only gradient Lipschitz condition. It is noteworthy, however, that most previous research has not covered mixtures of distributions with different tail growth behaviors, which may limit the range of real-life applications that can be applied to mixtures of distributions. It is therefore the purpose of this work to introduce generalized conditions for mixtures of distributions with different tail growth characteristics.
In particular, we introduce $\alpha_{G}$-mixture locally smooth and $\alpha_{H}$-mixture locally Hessian smooth (defined in the next section), which are typically satisfied with a mixture of distributions. Using our novel conditions, we show that we can work with either super-linear growth at infinity of the gradient or weakly smooth gradient in a mixture, which provides additional applicability while preserving the convergence property. Additionally, we improve the convergence rate when the potential is $\alpha_{G}$-smooth and $\alpha_{H}$-mixture locally Hessian smooth. In all of our results, weak dissipative conditions are used, which are less restricted than strongly convex and log Sobolev conditions. Weak dissipative conditions implies Poincar\'e inequality when $\beta \ge 1$, but the results can be weakened to the case $\beta > 0$.

We develop the results based on the convexification of a non-convex domain as isoperimetric inequality remains unchanged under bounded perturbation. To the best of our knowledge, the convexification results we obtain under $\alpha_G$-mixtures locally smooth are new because we do not require the commonly used strongly convex outside the ball conditions. The KL convergence in the previous section is then extended to include non-strongly convex outside the ball.

A new smoothing scheme based on $p$-generalized Gaussian distribution is also presented. Since this smoothing covers heavy tail distributions as well as lighter tail distributions, it is typically more flexible than Gaussian smoothing. Changing the smoothing scheme to a different distribution than Gaussian has been recognized as potentially improving convergence rate. Here we provide some nice theoretical properties of $p$-generalized Gaussian smoothing and prove the result for stochastic gradients with a very general setting. In addition, we also provide convergence in $L_{\beta}$-Wasserstein distance for the smoothing potential. Our contributions can be outlined as follows.

Assume that potential function $U$, satisfies $\beta$-dissipative.
Note that $\beta$-dissipative condition implies Poincar\'e inequality, however, to give the convergence
rate explicitly, we will assume the Poincar\'e constant is $\gamma.$

First, we prove that ULA achieves the convergence rate in KL-divergence
of
\begin{equation}
O\left(\frac{\gamma^{1+\frac{1}{\alpha_{G}}}d^{\lceil\frac{\ell_{G}+\alpha_{GN}+2}{\beta}\rceil\left(\ell_{G}+\alpha_{GN}+2\right)\left(1+\frac{1}{\alpha_{G}}\right)+\frac{\lceil\frac{4\ell_{G}}{\beta}\rceil}{2}+\frac{\lceil\frac{\left(\ell_{G}+\alpha_{GN}\right)4\alpha_{GN}}{\beta}\rceil}{2}\vee\frac{\lceil2\alpha_{GN}\rceil}{2}}\ln^{\left(1+\frac{1}{\alpha_{G}}\right)}\left(\frac{\left(H(p_{0}|\nu)\right)}{\epsilon}\right)}{\epsilon^{\left(\ell_{G}+\alpha_{GN}+1\right)\left(1+\frac{1}{\alpha_{G}}\right)+\frac{1}{\alpha_{G}}}}\right)
\end{equation}
if the potential is $\alpha_{G}$-mixture locally smooth and

\[
O\left(\frac{\gamma^{2}d^{\lceil\frac{\ell_{H}+\alpha_{HN}+3}{\beta}\rceil\left(\ell_{H}+\alpha_{HN}+3\right)2+\frac{\lceil\frac{4\left(\ell_{H}+\alpha_{HN}\right)}{\beta}\rceil}{2}+\frac{\lceil\frac{\left(\ell_{H}+\alpha_{HN}+1\right)\left(4\alpha_{HN}+4\right)}{\beta}\rceil}{2}}\ln^{2}\left(\frac{\left(H(p_{0}|\nu)\right)}{\epsilon}\right)}{\epsilon^{2\left(\ell_{H}+\alpha_{HN}+2\right)+1}}\right)
\]
if the potential is $\alpha_{H}$-mixture locally Hessian smooth.

Second, our convergence results are improved when the potential are
higher order of smoothness. Specifically, when a potential is $\alpha_{G}$-smooth
and $\alpha_{H}$-mixture locally Hessian smooth, it converges in

\[
O\left(\frac{d^{\frac{\lceil\frac{4\ell_{H}}{\beta}\rceil+\lceil\frac{\left(4\alpha_{HN}+4\right)}{\beta}\rceil}{2\left(\alpha_{H}+1\right)}+\lceil\frac{4}{\beta}\rceil\left(1+\frac{1}{\alpha_{H}+1}\right)}\ln^{\left(1+\frac{1}{\alpha_{H}+1}\right)}\left(\frac{\left(H(p_{0}|\nu)\right)}{\epsilon}\right)}{\epsilon^{\left(1+\frac{2}{\alpha_{H}+1}\right)}}\right).
\]

steps. Third, we apply the result to the case of non strongly convex
outside the ball of radius $R$ and obtain the convergence rate in
KL divergence of
\[\tilde{O}\left(\frac{\left(32C_{K}^{2}d\left(\frac{a+b+2aR^{2}+3}{a}\right)e^{4\left(2\sum_{i}L_{i}R^{1+\alpha_{i}}\right)}\right){}^{1+\frac{1}{\alpha_{G}}}d^{\lceil\frac{\alpha_{GN}+2}{\beta}\rceil\left(\alpha_{GN}+2\right)\left(1+\frac{1}{\alpha_{G}}\right)+\frac{\lceil2\alpha_{GN}\rceil}{2}}}{\epsilon^{\frac{\alpha_{GN}^{2}+2\alpha_{GN}+2}{\alpha_{G}}}}\right)
\]
for potential is $\alpha_{G}$-mixture locally smooth.

Fourth, we extend the result to stochastics gradient by $p$-generalized
Gaussian smoothing and obtain $\tilde{O}\left(\frac{d^{\lceil\frac{2\alpha_{GN}^{2}}{\beta}\rceil\frac{1}{\alpha_{G}}+\lceil\frac{\alpha_{GN}+2}{\beta}\rceil\left(\alpha_{GN}+2\right)\left(1+\frac{1}{\alpha_{G}}\right)}}{\gamma^{1+\frac{1}{\alpha_{G}}}\epsilon^{\left(\alpha_{GN}+1\right)\left(1+\frac{1}{\alpha_{G}}\right)+\frac{1}{\alpha_{G}}}}\right)$
for $\alpha_{G}$-mixture locally smooth with $\ell_{G}=0$. Note
that we also covers results of smoothing potentials, satisfying $\gamma$-Poincaré
inequality, $\alpha_{G}$-mixture locally smooth $\ell_{G}=0$, and
$\beta$-dissipative with convergence rate in $L_{\beta}$-Wasserstein
distance of
\begin{equation}
\tilde{O}\left(\frac{{\displaystyle d^{\frac{2}{\beta}\left(\lceil\frac{2\alpha_{GN}^{2}}{\beta}\rceil\frac{1}{\alpha_{G}}+\lceil\frac{\alpha_{GN}+2}{\beta}\rceil\left(\alpha_{GN}+2\right)\left(1+\frac{1}{\alpha_{G}}\right)\right)+2+\frac{4}{\alpha_{G}}}}}{\gamma_{1}^{\left(1+\frac{1}{\alpha_{G}}\right)}\epsilon^{\left(\alpha_{GN}+1\right)\left(1+\frac{1}{\alpha_{G}}\right)+\frac{1}{\alpha_{G}}}}\right).
\end{equation}

Finally, our convergence results remain valid under finite perturbations,
indicating that it is applicable to an even larger class of potentials.
Last but not least, convergence in KL divergence implies convergence
in total variation and in $L_{2}$-Wasserstein metrics, which in turn
gives convergence rates of $O(\cdot\epsilon^{-\left(6+\frac{8}{\alpha}\right)})$
and $O(\cdot\epsilon^{-\left(6+\frac{8}{\alpha}\right)\beta}d^{6+\frac{8}{\alpha}})$
in place of $O(\cdot\epsilon^{-\left(3+\frac{4}{\alpha}\right)})$
in the first case above, respectively for total variation and $L_{2}$-Wasserstein
metrics.

The rest of the paper is organized as follows. Section 2 sets out
the notation and smoothing properties necessary to give our main results
in section 3. Section 4 apply the result of \citep{nguyen2021unadjusted}
for non-strongly convex outside the ball while Section 5 gives some
simple applications. Section 6 presents our conclusions and possible
directions of extentions. 

\section{Preliminaries}

We furnish the space $\mathbb{R}^{d}$ with the regular $p$-norm
and throughout the paper, we drop the subscript and just write $\Vert x\Vert\stackrel{\triangle}{=}\Vert x\Vert_{2}$
whenever $p=2$. We use $\left\langle \ ,\ \right\rangle $ to specify
inner products and let $\left|s\right|$, for a real number $s\in\mathbb{R}$,
denote its absolute value. For a function $U$ :$\mathbb{R}^{d}\rightarrow\mathbb{R}$,
which is twice differentiable, we use $\nabla U(x)$ and $\nabla^{2}U(x)$
to denote correspondingly the gradient and the Hessian of $U$ with
respect to $x$. We use $A\succeq B$ if $A-B$ is a positive semi-definite
matrix. We use big-oh notation $O$ in the following sense that if
$f(x)={\displaystyle O(g(x))}$ implies $\lim_{x\rightarrow\infty}\sup\frac{f(x)}{g(x)}<\infty$
and $\tilde{O}$ suppresses the logarithmic factors.

While sampling from the exact distribution $\pi(\mathrm{x})$ is generally
computationally demanding, it is largely adequate to sample from an
approximated distribution $\tilde{\pi}(\mathrm{x})$ which is in the
vicinity of $\pi(\mathrm{x})$ by some distances. In this paper, we
use KL-divergence and Wasserstein distance and briefly define them
in Appendix A. We suppose some of the following conditions hold:

\begin{assumption} \label{A0} ($\alpha_{G}$-mixture locally-smooth)
There exist $\ell_{G}\geq0,$ $0<\alpha_{G}=\alpha_{G1}\leq...\leq\alpha_{GN}\leq1$,
$i=1,..,N$ $0<L_{Gi}\leq L_{G}<\infty$ so that $\forall x,\ y\in\mathbb{R}^{d}$,
we obtain $\left\Vert \nabla U(x)-\nabla U(y)\right\Vert \leq\left(1+\left\Vert x\right\Vert ^{\ell_{G}}+\left\Vert y\right\Vert ^{\ell_{G}}\right)\sum_{i=1}^{N}L_{i}\left\Vert x-y\right\Vert ^{\alpha_{Gi}}$
where $\nabla U(x)$ represents a gradient of $U$ at $x$. \end{assumption}

\begin{assumption} \label{A0-1} ($\alpha_{H}$-mixture Hessian locally-smooth)
There exist $\ell_{H}\geq0,$ $0\leq\alpha_{H}=\alpha_{H1}\leq...\leq\alpha_{HN}\leq1$,
$i=1,..,N$ $0<L_{Hi}\leq L_{H}<\infty$ so that $\forall x,\ y\in\mathbb{R}^{d}$,
we obtain $\left\Vert \nabla^{2}U(x)-\nabla^{2}U(y)\right\Vert _{op}\leq\left(1+\left\Vert x\right\Vert ^{\ell_{H}}+\left\Vert y\right\Vert ^{\ell_{H}}\right)\sum_{i=1}^{N}L_{i}\left\Vert x-y\right\Vert ^{\alpha_{Hi}}$
where $\nabla^{2}U(x)$ represents a Hessian of $U$ at $x$. \end{assumption}

\begin{assumption} \label{A3} ($\beta-$dissipativity). There exists
$\beta>0$, $a$, $b>0$ such that $\forall x\in\mathbb{R}^{d}$,
$\left\langle \nabla U(x),x\right\rangle \geq a\left\Vert x\right\Vert ^{\beta}-b.$
\end{assumption}

\begin{assumption} \label{A4} ($LSI\left(\gamma\right)$)
There exists some $\gamma>0,$ so that for all probability distribution
$p\left(x\right)$ absolutely continuous $w.r.t.\ \pi\left(x\right)$,
$H({\displaystyle p\vert\pi)\leq\frac{1}{2\gamma}I(p\vert\pi)}$ where
$H$ and $I$ are Kullback-Leibler (KL) divergence and relative Fisher
information defined respectively in Appendix A below.
\end{assumption}

\begin{assumption} \label{A5} ($PI\left(\gamma\right)$) There exists
some $\gamma>0,$ so that for all smooth function $g\colon\mathbb{R}^{d}\to\mathbb{R}$,
$\mathrm{Var}_{\pi}(g)\le\frac{1}{\gamma}E_{\pi}\left[\left\Vert \nabla g\right\Vert ^{2}\right]$
where $\mathrm{Var}_{\pi}(g)=E_{\pi}[g^{2}]-E_{\pi}[g]^{2}$ is the
variance of $g$ under $\pi$. \end{assumption}
\begin{assumption} \label{A6} (non-strongly convex outside the ball)
For every $\left\Vert x\right\Vert \geq R$, the Hessian of twice
diffentiable potential function $U(x)$ is positive semi-definite,
that is for every $y\in\mathbb{R}^{d}$, ${\displaystyle \left\langle y,\nabla^{2}U(x)\ y\right\rangle \geq0}.$
\end{assumption}
\begin{assumption} \label{A7} The function $U(x)$ has stationary
point at zero $\nabla U(0)=0.$ \end{assumption} \begin{remark}
Assumption \ref{A7} is imposed without loss of generality. Condition
\ref{A0} often holds for a mixture of distribution with different
tail growth behaviors. Condition \ref{A0} is an extension of $\alpha_{G}$-
mixture weakly smooth \citep{nguyen2022unadjusted}, that is when $\ell_{G}=0$, we recover
the normal $\alpha_{G}$ mixture-weakly smooth. When $N=1$ we have
a $\alpha_{G}-$Holder continuity of the gradients of $U$ while $\alpha_{G}=1$
gives us a Lipschitz continuous gradient. Note that when $\ell>0$,
we only have the potential behaves locally smooth. Similarly, condition
\ref{A0-1} extension of $\alpha_{H}$- mixture Hessian smooth when
$\ell_{H}=0$. When $N=1$ and $\alpha_{H}=1,$we get back to the
Hessian smoothness condition.
\end{remark}

A feature that follows straightforwardly from Assumption \ref{A0}
is that for $\forall x,\ y\in\mathrm{\mathbb{R}}^{d}$:
\begin{lemma} If potential $U:\mathbb{R}^{d}\rightarrow\mathbb{R}$
satisfies an $\alpha_{G}$-mixture quasi-smooth for some $\ell_{G}\geq0,$
$0<\alpha_{G}=\alpha_{G1}\leq...\leq\alpha_{GN}\leq1$, $i=1,..,N$
$0<L_{Gi}\leq L_{G}<\infty$ , then:
\begin{equation}
U(y)\leq U(x)+\left\langle \nabla U(x),\ y-x\right\rangle \leq U(x)+\left\langle \nabla U(x),\ y-x\right\rangle +\sum_{i}\left(1+L_{G}\right)\left(\left\Vert x\right\Vert ^{\ell_{G}}+\left\Vert y\right\Vert ^{\ell_{G}}\right)\left\Vert x-y\right\Vert ^{1+\alpha_{Gi}}.\label{eq:4-1}
\end{equation}
In addition, from Assumption \ref{A7}, for any $x\in\mathrm{\mathbb{R}}^{d}$,
\begin{align*}
\left\Vert \nabla U(x)\right\Vert  & \leq L_{G}\left(1+\left\Vert x\right\Vert ^{\ell_{G}}\right)\sum_{i=1}^{N}\left\Vert x\right\Vert ^{\alpha_{Gi}}\\
& \leq2NL_{G}\left(1+\left\Vert x\right\Vert ^{\ell_{G}+\alpha_{N}}\right).
\end{align*}
\end{lemma}
\begin{proof}
See Appendix A2.
\end{proof}
A similar property that follows from Assumption \ref{A0-1} is that for $\forall x,\ y\in\mathrm{\mathbb{R}}^{d}$:
\begin{lemma} If potential $U:\mathbb{R}^{d}\rightarrow\mathbb{R}$
satisfies an $\alpha_{H}$-mixture locally Hessian smooth for some
$\ell_{H}\geq0,$ $0\leq\alpha_{H}=\alpha_{H1}\leq...\leq\alpha_{HN}\leq1$,
$i=1,..,N$ $0<L_{Hi}\leq L_{H}<\infty$, then:
\begin{equation}
\left\Vert \nabla U(y)-\nabla U(x)\right\Vert \leq\left\Vert \nabla^{2}U(x)\right\Vert _{\mathrm{op}}\left\Vert x-y\right\Vert +\sum_{i}\left(1+L_{H}\right)\left(\left\Vert x\right\Vert ^{\ell_{H}}+\left\Vert y\right\Vert ^{\ell_{H}}\right)\left\Vert x-y\right\Vert ^{1+\alpha_{Hi}}.\label{eq:4-1-1}
\end{equation}
In addition, from Assumption \ref{A7}, for any $x\in\mathrm{\mathbb{R}}^{d}$,
let $C_{H}=\left\Vert \nabla^{2}U(0)\right\Vert _{\mathrm{op}}\vee2\sum_{i=1}^{N}L_{Hi}:$
\begin{align*}
\left\Vert \nabla^{2}U(x)\right\Vert _{\mathrm{op}} & \leq\left\Vert \nabla^{2}U(0)\right\Vert _{\mathrm{op}}+\left(1+\left\Vert x\right\Vert ^{\ell_{H}}\right)\sum_{i=1}^{N}L_{i}\left\Vert x\right\Vert ^{\alpha_{Hi}}\\
 & \leq C_{H}\left(1+\left\Vert x\right\Vert ^{\ell_{H}+\alpha_{HN}}\right).
\end{align*}
\end{lemma}
\begin{proof}
See Appendix A2.
\end{proof}
\section{Convergence under Poincar\'{e} inequality \label{sec:Experiment}}
\subsection{Main result: Convergence under Poincar\'{e} inequality, $\beta-$dissipative, $\alpha-$mixture
locally smooth}
\label{Sec:Review-1}

We first review the Langevin dynamics in continuous time under the Poincaré inequality before examining KL divergence in discrete time along the Unadjusted Langevin Algorithm (ULA). The Langevin dynamics for target distribution $\nu\propto e^{-U}$ is a continuous-time stochastic process $(X_{t})_{t\ge0}$ in $\mathbb{R}^{d}$ that proceeds as follows:
\begin{align}
dX_{t}=-\nabla U(X_{t})\,dt+\sqrt{2}\,dW_{t}\label{Eq:LD-1}
\end{align}
where $(W_{t})_{t\ge0}$ is the standard Brownian motion in $\mathbb{R}^{d}$.

If $(X_{t})_{t\ge0}$ is updated by the Langevin dynamics~\eqref{Eq:LD-1}, then their probability density function $(p_{t})_{t\ge0}$ will fulfill the Fokker-Planck equation:
\begin{align}
\frac{\partial p_{t}}{t}\,=\,\nabla\cdot(p_{t}\nabla U)+\Delta p_{t}\,=\,\nabla\cdot\left(p_{t}\nabla\log\frac{p_{t}}{\nu}\right).\label{Eq:FP-1}
\end{align}
As a distribution evolves along the Langevin dynamics, it will get nearer to the target distribution $\pi$. Along the Langevin dynamics~\eqref{Eq:LD-1} (or correspondingly, the Fokker-Planck equation~\eqref{Eq:FP-1}), we have,
\begin{align}
\frac{d}{dt}(\chi^{2}(p_{t}|\nu))=-E_{\pi}\left\Vert \nabla\frac{p_{t}}{\nu}\right\Vert ^{2},\label{Eq:HdotLD-1}
\end{align}
where ${\displaystyle \chi^{2}(p|\nu)\stackrel{\triangle}{=}\int_{\mathbb{R}^{d}}\left(\frac{p(x)}{\pi(x)}\right)^{2}\nu(x)dx-1.}$ $\chi^{2}$ divergence with respect to $\nu$ is decreasing along the Langevin dynamics as
$E_{\nu}\left\Vert \nabla\frac{p_{t}}{\nu}\right\Vert ^{2}\ge0$. In fact, when $\nu$ satisfies Poincar\'e inequality (PI), $\chi^{2}$ divergence converges exponentially fast along the Langevin dynamics. PI is retained under bounded perturbation \citep{holley1986logarithmic}, Lipschitz mapping, tensorization, among others and we will consider potential satisfied PI in this section.

For any fixed step size $\eta>0$, ULA converges to a biased limiting distribution $\nu_{\eta}\neq\nu$, which implies that
$H(p_{k}|\nu)$ does not converge to $0$ along ULA, as it has an asymptotic bias $H(\nu_{\eta}|\nu)>0$.
Here, we can adapt the technique proof of \citep{vempala2019rapid}
to analyze the convergence rate of ULA when the true target distribution $\nu$ satisfies an $\alpha_{G}$-mixture
locally smooth. This discretization technique has been used in many papers, including the
papers \citep{erdogdu2020convergence} and \citep{nguyen2021unadjusted}, but it is non-trivial to apply to our setting.
Our proofs are rested on it and the following key observations. The
first observation is to bound the norm of the gradient to power $r$,
which is rather general in the sense that $r$ could be any real number.
Let $x_{k}$ be the interpolation of the discretized process \eqref{eq:2}
and let $p_{k}$ denote its distribution, $\mathbb{E}_{p_{k}}\left[\left\Vert \nabla U(x_{k})\right\Vert ^{r}\right]$
can be upper bounded by the following lemma. \begin{lemma}\label{lem:C2}Suppose
$\nu$ is $\beta$-dissipative $\beta\geq1$, $\alpha_{G}$-mixture
locally-smooth. Start ULA algorithm from $x_{0}$ with the step size
$\eta>0,$ we have for any $r\in R,$$r\geq0:$
\[
E_{p_{k}}\left[\left\Vert \nabla U(x)\right\Vert ^{r}\right]\leq O\left(d^{\lceil\frac{\left(\ell_{G}+\alpha_{GN}\right)r}{\beta}\rceil}\right).
\]
\end{lemma} \begin{proof} See Appendix \ref{AC2}. \end{proof}
A result that follows directly from the first observation is that:
\begin{lemma}\label{lem:C3} Suppose $\nu$ is $\beta$-dissipative,
$\alpha_{G}$-mixture locally-smooth. If $0<\eta\le\min\left\{ 1,\left(\frac{\epsilon}{2TD}\right)^{\frac{1}{\alpha_{G}}}\right\} $
, then along each step of ULA~\eqref{eq:2},
\begin{align}
\frac{d}{dt}H(p_{k,t}|\nu)\le-\frac{3}{4}I(p_{k,t}|\nu)+\eta^{\alpha_{G}}D,\label{Eq:Main1-2-1-1-1-1-2-1}
\end{align}
where
\[
D=O\left(d^{\frac{\lceil\frac{4\ell_{G}}{\beta}\rceil}{2}+\left(\frac{\lceil\frac{\left(\ell_{G}+\alpha_{GN}\right)4\alpha_{GN}}{\beta}\rceil}{2}\vee\frac{\lceil2\alpha_{GN}\rceil}{2}\right)}\right),
\]
\end{lemma} \begin{proof} See Appendix \ref{AC3}. \end{proof}
The second observation is to bound the KL divergence along the dynamics.
\begin{lemma}\label{lem:C1} Suppose that $\nu$ satisfies $\gamma$-Poincar\'e
inequality, $\alpha_{G}$-mixture locally-smooth, then for any distribution
$\mu$,
\[
H(\mu|\nu)\leq C\gamma^{-\frac{1}{2q}}M_{\ell_{G}+\alpha_{GN}+2}\left(\mu+\nu\right)I^{\frac{1}{\ell_{G}+\alpha_{GN}+2}}\left(\mu|\nu\right),
\]
where $M_{s}(g)=\int g(x)\left(1+\left\Vert x\right\Vert ^{2}\right)^{\frac{s}{2}}dx$
for any function $g$. \end{lemma} \begin{proof} See Appendix \ref{AC1}.
\end{proof} Based on both observations, we are now ready to state
the main result in this section. \begin{thm} \label{thm:C4} Suppose
$\nu$ is $\gamma$-Poincar\'e inequality, $\beta$-dissipative $\beta\geq1$,
$\alpha_{G}$-mixture locally-smooth. For any $x_{0}\sim p_{0}$ with
$H(p_{0}|\nu)=C_{0}<\infty$, the iterates $x_{k}\sim p_{k}$ of ULA~
with step size $\eta$ sufficiently small satisfying the following
conditions
\[
\eta=\min\left\{ 1,\left(\frac{\epsilon}{2TD}\right)^{\frac{1}{\alpha_{G}}}\right\} .
\]
The ULA iterates reach $\epsilon$-accuracy of the target $\nu$ in
KL divergence after
\[
K=O\left(\frac{\gamma^{1+\frac{1}{\alpha_{G}}}d^{\lceil\frac{\ell_{G}+\alpha_{GN}+2}{\beta}\rceil\left(\ell_{G}+\alpha_{GN}+2\right)\left(1+\frac{1}{\alpha_{G}}\right)+\frac{\lceil\frac{4\ell_{G}}{\beta}\rceil}{2}+\frac{\lceil\frac{\left(\ell_{G}+\alpha_{GN}\right)4\alpha_{GN}}{\beta}\rceil}{2}\vee\frac{\lceil2\alpha_{GN}\rceil}{2}}\ln^{\left(1+\frac{1}{\alpha_{G}}\right)}\left(\frac{\left(H(p_{0}|\nu)\right)}{\epsilon}\right)}{\epsilon^{\left(\ell_{G}+\alpha_{GN}+1\right)\left(1+\frac{1}{\alpha_{G}}\right)+\frac{1}{\alpha_{G}}}}\right)
\]
steps. If we choose $\beta\geq2\alpha_{GN}$ and $\ell_{G}=0$, then
$K\approx\tilde{O}\left(\frac{\gamma^{1+\frac{1}{\alpha_{G}}}d^{\lceil\frac{\alpha_{GN}+2}{\beta}\rceil\left(\alpha_{GN}+2\right)\left(1+\frac{1}{\alpha_{G}}\right)+\frac{\lceil2\alpha_{GN}\rceil}{2}}}{\epsilon^{\frac{\alpha_{GN}^{2}+2\alpha_{GN}+2}{\alpha_{G}}}}\right).$\end{thm}
\begin{proof} See Appendix \ref{AC4}. \end{proof} If we initialize
with a Gaussian distribution $p_{0}=N(0,\frac{1}{L}I)$, we have the
following lemma. \begin{lemma}\label{lem:C5} Suppose $\nu=e^{-U}$
is $\alpha_{G}$-mixture locally smooth. Let $p_{0}=N(0,\frac{1}{L}I)$.
Then $H(p_{0}|\nu)=O\left(d^{\frac{\ell+1+\alpha_{GN}}{2}}\right).$
\end{lemma} \begin{proof} See Appendix \ref{AF1}. \end{proof}
Therefore, Theorem \ref{thm:C4} states that to achieve $H(p_{k}|\pi)\le\epsilon$,
ULA has computation complexity $\tilde{O}\left(\frac{\gamma^{1+\frac{1}{\alpha_{G}}}d^{\lceil\frac{\ell_{G}+\alpha_{GN}+2}{\beta}\rceil\left(\ell_{G}+\alpha_{GN}+2\right)\left(1+\frac{1}{\alpha_{G}}\right)+\frac{\lceil\frac{4\ell_{G}}{\beta}\rceil}{2}+\left(\frac{\lceil\frac{\left(\ell_{G}+\alpha_{GN}\right)4\alpha_{GN}}{\beta}\rceil}{2}\vee\frac{\lceil2\alpha_{GN}\rceil}{2}\right)}}{\epsilon^{\left(\ell_{G}+\alpha_{GN}+1\right)\left(1+\frac{1}{\alpha_{G}}\right)+\frac{1}{\alpha_{G}}}}\right).$
By Pinsker's inequality, we have $TV\left(p_{k}|\nu\right)\leq\sqrt{\frac{H(p_{k}|\nu)}{2}}$
which implies that to get $TV\left(p_{k}|\pi\right)\leq\epsilon$,
it is enough to obtain $H(p_{k}|\pi)\le2\epsilon^{2}$. This bound
indicates that the number of iteration to reach $\epsilon$ accuracy
for total variation is
\[
\tilde{O}\left(\frac{\gamma^{1+\frac{1}{\alpha_{G}}}d^{\lceil\frac{\ell_{G}+\alpha_{GN}+2}{\beta}\rceil\left(\ell_{G}+\alpha_{GN}+2\right)\left(1+\frac{1}{\alpha_{G}}\right)+\frac{\lceil\frac{4\ell_{G}}{\beta}\rceil}{2}+\left(\frac{\lceil\frac{\left(\ell_{G}+\alpha_{GN}\right)4\alpha_{GN}}{\beta}\rceil}{2}\vee\frac{\lceil2\alpha_{GN}\rceil}{2}\right)}}{\epsilon^{2\left(\ell_{G}+\alpha_{GN}+1\right)\left(1+\frac{1}{\alpha_{G}}\right)+\frac{1}{\alpha_{G}}}}\right).
\]
On the other hand, from Lemma \ref{AF4} we know that $\int e^{\frac{a}{4\beta}\Vert x\Vert^{\beta}}\pi(x)dx\leq e^{\tilde{d}+\tilde{c}}<\infty$.
By \citep{bolley2005weighted}'s Corollary 2.3, we can bound Wasserstein
distance by
\begin{align*}
W_{\beta}(p_{k},\ \nu) & \leq2\left[\frac{a}{4\beta}\left(1.5+\tilde{d}+\tilde{c}\right)\right]^{\frac{1}{\beta}}\left(H(p_{k}|\nu)^{\frac{1}{\beta}}+H(p_{k}|\nu)^{\frac{1}{2\beta}}\right).
\end{align*}
To have $W_{\beta}(p_{K},\ \pi)\leq\epsilon$, it is sufficient to
choose $H(p_{k}|\nu)^{\frac{1}{2\beta}}=\tilde{O}\left(\epsilon d^{\frac{-1}{\beta}}\right)$,
which in turn implies $H(p_{k}|\nu)=\tilde{O}\left(\epsilon^{2\beta}d^{-2}\right).$
By replacing this in the bound above, we obtain the number of iteration
for $L_{\beta}$-Wasserstein distance is
\[
\tilde{O}\left(\frac{\gamma^{1+\frac{1}{\alpha_{G}}}d^{\lceil\frac{\ell_{G}+\alpha_{GN}+2}{\beta}\rceil\left(\ell_{G}+\alpha_{GN}+2\right)\left(1+\frac{1}{\alpha_{G}}\right)+\frac{\lceil\frac{4\ell_{G}}{\beta}\rceil}{2}+\left(\frac{\lceil\frac{\left(\ell_{G}+\alpha_{GN}\right)4\alpha_{GN}}{\beta}\rceil}{2}\vee\frac{\lceil2\alpha_{GN}\rceil}{2}\right)+2\left(\ell_{G}+\alpha_{GN}+1\right)\left(1+\frac{1}{\alpha_{G}}\right)+\frac{2}{\alpha_{G}}}}{\epsilon^{2\beta\left(\ell_{G}+\alpha_{GN}+1\right)\left(1+\frac{1}{\alpha_{G}}\right)+\frac{1}{\alpha_{G}}}}\right).
\]
If $\beta=2$, we have the function satisfies log Sobolev and we have the following corrolary.
\begin{corrolary}Suppose $\nu$ satisfies $\gamma-$log-Sobolev, $\alpha_{G}$-mixture
locally-smooth, for any $x_{0}\sim p_{0}$ with $H(p_{0}\vert\pi)=C_{0}<\infty$,
the iterates $x_{k}\sim p_{k}$ of ULA~ with step size $\eta\leq1\wedge\frac{1}{4\gamma}\wedge\left(\frac{\gamma}{9N^{\frac{3}{2}}L_{G}^{3}}\right)^{\frac{1}{\alpha_{G}}}$satisfies
\begin{align}
H(p_{k}\vert\pi)\le e^{-\gamma\eta k}H(p_{0}\vert\pi)+\frac{8\eta^{\alpha_{G}}D}{3\gamma},\label{Eq:Main1-2-1-1}
\end{align}
Then, for any $\epsilon>0$, to achieve $H(p_{k}\vert\pi)<\epsilon$,
it suffices to run ULA with step size $\eta\le1\wedge\frac{1}{4\gamma}\wedge\left(\frac{\gamma}{9N^{\frac{3}{2}}L_{G}^{3}}\right)^{\frac{1}{\alpha_{G}}}\wedge\left(\frac{3\epsilon\gamma}{16D}\right)^{\frac{1}{\alpha}}$for
$k\ge\frac{1}{\gamma\eta}\log\frac{2H\left(p_{0}\vert\pi\right)}{\epsilon}$
iterations. \end{corrolary}\begin{proof} See Appendix \ref{AF1}. \end{proof}
\subsection{Convergence under $\beta-$dissipative, $\alpha_{H}-$mixture locally-Hessian
smooth}

If we have the potential satisfies $\alpha_{H}-$mixture locally-Hessian
smooth instead of $\alpha_{G}-$mixture locally smooth, we obtain
\label{Sec:Review-1-1}
\begin{align*}
\left\Vert \nabla U(x)-\nabla U(y)\right\Vert  & \leq C_{2H}\left(1+\left\Vert x\right\Vert ^{\ell_{H}+\alpha_{HN}}+\left\Vert y\right\Vert ^{\ell_{H}+\alpha_{HN}}\right)\sum_{i=0}\left\Vert x-y\right\Vert ^{1+\alpha_{Hi}},
\end{align*}
where $\alpha_{H0}=0$.
Applying exactly the previous process, we obtain the following result.

\begin{thm} \label{thm:C4-1} Suppose $\nu$ is $\gamma$-Poincaré
inequality, $\beta$-dissipative, $\alpha_{H}$-mixture locally Hessian
smooth. For any $x_{0}\sim p_{0}$ with $H(p_{0}|\pi)=C_{0}<\infty$,
the iterates $x_{k}\sim p_{k}$ of ULA~ with step size $\eta$ sufficiently
small satisfying the following conditions
\[
\eta=\min\left\{ 1,\left(\frac{\epsilon}{2TD}\right)\right\} ,
\]
where $D=O\left(d^{\frac{\lceil\frac{4\left(\ell_{H}+\alpha_{HN}\right)}{\beta}\rceil}{2}+\frac{\lceil\frac{\left(\ell_{H}+2\alpha_{HN}+1\right)4\left(1+\alpha_{HN}\right)}{\beta}\rceil}{2}}\right)$.
The ULA iterates reach $\epsilon$-accuracy of the target $\nu$ in
KL divergence after
\[
K=O\left(\frac{\gamma^{2}d^{\lceil\frac{\ell_{H}+\alpha_{HN}+3}{\beta}\rceil\left(\ell_{H}+\alpha_{HN}+3\right)2+\frac{\lceil\frac{4\left(\ell_{H}+\alpha_{HN}\right)}{\beta}\rceil}{2}+\frac{\lceil\frac{\left(\ell_{H}+\alpha_{HN}+1\right)\left(4\alpha_{HN}+4\right)}{\beta}\rceil}{2}}\ln^{2}\left(\frac{\left(H(p_{0}|\nu)\right)}{\epsilon}\right)}{\epsilon^{2\left(\ell_{H}+\alpha_{HN}+2\right)+1}}\right)
\]
steps. If $\ell_{H}=0$, then $K\approx\tilde{O}\left(\frac{\gamma^{2}d^{2\lceil\frac{\alpha_{HN}+3}{\beta}\rceil\left(\alpha_{HN}+3\right)+\frac{\lceil\frac{4\alpha_{HN}}{\beta}\rceil}{2}+\frac{\lceil\frac{\left(\alpha_{HN}+1\right)\left(4\alpha_{HN}+4\right)}{\beta}\rceil}{2}}}{\epsilon^{2\alpha_{HN}+5}}\right).$\end{thm}

\begin{proof} See Appendix \ref{AC4}. \end{proof}

\subsection{Convergence $\beta-$dissipative, gradient Lipschitz and $\alpha_{H}$-mixture
Hessian locally-smooth}

\label{Sec:Review-1-1-2}
Although our main results were obtained under the smoothness assumption
on Lipschitz gradients of the potential, prior analyses of Langevin
algorithms, \citep{mou2022improved,balasubramanian2022towards} suggest that the convergence
rate improves with additional assumptions on Hessian smoothness.
\begin{lemma}\label{lem:C2-1-2-2}Suppose $\nu$ is $\alpha_{G}$-smooth\textit{,
and} $\alpha_{H}$-mixture Hessian locally-smooth\textit{, the following
bound holds for the discretization error}.

\begin{align*}
\mathbb{E}\left[\left\Vert \nabla U(x_{k,t})-\mathbb{E}\left[\nabla U(x_{k})|x_{k,t}\right]\right\Vert ^{2}\right] & \leq24L_{G}^{2}\eta^{2}I\left(p_{k,t}|\nu\right)+12d\eta^{2}L_{G}^{3}+O\left(d^{\frac{\lceil\frac{4\ell_{H}}{\beta}\rceil+\lceil\frac{\left(4\alpha_{HN}+4\right)}{\beta}\rceil}{2}}\right)\eta^{\alpha_{H}+1}.
\end{align*}

\end{lemma} \begin{proof} See Appendix \ref{AC2}. \end{proof}
.

\begin{lemma}\label{lem:C1-1-2-2}

Suppose that $\nu$ satisfies $\gamma$-Poincaré inequality, $\alpha_{G}$-smooth,
$\alpha_{H}$-mixture locally Hessian smooth, then for any distribution
$\mu$,
\[
H(\mu|\nu)\leq\left(\sqrt{2}+2L_{G}\sqrt{\frac{1}{\gamma}}\right)M_{4}^{\frac{1}{2}}\left(\mu+\nu\right)\sqrt{I\left(\mu|\nu\right)}.
\]
\end{lemma} \begin{proof} See Appendix \ref{AC1}. \end{proof}
A result that follows directly from the second observation is that:
\begin{lemma}\label{lem:C3-1-2-2} Suppose $\nu$ is $\beta$-dissipative,
$\alpha_{H}$-mixture locally-Hessian smooth. If $0<\eta=\min\left\{ 1,\left(\frac{\epsilon}{2TD}\right)\right\} $
, then along each step of ULA~\eqref{eq:2},
\begin{align}
\frac{d}{dt}H(p_{k,t}|\nu)\le-\frac{1}{2}I(p_{k,t}|\nu)+\eta^{\alpha_{H}+1}D,\label{Eq:Main1-2-1-1-1-1-2-1-1-2-2}
\end{align}
where
\[
D=O\left(d^{\frac{\lceil\frac{4\ell_{H}}{\beta}\rceil+\lceil\frac{\left(4\alpha_{HN}+4\right)}{\beta}\rceil}{2}}\right),
\]
\end{lemma} \begin{proof} See Appendix \ref{AC3}. \end{proof}
Based on both observations, we are now ready to state the main result
in this section. \begin{thm} \label{thm:C4-1-2-2} Suppose $\nu$
is $\gamma$-Poincaré inequality, $\beta$-dissipative, $\alpha_{H}$-mixture
locally Hessian smooth. For any $x_{0}\sim p_{0}$ with $H(p_{0}|\pi)=C_{0}<\infty$,
the iterates $x_{k}\sim p_{k}$ of ULA~ with step size $\eta$ sufficiently
small satisfying the following conditions
\[
\eta=\min\left\{ 1,\left(\frac{\epsilon}{2TD}\right)\right\} ,
\]
where $D=O\left(d^{\frac{\lceil\frac{4\ell_{H}}{\beta}\rceil+\lceil\frac{\left(4\alpha_{HN}+4\right)}{\beta}\rceil}{2}}\right)$.

The ULA iterates reach $\epsilon$-accuracy of the target $\nu$ in
KL divergence after
\[
K=O\left(\frac{d^{\frac{\lceil\frac{4\ell_{H}}{\beta}\rceil+\lceil\frac{\left(4\alpha_{HN}+4\right)}{\beta}\rceil}{2\left(\alpha_{H}+1\right)}+\lceil\frac{4}{\beta}\rceil\left(1+\frac{1}{\alpha_{H}+1}\right)}\ln^{\left(1+\frac{1}{\alpha_{H}+1}\right)}\left(\frac{\left(H(p_{0}|\nu)\right)}{\epsilon}\right)}{\epsilon^{\left(1+\frac{2}{\alpha_{H}+1}\right)}}\right)
\]
steps. If $\ell_{H}=0$, $\alpha_{H}=1$, then $K\approx\tilde{O}\left(\frac{d^{\frac{\lceil\frac{8}{\beta}\rceil}{4}+2\lceil\frac{4}{\beta}\rceil}}{\epsilon^{2}}\right).$\end{thm}

\begin{proof} See Appendix \ref{AC4}. \end{proof}

\subsection{Sampling via smoothing potential }
In this case, $p$-generalize Gaussian smoothing is used to compensate for the weakly smooth behavior of some distributions in the mixture, \citep{nguyen2021unadjusted}. Specifically, for some $\mu\geq0$, they consider
\[
U_{\mu}(\mathrm{y}):=\mathrm{E}_{\xi}[U(\mathrm{y}+\mu\xi)]=\frac{1}{\kappa}\int_{\mathbb{R}^{d}}U(\mathrm{y}+\mu\xi)e^{-\left\Vert \xi\right\Vert _{p}^{p}/p}\mathrm{d}\xi,
\]
where $\kappa\stackrel{_{def}}{=}\int_{\mathbb{R}^{d}}e^{-\left\Vert \xi\right\Vert _{p}^{p}/p}\mathrm{d}\xi=\frac{2^{d}\Gamma^{d}(\frac{1}{p})}{p^{d-\frac{d}{p}}}$
and $\xi\sim N_{p}(0,I_{d\times d})$ (the $p$-generalized Gaussian
distribution). A $p$-generalized Gaussian smoothing is used instead of the origin potential because $U_{\mu}$ is smooth whereas $U$ is not. Due to its ability to provide normal distributions when $p=2$, Laplace distributions when $p=1$, tails heavier or lighter than normal and even continuous uniform distributions in the limit, this distribution family is preferred over Gaussian smoothing. More significantly, it can be proved that a smoothing potential $U_{\mu}(x)$ is actually smooth in any order. This nice property is novel and useful in the sampling process, especially when the potential exhibits some sort of weakly smooth behaviors and we want to improve the order of smoothness. Here, we extend \citep{nguyen2021unadjusted}'s
$p$-generalized Gaussian smoothing by considering $p\in\mathbb{R},$
$p>1$ and some primary features of $U_{\mu}$ based on adapting those
results of \citep{nesterov2017random}.\begin{lemma} \label{lem:B1}
If potential $U:\mathbb{R}^{d}\rightarrow\mathbb{R}$ satisfies an
$\alpha$-mixture weakly smooth for some $0<\alpha=\alpha_{1}\leq...\leq\alpha_{N}\leq1$,
$i=1,..,N$ $0<L_{i}<\infty$, let $L=1\vee\max\left\{ L_{i}\right\} $
then:

(i) $\forall x\in\mathbb{R}^{d}$ : $\left|U_{\mu}(x)-U(x)\right|{\displaystyle \leq\frac{NL\mu^{1+\alpha}}{(1+\alpha)}d^{\frac{2}{2\wedge p}},}$

(ii) $\forall x\in\mathbb{R}^{d}$: ${\displaystyle \left\Vert \nabla U_{\mu}(x)-\nabla U(x)\right\Vert \leq\begin{cases}
\frac{NL\mu^{1+\alpha}}{(1+\alpha)}d^{\frac{3}{p}} & 1\leq p\leq2,\\
\frac{NL\mu^{1+\alpha}}{(1+\alpha)}d^{\frac{5}{2}} & p>2,
\end{cases}}$

(iii) $\forall x,\ y\in\mathbb{R}^{d}$: ${\displaystyle \left\Vert \nabla U_{\mu}(y)-\nabla U_{\mu}(x)\right\Vert \leq\begin{cases}
\frac{NL}{\mu^{1-\alpha}}d^{\frac{2}{p}}\left\Vert y-x\right\Vert  & 1\leq p\leq2,\\
\frac{NL}{\mu^{1-\alpha}}d^{2}\left\Vert y-x\right\Vert  & p>2.
\end{cases}}$

(iv)$\forall\mathrm{x},\ \mathrm{y}\in\mathbb{R}^{d}$: for $p>2,$${\displaystyle \left\Vert \nabla^{2}U_{\mu}(\mathrm{y})-\nabla^{2}U_{\mu}(\mathrm{x})\right\Vert _\mathrm{op}\leq\frac{NL}{\mu^{2-\alpha}}d^{4-\frac{2}{p}}\left\Vert y-x\right\Vert }$.

If $p=2,$${\displaystyle \left\Vert \nabla^{2}U_{\mu}(y)-\nabla^{2}U_{\mu}(\mathrm{x})\right\Vert \leq\frac{2NL}{\mu^{2-\alpha}}d^{2}\left\Vert y-x\right\Vert .}$
\end{lemma} \begin{proof} Due to space limitation, we provide the
proof in the Supplement.\end{proof}

Based on a result of \citep{nguyen2021unadjusted}, we study the
convergence of the discrete-time process for the smoothing potential
that have the following form:
\begin{equation}
U_{\mu}(x):=\mathrm{\mathbb{E}}_{\xi}[U(y+\mu\xi)].\label{eq:2.5b}
\end{equation}
Keep in mind that $U(\cdot)$ is $\alpha_G$-mixture
locally smooth with $\ell_G=0$ but $U_{\mu}(x)$ is smooth. In terms
of the smoothing potential $U_{\mu}$, ULA can be specified as:
\begin{equation}
x_{k+1}=x_{k}-\eta_{k}\nabla U_{\mu}(x_{k})+\sqrt{2\eta_{k}}\varsigma_{k},\label{eq:LMC}
\end{equation}
where $\varsigma_{k}\sim N(0,\ I_{d\times d})$ are independent Gaussian
random vectors. From \citep{nguyen2021unadjusted}'s Lemma 3.4, $W_{2}^{2}(\nu,\ \nu_{\mu})\leq8.24NL\mu^{1+\alpha}d^{\frac{2}{p}}E_{2},$
for any $\mu\leq\left(\frac{0.05}{NLd^{\frac{2}{p}}}\right)^{\frac{1}{1+\alpha}}$
where $E_{2}=\int\left\Vert x\right\Vert ^{2}\nu(x)dx<\infty$, $L=1\vee\max\left\{ L_{i}\right\}.$
Moreover, Poincaré is preserved under bounded perturbation, we have
$U_{\mu}$ also satisfies Poincaré inequality. As a result, we obtain
the following lemma. \begin{lemma} \label{lem:D2}Suppose that $\nu$
satisfies $\gamma$-Poincaré inequality and $\alpha_{G}$-mixture
weakly smooth. Then for any distribution $p$,
\begin{align*}
H\left(p|\pi_{\mu}\right) & \leq\left(\sqrt{2}+2\frac{NL\mu^{1+\alpha_{G}}}{(1+\alpha)}d^{\frac{2}{p}\vee2}\sqrt{\frac{1}{\gamma_{1}}}\right)M_{4}^{\frac{1}{2}}\left(p+\nu_{\mu}\right)\sqrt{I},
\end{align*}
where $\gamma_{1}=\gamma e^{-4L\mu^{1+\alpha}d^{\frac{1+\alpha}{2\wedge p}}}$.
\end{lemma} \begin{proof} See Appendix \ref{AD2}. \end{proof}
In addition, we observe that
$U_{\mu}$ is $\beta$ dissipative with constant
$\left(\frac{a}{2},b+\frac{L}{2}\mu^{\alpha_{G}}d^{\frac{5}{2}\vee\frac{3}{p}}\left(\frac{L\mu^{\alpha_{G}}d^{\frac{5}{2}\vee\frac{3}{p}}}{a}\right)^{\frac{1}{\beta-1}}\right).$
With all of these properties, Theorem \ref{thm:C4} is applicable
to sampling from $U_{\mu}.$ However, in general, we do not have access
to $\nabla U_{\mu}(x)$, but an unbiased estimate of it:
\begin{align}
g_{\mu}(x,\xi)=\nabla U(x+\mu\xi)\label{eq:gradest}
\end{align}
where $\xi\sim N_{p}(0,I_{d})$. \citep{nguyen2021unadjusted}'s Lemma
3.3 states that the variance of the estimate can be bounded.\begin{lemma}
\label{lem:D3} For any $x_{k}\in\mathbb{R}^{d}$, $g_{\mu}(x_{k},\xi)$
is an unbiased estimator of $\nabla U_{\mu}$ such that
\begin{align*}
\mathrm{Var}\left[g_{\mu}(x_{k},\xi)\right]\leq4N^{2}L^{2}\mu^{2\alpha_{G}}d^{\frac{2\alpha_{G}}{p}}.
\end{align*}
\end{lemma}
Let $x_{\mu,k}$ be the interpolation of the discretized process \eqref{eq:LMC}
and let $p_{\mu,k}$ denote its distribution, $\mathbb{E}_{p_{k}}\left[\left\Vert \nabla U(x_{k})\right\Vert ^{2\alpha_{N}}\right]$
can similarly be upper bounded by the following lemma.With stochastic
approximation of the gradient of the smoothing potential, we have
the following bound. \begin{lemma}\label{lem:D5} Suppose $\pi$
is $\beta$-dissipative, $\alpha_{G}$-mixture weakly smooth. If $0<\eta\le\min\left\{ 1,\left(\frac{\epsilon}{2TD_{\mu}}\right)^{\frac{1}{\alpha_{G}}}\right\} $,
then along each step of ULA~\eqref{eq:2},
\begin{align}
\frac{d}{dt}H(p_{\mu,k,t}|\nu_{\mu})\le-\frac{3}{4}I(p_{\mu,k,t}|\nu_{\mu})+\eta^{\alpha_{G}}D_{\mu},\label{Eq:Main1-2-1-1-1-1-2-4-1}
\end{align}
where $D_{\mu}=O\left(d^{\lceil\frac{2\alpha_{GN}^{2}}{\beta}\rceil}\right).$
\end{lemma} \begin{proof} See Appendix \ref{AD5}. \end{proof}
Another result is stated in the subsequent theorem. \begin{thm}
\label{thm:D6}Suppose $\pi$ is $\gamma$-Poincaré inequality, $\beta$-dissipative,
$\alpha_{G}$-mixture weakly smooth. For any $x_{0}\sim p_{0}$ with
$H(p_{0}|\pi)=C_{0}<\infty$, the iterates $x_{k}\sim p_{k}$ of ULA~
with step size $\eta$ sufficiently small satisfying the following
conditions
\[
\eta=\min\left\{ 1,\left(\frac{\epsilon}{2TD_{\mu}}\right)^{\frac{1}{\alpha_{G}}}\right\} ,
\]
where $D_{\mu}$ defined as above. For any even integer $k>4$, the
ULA iterates reach $\epsilon$-accuracy of the target $\nu$ in
\[
K=O\left(\frac{\gamma^{1+\frac{1}{\alpha_{G}}}d^{\lceil\frac{2\alpha_{GN}^{2}}{\beta}\rceil\frac{1}{\alpha_{G}}+\lceil\frac{\alpha_{GN}+2}{\beta}\rceil\left(\alpha_{GN}+2\right)\left(1+\frac{1}{\alpha_{G}}\right)}\ln^{\left(1+\frac{1}{\alpha_{G}}\right)}\left(\frac{\left(H(p_{0}|\nu)\right)}{\epsilon}\right)}{\epsilon^{\left(\alpha_{GN}+1\right)\left(1+\frac{1}{\alpha_{G}}\right)+\frac{1}{\alpha_{G}}}}\right)
\]
steps. If we choose $\eta$ small enough then for any $\epsilon>0$,
to achieve $W_{\beta}(p_{K},\nu)<\epsilon$, it suffices to run ULA
with step size
\[
\eta=\min\left\{ 1,\left(\frac{\epsilon}{2TD_{\mu}}\right)^{\frac{1}{\alpha_{G}}},\left(\frac{\epsilon}{9\sqrt{NLE_{2}}d^{\frac{1}{p}}}\right)^{\frac{2}{\alpha_{G}}}\right\} ,
\]
for
\[
K\approx\tilde{O}\left(\frac{{\displaystyle d^{\frac{2}{\beta}\left(\lceil\frac{2\alpha_{GN}^{2}}{\beta}\rceil\frac{1}{\alpha_{G}}+\lceil\frac{\alpha_{GN}+2}{\beta}\rceil\left(\alpha_{GN}+2\right)\left(1+\frac{1}{\alpha_{G}}\right)\right)+2+\frac{4}{\alpha_{G}}}}}{\gamma_{1}^{\left(1+\frac{1}{\alpha_{G}}\right)}\epsilon^{\left(\alpha_{GN}+1\right)\left(1+\frac{1}{\alpha_{G}}\right)+\frac{1}{\alpha_{G}}}}\right),
\]
iterations.\end{thm} \begin{proof} See Appendix \ref{AD6}. \end{proof}

\section{Extended result \label{sec:Experiment-1}}


\subsection{ULA convergence under non-strongly convex outside the ball, $\alpha$-mixture
weakly smooth and $\beta-$dissipativity}

Since Poincar\'e inequalities are preserved under bounded perturbations
by \citep{holley1986logarithmic}'s theorem, we provide our extended
results through convexification of non-convex domain \citep{ma2019sampling,yan2012extension}.
Convexification of non-convex domain is an original approach proposed
by \citep{ma2019sampling,yan2012extension}, developed and apply to
strongly convex outside a compact set by \citep{ma2019sampling}.
Adapted techniques from \citep{ma2019sampling} for non-strongly convex
and $\alpha_G$-mixture weakly smooth potentials, \citep{nguyen2021unadjusted}
derive a tighter bound for the difference between constructed convex
potential and the original one. Using this result, we obtain the following
lemma.\begin{lemma} Suppose $\nu$ is non-strongly convex outside
the ball of radius $R$, $\alpha_G$-mixture weakly smooth and $\beta-$dissipativity,
there exists $\breve{U}\in C^{1}(\mathbb{R}^{d})$ with a Hessian
that exists everywhere on $\mathbb{R}^{d}$, and $\breve{U}$ is convex
on $\mathbb{R}^{d}$ such that
\begin{equation}
\sup\left(\breve{U}(\ x)-U(\ x)\right)-\inf\left(\breve{U}(\ x)-U(\ x)\right)\leq\sum_{i}L_{i}R^{1+\alpha_{Gi}}.
\end{equation}
\label{lem:D4} \end{lemma} \begin{proof} It comes directly from
\citep{nguyen2021unadjusted} Lemma 4.2. \end{proof}Based on it,
we get the following result. \begin{thm} \label{thm:E3} Suppose
$\nu$ is non-strongly convex outside the ball $\mathbb{B}(0,R)$,
$\beta$-dissipative, $\alpha_{G}$-mixture weakly smooth. For any
$x_{0}\sim p_{0}$ $ $ with $H(p_{0}|\nu)=C_{0}<\infty$, the iterates
$x_{k}\sim p_{k}$ of ULA~ with step size $\eta$ sufficiently small
satisfying the following conditions
\[
\eta=\min\left\{ 1,\left(\frac{\epsilon}{2TD_{\mu}}\right)^{\frac{1}{\alpha_{G}}}\right\} ,
\]
where $D_{\mu}$ defined as above. The ULA iterates reach $\epsilon$-accuracy
of the target $\nu$, after
\[
K\approx\tilde{O}\left(\frac{\left(32C_{K}^{2}d\left(\frac{a+b+2aR^{2}+3}{a}\right)e^{4\left(2\sum_{i}L_{i}R^{1+\alpha_{i}}\right)}\right){}^{1+\frac{1}{\alpha_{G}}}d^{\lceil\frac{\alpha_{GN}+2}{\beta}\rceil\left(\alpha_{GN}+2\right)\left(1+\frac{1}{\alpha_{G}}\right)+\frac{\lceil2\alpha_{GN}\rceil}{2}}}{\epsilon^{\frac{\alpha_{GN}^{2}+2\alpha_{GN}+2}{\alpha_{G}}}}\right)
\]
steps where $C_{k}$ is a universal constant. If we choose $\eta$
small enough then, for any $\epsilon>0$, to achieve $H(p_{k}|\nu)<\epsilon$,
it suffices to run ULA with step size
\[
\eta=\min\left\{ 1,\left(\frac{\epsilon}{2TD_{\mu}}\right)^{\frac{1}{\alpha_{G}}}\right\} ,
\]
for
\[
K\approx\tilde{O}\left(\frac{\left(32C_{K}^{2}d\left(\frac{a+b+2aR^{2}+3}{a}\right)e^{4\left(2\sum_{i}L_{i}R^{1+\alpha_{i}}\right)}\right){}^{1+\frac{1}{\alpha_{G}}}d^{\lceil\frac{\alpha_{GN}+2}{\beta}\rceil\left(\alpha_{GN}+2\right)\left(1+\frac{1}{\alpha_{G}}\right)+\frac{\lceil2\alpha_{GN}\rceil}{2}}}{\epsilon^{\frac{\alpha_{GN}^{2}+2\alpha_{GN}+2}{\alpha_{G}}}}\right)
\]
iterations.\end{thm} \begin{proof} See Appendix \ref{AE1}.\end{proof}



\section{Applications}
We employ the outcomes of Sections \ref{sec:Experiment} and \ref{sec:Experiment-1} to a few of illustrative potential functions in this section. To the best of our knowledge,
these results can not be obtained by any of these previous work.

\begin{example}
 (\textbf{$\alpha_{G}$-}mixture locally smooth potential
with lighter tails). Let us analyze the potential function $U(x)=\sum_{i}^{N}L_{i}\left\Vert x\right\Vert ^{\alpha_{i}}$
for $2<\alpha\leq\alpha_{i}\in(2,3]$, $L_{i}>0$. Since $\nabla U(x)=\sum_{i}L_{i}\alpha_{i}x\left\Vert x\right\Vert ^{\alpha_{i}-2}$,
by triangle inequality
\begin{align*}
\left\Vert \nabla U(x)-\nabla U(y)\right\Vert  & \leq\sum_{i}L_{i}\alpha_{i}\left\Vert x\left\Vert x\right\Vert ^{\alpha_{i}-2}-y\left\Vert y\right\Vert ^{\alpha_{i}-2}\right\Vert \\
 & \stackrel{_{1}}{\leq}8\sum_{i}L_{i}\alpha_{i}\left\Vert x-y\right\Vert ^{\frac{\alpha_{i}-1}{3}}\left(1+\left\Vert x\right\Vert ^{\frac{2\left(\alpha_{N}-1\right)}{3}}+\left\Vert y\right\Vert ^{\frac{2\left(\alpha_{N}-1\right)}{3}}\right),
\end{align*}
where 1 is the result of Lemma \ref{lem:F4} below. This indicates
that the potential $U(x)$ is $\left(\frac{\alpha_{i}-1}{3}\right)$-mixture
locally smooth. In addition, we have
\begin{align*}
\left\langle \nabla U(x),\ x\right\rangle  & =\left\langle \sum_{i}L_{i}\alpha_{i}x\left\Vert x\right\Vert ^{\alpha_{i}-2},x\right\rangle \\
 & \geq a\left\Vert x\right\Vert ^{\alpha_{N}}-0,
\end{align*}
which implies $U(x)$ is $\alpha_{N}$-dissipative. In order to apply
the mixture of tail condition, we need a specific Poincar\'e constant
$\gamma$ from our assumption. As a result, we can use Theorem \ref{thm:C4}
to get $\epsilon$-precision in KL-divergence in $K\approx\tilde{O}\left(\frac{\gamma^{1+\frac{1}{\alpha_{G}}}d^{\frac{2\left(10\alpha_{N}+8\right)}{3}\left(1+\frac{1}{\alpha_{G}}\right)+2+4\alpha_{N}}}{\epsilon^{\left(5\alpha_{N}+1\right)\left(1+\frac{1}{\alpha_{G}}\right)+\frac{1}{\alpha_{G}}}}\right)$
steps. In general, this bound is weaker compared to previous single
tail growth results but it is applicable for larger range of mixture
distributions. If we apply Theorem \ref{thm:D6}, we can obtain $\epsilon$
precision in $L_{\alpha_{N}}$-Wasserstein distance after taking
\[
K\approx\tilde{O}\left(\frac{{\displaystyle d^{\frac{2}{\beta}\left(1+\frac{1}{\alpha}\right)+2\vee\frac{3\alpha_{N}}{p}+2+\frac{4}{\alpha}}}}{\epsilon^{2\alpha_{N}\left(1+\frac{2}{\alpha}\right)}}\right).
\]
\end{example}
\begin{example} ($\alpha_{H}$-mixture locally Hessian smooth potential
with lighter tails). Let us analyze the potential function $U(x)=\sum_{i}^{N}L_{i}\left\Vert x\right\Vert ^{\alpha_{i}}$
for $2<\alpha\leq\alpha_{i}\in(2,3]$, $L_{i}>0$. Since $\nabla U(x)=\sum_{i}L_{i}\alpha_{i}x\left\Vert x\right\Vert ^{\alpha_{i}-2}$,
by triangle inequality
\begin{align*}
\left\Vert \nabla^{2}U(x)-\nabla U^{2}(y)\right\Vert  & \leq\sum_{i}L_{i}\alpha_{i}\left\Vert x\left\Vert x\right\Vert ^{\alpha_{i}-2}-y\left\Vert y\right\Vert ^{\alpha_{i}-2}\right\Vert \\
 & \stackrel{_{1}}{\leq}\sum_{i}L_{i}\left(\alpha_{i}-1+\left(\alpha_{i}-2\right)2^{6-\alpha_{i}}\right)\left\Vert x-y\right\Vert ^{\alpha_{i}-2},
\end{align*}
where 1 is the result of Lemma \ref{lem:F4} below. This indicates
that the potential $U(x)$ is $\left(\frac{\alpha_{i}-2}{3}\right)$-mixture
locally Hessian smooth. In addition, we have
\begin{align*}
\left\langle \nabla U(x),\ x\right\rangle  & =\left\langle \sum_{i}L_{i}\alpha_{i}x\left\Vert x\right\Vert ^{\alpha_{i}-2},x\right\rangle \\
 & \geq a\left\Vert x\right\Vert ^{\alpha_{N}}-0,
\end{align*}
which implies $U(x)$ is $\alpha_{N}$-dissipative. In order to apply
the mixture of tail condition, we need a specific Poincar\'e constant
$\gamma$ from our assumption. As a result, we can use Theorem \ref{thm:C4}
to get $\epsilon$-precision in KL-divergence in $K\approx\tilde{O}\left(\frac{\gamma^{2}d^{\left(6\alpha_{N}+19\right)}}{\epsilon^{2\alpha_{N}+5}}\right)$
steps. In general, this bound is weaker compared to previous single
tail growth results but it is applicable for larger range of mixture
distributions. If we apply Theorem \ref{thm:D6}, we can obtain $\epsilon$
precision in $L_{\alpha_{N}}$-Wasserstein distance after taking
\[
K\approx\tilde{O}\left(\frac{{\displaystyle d^{\frac{2}{\beta}\left(1+\frac{1}{\alpha}\right)+2\vee\frac{3\alpha_{N}}{p}+2+\frac{4}{\alpha}}}}{\epsilon^{2\alpha_{N}\left(1+\frac{2}{\alpha}\right)}}\right).
\]
\end{example}
\begin{example}
(Mixture of smooth potential with linear tails): We consider
$U(x)=\sum_{i}(1+\left\Vert x\right\Vert ^{1+\alpha_{i}})^{\frac{1}{1+\alpha_{i}}}$
where $1\leq\alpha\leq\alpha_{1}\ldots\leq\alpha_{N}$. Calculating
its gradient we have
\[
\nabla U(x)=\sum_{i}(1+\left\Vert x\right\Vert ^{1+\alpha_{i}})^{\frac{-\alpha_{i}}{1+\alpha_{i}}}\left\Vert x\right\Vert ^{\alpha_{i}-1}x.
\]
Therefore,
\begin{align*}
\left\langle \nabla U(x),\ x\right\rangle  & =\left\langle \sum_{i}(1+\left\Vert x\right\Vert ^{1+\alpha_{i}})^{\frac{-\alpha_{i}}{1+\alpha_{i}}}\left\Vert x\right\Vert ^{\alpha_{i}-1}x,x\right\rangle \\
 & \geq\sum_{i}(1+\left\Vert x\right\Vert ^{1+\alpha_{i}})^{\frac{-\alpha_{i}}{1+\alpha_{i}}}\left\Vert x\right\Vert ^{1+\alpha_{i}}\\
 & \geq\sum_{i}(1+\left\Vert x\right\Vert ^{1+\alpha_{i}})^{\frac{1}{1+\alpha_{i}}}-(1+\left\Vert x\right\Vert ^{1+\alpha_{i}})^{\frac{-\alpha_{i}}{1+\alpha_{i}}}\\
 & \geq\left\Vert x\right\Vert -1
\end{align*}
which suggests that $U(x)$ is $1$-dissipative. On the other hand,
the Hessian of this potential can be calculated as
\begin{align*}
\nabla^{2}U(x) & =\sum_{i}(1+\left\Vert x\right\Vert ^{1+\alpha_{i}})^{\frac{-\alpha_{i}}{1+\alpha_{i}}}\left\Vert x\right\Vert ^{\alpha_{i}-1}I_{d}-\alpha_{i}(1+\left\Vert x\right\Vert ^{1+\alpha_{i}})^{\frac{-\alpha_{i}}{1+\alpha_{i}}-1}\left\Vert x\right\Vert ^{2\alpha_{i}-2}xx^{\mathrm{T}}\\
 & +(\alpha_{i}-1)(1+\left\Vert x\right\Vert ^{1+\alpha_{i}})^{\frac{-\alpha_{i}}{1+\alpha_{i}}}\left\Vert x\right\Vert ^{\alpha_{i}-3}xx^{\mathrm{T}}\\
 & =\sum_{i}(1+\left\Vert x\right\Vert ^{1+\alpha_{i}})^{\frac{-\alpha_{i}}{1+\alpha_{i}}}\left\Vert x\right\Vert ^{\alpha_{i}-3}(\left\Vert x\right\Vert ^{2}I_{d}-xx^{T})+\alpha_{i}(1+\left\Vert x\right\Vert ^{1+\alpha_{i}})^{\frac{-\alpha_{i}}{1+\alpha_{i}}-1}\left\Vert x\right\Vert ^{\alpha_{i}-3}xx^{\mathrm{T}}\\
 & =\sum_{i}(1+\left\Vert x\right\Vert ^{1+\alpha_{i}})^{\frac{-1}{1+\alpha_{i}}}(1+\left\Vert x\right\Vert ^{1+\alpha_{i}})^{\frac{-(\alpha_{i}-1)}{1+\alpha_{i}}}\left\Vert x\right\Vert ^{\alpha_{i}-1}\left(\left\Vert x\right\Vert ^{-2}(\left\Vert x\right\Vert ^{2}I_{d}-xx^{T})\right)\\
 & +\alpha_{i}\sum_{i}(1+\left\Vert x\right\Vert ^{1+\alpha_{i}})^{\frac{-\alpha_{i}}{1+\alpha_{i}}-1}\left\Vert x\right\Vert ^{\alpha_{i}-1}\left(\left\Vert x\right\Vert ^{-2}xx^{\mathrm{T}}\right).
\end{align*}
Since $1\leq\alpha_{i}$, each component is bounded, which implies
its Hessian is bounded, therefore, it satisfies $1$-Holder continuous.
Additionally, the norm of gradient is bounded by,
\begin{align*}
\left\Vert \nabla U(x)\right\Vert  & =\left\Vert \sum_{i}(1+\left\Vert x\right\Vert ^{1+\alpha_{i}})^{\frac{-\alpha_{i}}{1+\alpha_{i}}}\left\Vert x\right\Vert ^{\alpha_{i}-1}x\right\Vert \\
 & \leq\sum_{i}(1+\left\Vert x\right\Vert ^{1+\alpha_{i}})^{\frac{-\alpha_{i}}{1+\alpha_{i}}}\left\Vert x\right\Vert ^{\alpha_{i}}\\
 & \leq N,
\end{align*}
which implies the potential is $1$-smooth and $1$-mixture locally
Hessian smooth with $\ell_{H}=0$. Applying to our Theorem \ref{thm:C4},
we achieve the convergence rate of $K\approx\tilde{O}\left(\frac{d^{10}}{\epsilon^{2}}\right)$
in KL-divergence.
\end{example}
\section{Conclusion\label{sec:conclusion}}

In this paper, we develop polynomial-dimension theoretical justifications of unadjusted Langevin Monte Carlo algorithm for a class of $\alpha_G$ mixture locally smooth potentials that satisfy weak dissipative inequality. In addition, we also study the class of potentials which are $\alpha_H$ mixture locally Hessian smooth. Convergence results improve when a potential is $\alpha_{G}$-smooth and $\alpha_{H}$-mixture locally Hessian smooth. By convexifying non-convex domains, we get the result for non-strongly convex outside the ball of radius $R$. We provide some nice theoretical properties of $p$-generalized Gaussian smoothing and prove the result for stochastic gradients in a very general setting. For the smoothing potential, we also provide convergence in the $L_{\beta}$-Wasserstein distance. Poincar\'{e}  inequality can be easily weakened, while computational complexity remains polynomial of $d$ dimension. It is rather straightforward to generalize our condition to $\beta>0$, which is typically satisfied by weak Poincar\'{e} inequality. An interesting application of this approach would be to sample from higher order LMC or to integrate it into a derivative-free LMC algorithm.


\appendix

\section{Measure definitions and isoperimetry \label{App0}}
Let $p,\pi$ be probability distributions on $\mathbb{R}^{d}$ with
full support and smooth densities, define the Kullback-Leibler (KL)
divergence of $p$ with respect to $\pi$ as
\begin{equation}
H(p|\pi)\stackrel{\triangle}{=}\int_{\mathbb{R}^{d}}p(x)\log\frac{p(x)}{\pi(x)}\,dx.
\end{equation}
\label{Eq:Hnu-1} Likewise, we denote the Renyi divergence of order
$q>1$ of a distribution $p$ with respect to $\pi$ as
\[
R_{q}(p|\pi)=\frac{1}{q-1}\log\int_{\mathbb{R}^{d}}\frac{p(x)^{q}}{\pi(x)^{q-1}}\,dx
\]
and for $\mathcal{B}(\mathbb{R}^{d})$ denotes the Borel $\sigma$-field
of $\mathbb{R}^{d}$, define the relative Fisher information and total
variation metrics correspondingly as
\begin{equation}
{\displaystyle I(p|\pi)\stackrel{\triangle}{=}\int_{\mathbb{R}^{d}}p(x)\Vert\nabla\log\frac{p(x)}{\pi(x)}\Vert^{2}dx},
\end{equation}
\begin{equation}
{\displaystyle TV(p,{\displaystyle \ \pi)\stackrel{\triangle}{=}\sup_{A\in\mathcal{B}(\mathbb{R}^{d})}|\int_{A}p(x)dx-\int_{A}\pi(x)dx|}.}
\end{equation}
Furthermore, we define a transference plan $\zeta$, a distribution
on $(\mathbb{R}^{d}\times\mathbb{R}^{d},\ \mathcal{B}(\mathbb{R}^{d}\times\mathbb{R}^{d}))$
(where $\mathcal{B}(\mathbb{R}^{d}\times\mathbb{R}^{d})$ is the Borel
$\sigma$-field of ($\mathbb{R}^{d}\times\mathbb{R}^{d}$)) so that
$\zeta(A\times\mathbb{R}^{d})=p(A)$ and $\zeta(\mathbb{R}^{d}\times A)=\pi(A)$
for any $A\in\mathcal{B}(\mathbb{R}^{d})$. Let $\Gamma(P,\ Q)$ designate
the set of all such transference plans. Then for $\beta>0$, the $L_{\beta}$-Wasserstein
distance is formulated as:
\begin{equation}
W_{\beta}(p,\pi)\stackrel{\triangle}{=}\left(\inf_{\zeta\in\Gamma(P,Q)}\int_{x,y\in\mathbb{R}^{d}}\Vert x-y\Vert^{\beta}\mathrm{d}\zeta(x,\ y)\right)^{1/\beta}.
\end{equation}
\section{Proofs under Poincar\'{e} inequality \label{AppB}}
\subsection{Proof of $\alpha_{G}$-mixture locally-smooth property\label{Asmooth}}
\begin{lemma} If potential $U:\mathbb{R}^{d}\rightarrow\mathbb{R}$
satisfies $\alpha_{G}$-mixture locally-smooth then:
\[
U(y)\leq U(x)+\left\langle \nabla U(x),\ y-x\right\rangle +\frac{2L_{G}}{1+\alpha_{G}}\left(1+\left\Vert x\right\Vert ^{\ell_{G}}+\left\Vert y\right\Vert ^{\ell_{G}}\right)\sum_{i}\left\Vert x-y\right\Vert ^{1+\alpha_{Gi}}.
\]
\end{lemma}
\begin{proof}
We have
\begin{align*}
 & \left\Vert U(x)-U(y)-\langle\nabla U(y),x-y\rangle\right\Vert \\
= & \Big\vert\int_{0}^{1}\langle\nabla U(y+t(x-y)),x-y\rangle dt-\langle\nabla U(y),x-y\rangle\Big\vert\\
= & \Big\vert\int_{0}^{1}\langle\nabla U(y+t(x-y))-\nabla U(y),x-y\rangle dt\Big\vert.\\
\leq & \int_{0}^{1}\left\Vert \nabla U(y+t(x-y))-\nabla U(y)\right\Vert \left\Vert x-y\right\Vert dt\\
\leq & \int_{0}^{1}\left(1+\left\Vert tx+(1-t)y\right\Vert ^{\ell_{G}}+\left\Vert y\right\Vert ^{\ell_{G}}\right)\sum_{i=1}^{N}L_{Gi}t^{\alpha_{Gi}}\left\Vert x-y\right\Vert ^{\alpha_{Gi}}\left\Vert x-y\right\Vert dt\\
\leq & \sum_{i}\left(\frac{2L_{Gi}}{1+\alpha_{Gi}}\right)\left(1+\left\Vert x\right\Vert ^{\ell_{G}}+\left\Vert y\right\Vert ^{\ell_{G}}\right)\left\Vert x-y\right\Vert ^{1+\alpha_{Gi}}\\
\leq & \frac{2L_{G}}{1+\alpha_{G}}\left(1+\left\Vert x\right\Vert ^{\ell_{G}}+\left\Vert y\right\Vert ^{\ell_{G}}\right)\sum_{i}\left\Vert x-y\right\Vert ^{1+\alpha_{Gi}},
\end{align*}
where the first line comes from Taylor expansion, the third line follows
from Cauchy-Schwarz inequality and the fourth line is due to Assumption~\ref{A0}.
This gives us the desired result.
\end{proof}
\begin{lemma}Suppose $\pi=e^{-U}$ satisfies $\alpha$-mixture weakly
smooth. Let $p_{0}=N(0,\frac{1}{L}I)$. Then $H(p_{0}\vert\pi)\le U(0)-\frac{d}{2}\log\frac{2\Pi e}{L}+\sum_{i}\frac{L_{i}}{1+\alpha_{i}}\left(\frac{d}{L}\right)^{\frac{1+\alpha_{i}}{2}}=O(d).$
\end{lemma}
\begin{proof}
\label{BInitial-1}Since $U$ is $\alpha_{G}$-mixture locally-smooth,
for all $x\in\mathbb{R}^{d}$ we have
\begin{align*}
U(x) & \le U(0)+\langle\nabla U(0),x\rangle+\frac{2L_{G}}{1+\alpha_{G}}\left(1+\left\Vert x\right\Vert ^{\ell_{G}}\right)\sum_{i}\left\Vert x\right\Vert ^{1+\alpha_{Gi}}\\
 & =U(0)+\frac{2L_{G}}{1+\alpha_{G}}\left(1+\left\Vert x\right\Vert ^{\ell_{G}}\right)\sum_{i}\left\Vert x\right\Vert ^{1+\alpha_{Gi}}.
\end{align*}
Let $X\sim p_{0}=N(0,\frac{1}{L}I)$. Then
\begin{align*}
\mathbb{E}_{p_{0}}\left[U(X)\right] & \le U(0)+\mathbb{E}_{p_{0}}\left(\frac{2L_{G}}{1+\alpha_{G}}\left(1+\left\Vert x\right\Vert ^{\ell_{G}}\right)\sum_{i}\left\Vert x\right\Vert ^{1+\alpha_{Gi}}\right)\\
 & \leq U(0)+\sum_{i}\frac{2L_{G}}{1+\alpha_{G}}\mathbb{E}_{\rho}\left(\left\Vert x\right\Vert ^{2}\right)^{\frac{1+\alpha_{Gi}}{2}}+\sum_{i}\frac{2L_{G}}{1+\alpha_{G}}\mathbb{E}_{\rho}\left(\left\Vert x\right\Vert ^{\ell_{G}+1+\alpha_{Gi}}\right)\\
 & \leq U(0)+O(d)+O\left(\left(d+\ell_{G}+1+\alpha_{GN}\right)^{\frac{\ell_{G}+1+\alpha_{GN}}{2}}\right).
\end{align*}
Recall the entropy of $p_{0}$ is $H(p_{0})=-\mathbb{E}_{p_{0}}[\log p_{0}(X)]=\frac{d}{2}\log\frac{2\Pi e}{L}$.
Therefore, for $\ell_{G}$ is relative small compare to $d$, the
KL divergence is
\begin{align*}
\mathbb{E}(p_{0}\vert\nu) & =\int p_{0}\left(\log p_{0}+U\right)dx\\
 & =-H(p_{0})+\mathbb{E}_{p_{0}}[U]\\
 & \le U(0)-\frac{d}{2}\log\frac{2\Pi e}{L}+O(d)+O\left(\left(d+\ell_{G}+1+\alpha_{GN}\right)^{\frac{\ell_{G}+1+\alpha_{GN}}{2}}\right)\\
 & =O\left(d^{\frac{\ell_{G}+1+\alpha_{GN}}{2}}\right).
\end{align*}
This is the desired result.
\end{proof}

\subsection{Proof of Lemma \ref{lem:C2}\label{AC2}}
First, we preface the proof by a lemma. \begin{lemma}Let $M_{e,\beta}\left(p_{t}\right)=E_{p_{t}}\left[e^{c\left(1+\left\Vert x\right\Vert ^{2}\right)^{\frac{\beta}{2}}}\right]$
\[
M_{e,\beta}\left(p_{t}\right)\leq e^{-\frac{a^{2}}{8}t}M_{e,\beta}\left(p_{0}\right)+\frac{4\left(d+2a+b\right)}{a}e^{\frac{2d+3a+2b}{\beta}}
\]
and if we initialize $X_{0}$with a Gaussian distribution $p_{0}=N(0,\frac{1}{L}I),$
for any $n,k>0$ and $r\in\mathbb{R},$is relative small compare to
$d$
\[
E_{p_{k}}\left[\left\Vert x\right\Vert ^{n\beta}\right]=\tilde{O}\left(n^{n}d^{n}\right).
\]
\[
E_{p_{k}}\left[\left\Vert x_{k}\right\Vert ^{r}\right]=\tilde{O}\left(d^{\lceil\frac{r}{\beta}\rceil}\right).
\]
\end{lemma}\begin{proof}
For any $x\in\mathbb{R}^{d}$let $g_{\beta}(x)=e^{c\left(1+\left\Vert x\right\Vert ^{2}\right)^{\frac{\beta}{2}}}$.
First, let $c=\frac{a}{j\beta}$ for some $j\in N,j\geq2$, we have
\begin{align*}
\nabla g_{\beta}(x) & =\frac{a}{j}e^{c\left(1+\left\Vert x\right\Vert ^{2}\right)^{\frac{\beta}{2}}}\left(1+\left\Vert x\right\Vert ^{2}\right)^{\frac{\beta}{2}-1}x,\\
\nabla^{2}g_{\beta}(x) & =\frac{a}{j}e^{c\left(1+\left\Vert x\right\Vert ^{2}\right)^{\frac{\beta}{2}}}\left(\frac{a}{j}\left(1+\left\Vert x\right\Vert ^{2}\right)^{\beta-2}xx^{\mathrm{T}}+\left(1+\left\Vert x\right\Vert ^{2}\right)^{\frac{\beta}{2}-1}I+\left(\beta-2\right)\left(1+\left\Vert x\right\Vert ^{2}\right)^{\frac{\beta}{2}-2}xx^{\mathrm{T}}\right),\\
\triangle g_{\beta}(x) & =\left(\frac{a}{j}e^{c\left(1+\left\Vert x\right\Vert ^{2}\right)^{\frac{\beta}{2}}}\left(\frac{a}{j}\left(1+\left\Vert x\right\Vert ^{2}\right)^{\beta-2}\left\Vert x\right\Vert ^{2}+\left(1+\left\Vert x\right\Vert ^{2}\right)^{\frac{\beta}{2}-1}d+\left(\beta-2\right)\left(1+\left\Vert x\right\Vert ^{2}\right)^{\frac{\beta}{2}-2}\left\Vert x\right\Vert ^{2}\right)\right)\\
 & \leq\frac{a}{j}e^{c\left(1+\left\Vert x\right\Vert ^{2}\right)^{\frac{\beta}{2}}}\left(\frac{a}{j}\left(1+\left\Vert x\right\Vert ^{2}\right)^{\beta-1}+d\left(1+\left\Vert x\right\Vert ^{2}\right)^{\frac{\beta}{2}-1}\right).
\end{align*}
From these equations, we deduce:
\begin{align*}
\frac{d}{dt}M_{e,\beta}\left(p_{t}\right) & =\int p_{t}(x)\left(\triangle g_{\beta}(x)-\left\langle \nabla U\left(x\right),\nabla g_{\beta}(x)\right\rangle \right)dx\\
 & \leq\frac{a}{j}\int p_{t}(x)e^{c\left(1+\left\Vert x\right\Vert ^{2}\right)^{\frac{\beta}{2}}}\left[\left(\frac{a}{j}\left(1+\left\Vert x\right\Vert ^{2}\right)^{\beta-1}+d\left(1+\left\Vert x\right\Vert ^{2}\right)^{\frac{\beta}{2}-1}\right)-\left(1+\left\Vert x\right\Vert ^{2}\right)^{\frac{\beta}{2}-1}\left\langle \nabla U\left(x\right),x\right\rangle \right]dx\\
 & \leq\frac{a}{j}\int p_{t}(x)e^{c\left(1+\left\Vert x\right\Vert ^{2}\right)^{\frac{\beta}{2}}}\left(1+\left\Vert x\right\Vert ^{2}\right)^{\frac{\beta}{2}-1}\left[\left(\frac{a}{j}\left(1+\left\Vert x\right\Vert ^{2}\right)^{\frac{\beta}{2}}+d\right)-a\left\Vert x\right\Vert ^{\beta}+b\right]dx.
\end{align*}
Since $1\geq\frac{\beta}{2}\geq0\geq\frac{\beta}{2}-1$, for $\left\Vert x\right\Vert \geq R=\left(4\frac{d+a+b}{a}\right)^{\frac{1}{\beta}}$,
we have:
\begin{align*}
\frac{d}{dt}M_{e,\beta}\left(p_{t}\right) & \leq\frac{a}{j}\int p_{t}(x)e^{c\left(1+\left\Vert x\right\Vert ^{2}\right)^{\frac{\beta}{2}}}\left(1+\left\Vert x\right\Vert ^{2}\right)^{\frac{\beta}{2}-1}\left(-\frac{a}{2}\left(1+\left\Vert x\right\Vert ^{2}\right)^{\frac{\beta}{2}}+d+a+b\right)dx\\
 & \leq-\frac{a^{2}}{4j}\int p_{t}(x)e^{c\left(1+\left\Vert x\right\Vert ^{2}\right)^{\frac{\beta}{2}}}\left(1+\left\Vert x\right\Vert ^{2}\right)^{\beta-1}dx\\
 & \leq-\frac{a^{2}}{4j}M_{e,\beta}\left(p_{t}\right).
\end{align*}
If $\left\Vert x\right\Vert <\left(4\frac{d+a+b}{a}\right)^{\frac{1}{\beta}}$,
we have $M_{e,\beta}\left(p_{t}\right)\leq e^{c\left(1+R^{2}\right)^{\frac{\beta}{2}}}\leq e^{c+4c\frac{d+a+b}{a}}\leq e^{\frac{4d+5a+4b}{j\beta}}$
and
\begin{align*}
\frac{d}{dt}M_{e,\beta}\left(p_{t}\right) & \leq\frac{a\left(d+a+b\right)}{j}e^{\frac{4d+5a+4b}{j\beta}}\\
 & \leq-\frac{a^{2}}{4j}M_{e,\beta}\left(p_{t}\right)+\frac{a^{2}}{4j}e^{\frac{4d+5a+4b}{j\beta}}+\frac{a\left(d+a+b\right)}{j}e^{\frac{4d+5a+4b}{j\beta}}\\
 & \leq-\frac{a^{2}}{4j}M_{e,\beta}\left(p_{t}\right)+\frac{a\left(d+2a+b\right)}{j}e^{\frac{4d+5a+4b}{j\beta}}.
\end{align*}
Combining these inequalities and from Gronwall inequality, for any
$k\in N$ we have
\[
M_{e,\beta}\left(p_{k}\right)\leq e^{-\frac{a^{2}}{4j}\eta}M_{e,\beta}\left(p_{\left(k-1\right)}\right)+\frac{4\left(d+2a+b\right)}{a}e^{\frac{4d+5a+4b}{j\beta}}
\]
So
\[
M_{e,\beta}\left(p_{k}\right)\leq e^{-\frac{a^{2}}{4j}k\eta}M_{e,\beta}\left(p_{0}\right)+\left(\frac{1}{1-e^{-\frac{a^{2}}{4j}\eta}}\right)\frac{4\left(d+2a+b\right)}{a}e^{\frac{4d+5a+4b}{j\beta}}
\]
or
\begin{align*}
E_{p_{k}}\left[e^{\frac{a}{j\beta}\left(1+\left\Vert x\right\Vert ^{2}\right)^{\frac{\beta}{2}}}\right] & \leq E_{p_{0}}\left[e^{\frac{a}{j\beta}\left(1+\left\Vert x\right\Vert ^{2}\right)^{\frac{\beta}{2}}}\right]+O\left(de^{\frac{4d}{j\beta}}\right)
\end{align*}
If we initialize with a Gaussian distribution $p_{0}=N(0,\frac{1}{L}I)$,
we have
\[
E_{p_{0}}\left[e^{\frac{1}{4}\left\Vert x\right\Vert ^{2}}\right]=O\left(de^{2d}\right)
\]
from which if we choose $2\leq j$ so that $a\leq4j\beta$ then
\[
E_{p_{k}}\left[e^{\frac{a}{j\beta}\left(1+\left\Vert x\right\Vert ^{2}\right)^{\frac{\beta}{2}}}\right]=O\left(de^{2d}\right).
\]
By Jensen inequality
\begin{align*}
e^{E_{p_{k}}\left[\frac{a}{j\beta}\left\Vert x\right\Vert ^{\beta}\right]} & \leq O\left(de^{2d}\right)
\end{align*}
which implies
\[
E_{p_{k}}\left[\left\Vert x\right\Vert ^{\beta}\right]\leq\tilde{O}(d).
\]
Let $E_{p_{k}}\left[\left(\frac{a}{j\beta}\right)^{n}\left\Vert x\right\Vert ^{n\beta}\right]=m\geq0$
for some $n\in N,$$n>1$. By Jensen inequality for any $i\geq n,$$E_{p_{k}}\left[\left(\frac{a}{j\beta}\right)^{i}\left\Vert x\right\Vert ^{i\beta}\right]\geq m^{\frac{i}{n}}\geq0.$
Let $f(m)=e^{m^{\frac{1}{n}}}-1-m^{\frac{1}{n}}-\frac{1}{2!}m^{\frac{2}{n}}..-\frac{1}{(n-1)!}m^{\frac{n-1}{n}},$
we have
\begin{align*}
f(m) & =\sum_{i\geq n}\frac{m^{\frac{i}{n}}}{i!}\\
 & \leq\sum_{i\geq n}\frac{E_{p_{k}}\left[\left(\frac{a}{j\beta}\right)^{i}\left\Vert x\right\Vert ^{i\beta}\right]}{i!}\\
 & \leq E_{p_{k}}\left[e^{\frac{a}{j\beta}\left\Vert x\right\Vert ^{\beta}}\right]\\
 & \leq Kde^{2d},
\end{align*}
for some fixed $K$. Differentiate $f$ with respect to $m$ we get
$f^{\prime}=\frac{1}{nm^{\frac{n-1}{n}}}\left(e^{m^{\frac{1}{n}}}-1-m^{\frac{1}{n}}-...-\frac{1}{(n-2)!}m^{\frac{n-2}{n}}\right)\geq0$
so the function is increasing in $m$. Since for $d$ large enough
$f(2nd)=e^{2nd}-\left(2nd\right)^{\frac{1}{n}}...-\frac{1}{(n-1)!}\left(2nd\right)^{\frac{n-1}{n}}\geq Kde^{2d}\geq f(m^{\frac{1}{n}})$,
which implies $m^{\frac{1}{n}}\leq2nd$ or $m\leq2^{n}n^{n}d^{n}$
which is the desired result.
\end{proof}
\subsection{Proof of Lemma \ref{lem:C3}\label{AC3}}
\begin{proof}
First, recall that the discretization of the ULA is
\begin{center}
$x_{k,t}\stackrel{}{=}x_{k}-t\nabla U(x_{k})+\sqrt{2t}\,z_{k}$,
\par\end{center}
where $z_{k}\sim N(0,I)$ is independent of $x_{k}$. Since $U$ satisfies
$\alpha_{G}$-mixture locally smooth, $\left\Vert \nabla U(x_{k})\right\Vert \leq2NL\left(1+\left\Vert x_{k}\right\Vert ^{\ell_{G}+\alpha_{GN}}\right)$,
which in turn implies $\mathbb{E}_{p_{k}}\left[\left\Vert \nabla U(x_{k})\right\Vert ^{r}\right]\stackrel{}{\leq}C\left(1+\mathbb{E}\left[\left\Vert x_{k}\right\Vert ^{r\left(\ell_{G}+\alpha_{GN}\right)}\right]\right).$
We have
\begin{align*}
\mathrm{\mathbb{E}}_{p_{kt}}\left\Vert x_{k,t}-x_{k}\right\Vert ^{r} & =\mathrm{\mathbb{E}}_{p_{k}}\left\Vert -t\nabla U(x_{k})+\sqrt{2t}z_{k}\right\Vert ^{r}\\
 & \leq2^{r-1}\eta^{r}\mathrm{\mathbb{E}}_{p_{k}}\left\Vert \nabla U(x_{k})\right\Vert ^{r}+2^{r-1}2^{\frac{r}{2}}\eta^{\frac{r}{2}}\left(d+\lceil\frac{r}{2}\rceil\right)^{\lceil\frac{r}{2}\rceil}\\
 & \leq C\eta^{r}\mathrm{\mathbb{E}}_{p_{k}}\left[1+\left\Vert x_{k}\right\Vert ^{\left(\ell_{G}+\alpha_{GN}\right)r}\right]+2^{r-1}2^{\frac{r}{2}}\eta^{\frac{r}{2}}\left(d+\lceil\frac{r}{2}\rceil\right)^{\lceil\frac{r}{2}\rceil}\\
 & =O\left(d^{\lceil\frac{\left(\ell_{G}+\alpha_{GN}\right)r}{\beta}\rceil}\right)\eta^{r}+O\left(d^{\lceil\frac{r}{2}\rceil}\right)\eta^{\frac{r}{2}}\\
 & =O\left(d^{\lceil\frac{\left(\ell_{G}+\alpha_{GN}\right)r}{\beta}\rceil\vee\lceil\frac{r}{2}\rceil}\right)\eta^{\frac{r}{2}}.
\end{align*}
\begin{align*}
\mathrm{\mathbb{E}}_{p_{kt}}\left(1+\left\Vert x_{k}\right\Vert ^{r}+\left\Vert x_{k,t}\right\Vert ^{r}\right) & \leq C_{r}\mathrm{\mathbb{E}}_{p_{kt}}\left[1+\left\Vert x_{k}\right\Vert ^{r}+\left\Vert x_{k,t}-x_{k}\right\Vert ^{r}\right]\\
 & =O\left(d^{\lceil\frac{r}{\beta}\rceil}\right)+O\left(d^{\lceil\frac{\left(\ell_{G}+\alpha_{GN}\right)r}{\beta}\rceil\vee\lceil\frac{r}{2}\rceil}\right)\eta^{\frac{r}{2}}
\end{align*}
For $\eta$ small enough, it is
\[
\mathrm{\mathbb{E}}_{p_{kt}}\left(1+\left\Vert x_{k}\right\Vert ^{r}+\left\Vert x_{k,t}\right\Vert ^{r}\right)\leq O\left(d^{\lceil\frac{r}{\beta}\rceil}\right)
\]
\begin{align}
 & \mathrm{\mathbb{E}}_{p_{kt}}\left\Vert \nabla U(x_{k})-\nabla U(x_{k,t})\right\Vert ^{2}\nonumber \\
 & \stackrel{}{\leq}\mathrm{\mathbb{E}}_{p_{kt}}\left(1+\left\Vert x_{k}\right\Vert ^{\ell_{G}}+\left\Vert x_{k,t}\right\Vert ^{\ell_{G}}\right)^{2}\left(\sum_{i=1}^{N}L_{i}\left\Vert x_{k}-x_{k,t}\right\Vert ^{\alpha_{Gi}}\right)^{2}\nonumber \\
 & \stackrel{_{1}}{\leq}\sqrt{\mathrm{\mathbb{E}}_{p_{kt}}\left[\left(1+\left\Vert x_{k}\right\Vert ^{\ell_{G}}+\left\Vert x_{k,t}\right\Vert ^{\ell_{G}}\right)^{4}\right]}\sqrt{\mathrm{\mathbb{E}}_{p_{kt}}\left(\sum_{i=1}^{N}L_{i}\left\Vert x_{k}-x_{k,t}\right\Vert ^{\alpha_{Gi}}\right)^{4}}\\
 & \stackrel{_{2}}{\leq}C\sqrt{\mathrm{\mathbb{E}}_{p_{kt}}\left[1+\left\Vert x_{k}\right\Vert ^{4\ell_{G}}+\left\Vert x_{k,t}\right\Vert ^{4\ell_{G}}\right]}\sqrt{\sum_{i}\mathrm{\mathbb{E}}_{p_{kt}}\left\Vert x_{k,t}-x_{k}\right\Vert ^{4\alpha_{Gi}}}\nonumber \\
 & \stackrel{_{3}}{\leq}C\sqrt{\mathrm{\mathbb{E}}_{p_{kt}}\left[1+\left\Vert x_{k}\right\Vert ^{4\ell_{G}}+\left\Vert x_{k,t}-x_{k}\right\Vert ^{4\ell_{G}}\right]}\sqrt{\sum_{i}\mathrm{\mathbb{E}}_{p_{kt}}\left\Vert x_{k,t}-x_{k}\right\Vert ^{4\alpha_{Gi}}}\\
 & \stackrel{}{\leq}\sqrt{\left[O\left(d^{\lceil\frac{4\ell_{G}}{\beta}\rceil}\right)+O\left(d^{\lceil\frac{\left(\ell_{G}+\alpha_{GN}\right)4\ell_{G}}{\beta}\rceil\vee\lceil2\ell_{G}\rceil}\right)\eta^{2\ell_{G}}\right]}\sqrt{\left(\sum_{i}O\left(d^{\lceil\frac{\left(\ell_{G}+\alpha_{GN}\right)4\alpha_{Gi}}{\beta}\rceil\vee\lceil2\alpha_{Gi}\rceil}\right)\eta^{2\alpha_{Gi}}\right)}\\
 & =\left[O\left(d^{\frac{\lceil\frac{4\ell_{G}}{\beta}\rceil}{2}}\right)+O\left(d^{\frac{\lceil\frac{\left(\ell_{G}+\alpha_{GN}\right)4\ell_{G}}{\beta}\rceil}{2}\vee\frac{\lceil2\ell_{G}\rceil}{2}}\right)\eta^{\ell_{G}}\right]O\left(d^{\frac{\lceil\frac{\left(\ell_{G}+\alpha_{GN}\right)4\alpha_{GN}}{\beta}\rceil}{2}\vee\frac{\lceil2\alpha_{GN}\rceil}{2}}\right)\eta^{\alpha_{G}}\\
 & =O\left(d^{\frac{\lceil\frac{4\ell_{G}}{\beta}\rceil}{2}+\frac{\lceil\frac{\left(\ell_{G}+\alpha_{GN}\right)4\alpha_{GN}}{\beta}\rceil}{2}\vee\frac{\lceil2\alpha_{GN}\rceil}{2}}\right)\eta^{\alpha_{G}}.\nonumber
\end{align}
where step $1$ follows from Assumption \ref{A0}, step $2$ comes
from Holder inequality and normal distribution, step $3$ is because
of Lemma \ref{lem:C2} and Young inequality, the last three steps
come from choosing $\eta$ small enough. Therefore, from \citep{vempala2019rapid}
Lemma 3, the time derivative of KL divergence along ULA is bounded
by
\begin{align*}
\frac{d}{dt}H\left(p_{k,t}|\pi\right) & \leq-\frac{3}{4}I\left(p_{k,t}|\pi\right)+D\eta^{\alpha_{G}},
\end{align*}
where in the last inequality, we have used the definitions of $D=O\left(d^{\frac{\lceil\frac{4\ell_{G}}{\beta}\rceil}{2}+\frac{\lceil\frac{\left(\ell_{G}+\alpha_{GN}\right)4\alpha_{GN}}{\beta}\rceil}{2}\vee\frac{\lceil2\alpha_{GN}\rceil}{2}}\right)$.
\end{proof}
\begin{remark}
When $\ell_{G}=0,$$D_{3}=O\left(d^{\frac{\lceil\frac{4\alpha_{GN}^{2}}{\beta}\rceil}{2}\vee\frac{\lceil2\alpha_{GN}\rceil}{2}}\right).$
\end{remark}

\subsection{Proof of Theorem \ref{lem:C1}\label{AC1}}

\begin{proof} Let $f=\frac{d\mu}{d\nu}$, we have
\begin{align*}
E_{\nu}(f)-E_{\nu}^{2}(\sqrt{f}) & =1-E_{\nu}^{2}\left(\sqrt{\frac{d\mu}{d\nu}}\right)\\
 & =\left(1-E_{\nu}\left(\sqrt{\frac{d\mu}{d\nu}}\right)\right)\left(1+E_{\pi}\left(\sqrt{\frac{d\mu}{d\nu}}\right)\right)\\
 & \geq\frac{1}{2}\int\left(\sqrt{d\mu}-\sqrt{\nu}\right)^{2}dx.
\end{align*}

Since $\nu$ satisfies $\gamma$-Poincaré inequality, we obtain
\begin{align*}
E_{\nu}(f)-E_{\nu}^{2}(\sqrt{f}) & \leq\frac{1}{\gamma}E_{\nu}\left\Vert \nabla\left(\sqrt{f}\right)\right\Vert ^{2}\\
 & \leq\frac{1}{4\gamma}E_{\nu}\frac{1}{f}\left\Vert \nabla f\right\Vert ^{2}\\
 & \leq\frac{1}{4\gamma}E_{\mu}\left\Vert \nabla\log f\right\Vert ^{2}\\
 & \leq\frac{1}{4\gamma}I\left(p|\nu\right).
\end{align*}
As a result, we have
\begin{equation}
\int\left(\sqrt{d\mu}-\sqrt{\nu}\right)^{2}dx\leq\frac{1}{2\gamma}I\left(\mu|\nu\right).\label{eq:Information}
\end{equation}
By using an inequality from \citep{villani2008optimal}, we get

\begin{align*}
W_{q}^{q}(\mu,\nu) & \leq2^{q-1}\int_{\mathbb{R}^{d}}\left\Vert x\right\Vert ^{q}\left|\mu(x)-\nu(x)\right|dx\\
 & \stackrel{_{1}}{\leq}2^{q-1}\left(\int_{\mathbb{R}^{d}}\left\Vert x\right\Vert ^{2q}\left(\sqrt{\mu}+\sqrt{\nu}\right)^{2}dx\right)^{\frac{1}{2}}\left(\int\left(\sqrt{\mu}-\sqrt{\nu}\right)^{2}dx\right)^{\frac{1}{2}}\\
 & \stackrel{_{2}}{\leq}2^{q-1}\sqrt{2}\left(\int_{\mathbb{R}^{d}}\left\Vert x\right\Vert ^{2q}\left(d\mu+d\nu\right)\right)^{\frac{1}{2}}\left(\int\left(\sqrt{\mu}-\sqrt{\nu}\right)^{2}dx\right)^{\frac{1}{2}}\\
 & \stackrel{}{\leq}2^{q-1}M_{2q}^{\frac{1}{2}}\left(\mu+\nu\right)\sqrt{\frac{1}{\gamma}}\sqrt{I\left(\mu|\nu\right),}
\end{align*}
where step $1$ follows from Holder inequality, step $2$ is because
of Young inequality and both $\mu$ and $\nu$ are non-negative, and
in the last step, we have used the definition of $M_{2q}$ and the
inequality \ref{eq:Information} above. As a result, we have

\[
W_{q}(\mu,\nu)\leq2M_{2q}^{\frac{1}{2q}}\left(\mu+\nu\right)\gamma^{-\frac{1}{2q}}I^{\frac{1}{2q}}\left(\mu|\nu\right).
\]
Adapted \citet{polyanskiy2016wasserstein}'s Proposition 1 technique,
we have
\begin{align*}
\left|\log\nu(x)-\log\nu(x^{*})\right| & =\left|\int_{0}^{1}\left\langle \nabla\log\nu(tx+(1-t)x^{*}),x-x^{*}\right\rangle dt\right|\\
 & \leq\int_{0}^{1}\left\Vert \nabla\log\nu(tx+(1-t)x^{*})\right\Vert \left\Vert x-x^{*}\right\Vert dt\\
 & \leq2NL\int_{0}^{1}\left(1+\left\Vert tx+(1-t)x^{*}\right\Vert ^{\ell_{G}+\alpha_{GN}}\right)\left\Vert x-x^{*}\right\Vert dt\\
 & \leq2NL\int_{0}^{1}\left(1+\left\Vert x\right\Vert ^{\ell_{G}+\alpha_{GN}}+\left\Vert x^{*}\right\Vert ^{\ell_{G}+\alpha_{GN}}\right)\left\Vert x-x^{*}\right\Vert dt\\
 & \leq2NL\left(1+\left\Vert x\right\Vert ^{\ell_{G}+\alpha_{GN}}+\left\Vert x^{*}\right\Vert ^{\ell_{G}+\alpha_{GN}}\right)\left\Vert x-x^{*}\right\Vert .
\end{align*}
Let $\left(x,x^{*}\right)$ be $W_{q}$-optimal coupling of $\mu$
and $\nu$, $\frac{1}{q^{\prime}}+\frac{1}{q}=1,$ let $q=\frac{\ell_{G}+\alpha_{GN}}{2}+1$,
$q^{\prime}=\frac{\ell_{G}+\alpha_{GN}+2}{\ell_{G}+\alpha_{GN}}$,
taking the expectation with respect to this optimal coupling we obtain
\begin{center}
\begin{align*}
H(\mu|\nu) & \leq\mathbb{E}\left[2NL\left(1+\left\Vert x\right\Vert ^{\ell_{G}+\alpha_{GN}}+\left\Vert x^{*}\right\Vert ^{\ell_{G}+\alpha_{GN}}\right)\left\Vert x-x^{*}\right\Vert \right]\\
 & \stackrel{_{1}}{\leq}2NL\mathbb{E}\left[\left(1+\left\Vert x\right\Vert ^{\ell_{G}+\alpha_{GN}}+\left\Vert x^{*}\right\Vert ^{\ell_{G}+\alpha_{GN}}\right)^{q^{\prime}}\right]^{\frac{1}{q^{\prime}}}\left(\mathbb{E}\left\Vert x-x^{*}\right\Vert ^{q}\right)^{\frac{1}{q}}\\
 & \stackrel{_{2}}{\leq}2NL\mathbb{E}\left[3^{q^{\prime}-1}\left(1+\left\Vert x\right\Vert ^{q^{\prime}\left(\ell_{G}+\alpha_{GN}\right)}+\left\Vert x^{*}\right\Vert ^{q^{\prime}\left(\ell_{G}+\alpha_{GN}\right)}\right)\right]^{\frac{1}{q^{\prime}}}W_{q}(\mu,\nu)\\
 & \leq CM_{q^{\prime}\left(\ell_{G}+\alpha_{GN}\right)}^{\frac{1}{q^{\prime}}}\left(\mu+\nu\right)W_{q}(\mu,\nu)\\
 & \stackrel{}{\leq}CM_{q^{\prime}\left(\ell_{G}+\alpha_{GN}\right)}^{\frac{1}{q^{\prime}}}\left(\mu+\nu\right)M_{2q}^{\frac{1}{2q}}\left(\mu+\nu\right)\gamma^{-\frac{1}{2q}}I^{\frac{1}{2q}}\left(\mu|\nu\right)\\
 & \stackrel{3}{\leq}C\gamma^{-\frac{1}{2q}}M_{2q}\left(\mu+\nu\right)I^{\frac{1}{2q}}\left(\mu|\nu\right),\\
 & {\displaystyle =}C\gamma^{-\frac{1}{2q}}M_{\ell_{G}+\alpha_{GN}+2}\left(\mu+\nu\right)I^{\frac{1}{\ell_{G}+\alpha_{GN}+2}}\left(\mu|\nu\right),
\end{align*}
\par\end{center}

where step $1$ follows from Holder inequality, step $2$ is due to
$\alpha_{N}\leq1$ and Cauchy-Schwartz inequality, step $3$ is because
$\frac{1}{q\prime}+\frac{1}{q}=1$ and we have used $q=\frac{\ell_{G}+\alpha_{GN}}{2}+1$,
$q^{\prime}=\frac{\ell_{G}+\alpha_{GN}+2}{\ell_{G}+\alpha_{GN}}$.

\end{proof}
\subsection{Proof of Lemma \ref{lem:C4a}\label{AC4a}}
\begin{lemma}\label{lem:C4a} (\citep{nguyen2021unadjusted} Lemma 1). Suppose $\nu$ is
$\beta$-dissipative, $\alpha_{G}$-mixture locally smooth. If $H(\tilde{p}_{k,t}|\nu)\leq C\gamma^{-\frac{1}{\ell_{G}+\alpha_{GN}+2}}I^{\frac{1}{\ell_{G}+\alpha_{GN}+2}}\left(\tilde{p}_{k,t}|\nu\right)M_{\ell_{G}+\alpha_{GN}+2}\left(\tilde{p}_{k,t}+\nu\right)$
then for step size $\eta$ small enough
\[
H(p_{k+1}|\nu)\leq H(p_{k}|\nu)\left(1-\frac{3CH(p_{k}|\nu)^{\ell_{G}+\alpha_{GN}+1}}{8\gamma\left(M_{\ell_{G}+\alpha_{GN}+2}^{\ell_{G}+\alpha_{GN}+2}(\tilde{p}_{k,t}+\nu)\right)}\eta\right)+D_{3}\eta^{\alpha_{G}+1}.
\]
\end{lemma}
\subsection{Proof of Theorem \ref{thm:C4}\label{AC4}}
\begin{proof}Integrating both sides of equation
from $t=0$ to $t=\eta$ we obtain
\begin{align*}
 & H(p_{k+1}|\nu)-H(p_{k}|\nu)\leq D\eta^{1+\alpha_{G}},
\end{align*}
where the inequality holds since the first term is negative. Using
discrete Grönwall inequality, we have, \textit{for any} $k\in\mathbb{N}$
\begin{align*}
H(p_{K}|\nu) & \leq H(p_{k_{0}}|\nu)+KD\eta^{1+\alpha_{G}}\\
 & \leq H(p_{k_{0}}|\nu)+TD\eta^{\alpha_{G}}\\
 & \leq H(p_{k_{0}}|\nu)+\frac{\epsilon}{2}.
\end{align*}
If there exists some $k<K$ such that $H(p_{k}|\nu)\leq\frac{\epsilon}{2}$
then we can choose $\eta\leq\left(\frac{\epsilon}{2TD}\right)^{\frac{1}{\alpha_{G}}}$
so that $H(p_{K}|\nu)\leq\epsilon$. If there is no such $k$, we
will prove for sufficiently large $K$, $H(p_{K}|\nu)\leq\epsilon$.
Let $A=\frac{3C}{8\gamma\left(M_{\ell_{G}+\alpha_{GN}+2}^{\ell_{G}+\alpha_{GN}+2}(\tilde{p}_{k,t}+\pi)\right)}\left(\frac{\epsilon}{2}\right)^{\ell_{G}+\alpha_{GN}+1}$,
the above expression leads to
\begin{align*}
H(p_{k+1}|\nu) & \leq H(p_{k}|\nu)\left(1-A\eta\right)+D\eta^{\alpha_{G}+1}.
\end{align*}
By iterating the process we get
\begin{align*}
H(p_{k}|\nu) & \leq H(p_{0}|\nu)\left(1-A\eta\right)^{k}+\frac{D}{A}\eta^{\alpha_{G}}.
\end{align*}
To get $H(p_{K}|\nu)\leq\epsilon$, for $\eta$ small enough so that
$\eta\leq\left(\frac{A\epsilon}{2D}\right)^{\frac{1}{\alpha_{G}}}$,
it suffices to run $K$ iterations such that
\[
\left(1-A\eta\right)^{K}\leq\frac{\epsilon}{2H(p_{0}|\nu)}.
\]
As a result, we obtain
\begin{align*}
K & =\log_{\left(1-A\eta\right)}\left(\frac{\epsilon}{\left(2H(p_{0}|\nu)\right)}\right)\\
 & =\frac{\ln\left(\frac{\left(H(p_{0}|\nu)\right)}{\epsilon}\right)}{\ln\left(\frac{1}{1-A\eta}\right)}\\
 & \leq\frac{\ln\left(\frac{\left(H(p_{0}|\nu)\right)}{\epsilon}\right)}{\frac{3C}{8\gamma\left(M_{\ell_{G}+\alpha_{GN}+2}^{\ell_{G}+\alpha_{GN}+2}(\tilde{p}_{k,t}+\pi)\right)}\left(\frac{\epsilon}{2}\right)^{\ell_{G}+\alpha_{GN}+1}\eta}.
\end{align*}
By plugging $T=K\eta$ and assuming without loss of generality that
$T>1$ (since we can choose $T$), we obtain
\begin{align*}
T & \leq\frac{\ln\left(\frac{\left(H(p_{0}|\nu)\right)}{\epsilon}\right)}{\frac{3C}{8\gamma\left(M_{\ell_{G}+\alpha_{GN}+2}^{\ell_{G}+\alpha_{GN}+2}(\tilde{p}_{k,t}+\pi)\right)}\left(\frac{\epsilon}{2}\right)^{\ell_{G}+\alpha_{GN}+1}}
\end{align*}
which is satisfied if we choose
\[
T=O\left(\frac{\gamma\left(M_{\ell_{G}+\alpha_{GN}+2}^{\ell_{G}+\alpha_{GN}+2}(\tilde{p}_{k,t}+\nu)\right)\ln\left(\frac{\left(H(p_{0}|\nu)\right)}{\epsilon}\right)}{\epsilon^{\ell_{G}+\alpha_{GN}+1}}\right).
\]
Without loss of generality, since $H(p_{0}|\nu)=O\left(d^{\frac{\ell_{G}+1+\alpha_{GN}}{2}}\right),$
we can assume that $H(p_{0}|\nu)\geq1>\epsilon.$ We have $\ln\left(\frac{\left(H(p_{0}|\nu)\right)}{\epsilon}\right)>1$.
Therefore,
\[
\eta=\min\left\{ 1,\left(\frac{A\epsilon}{2D}\right)^{\frac{1}{\alpha_{G}}},\left(\frac{\epsilon}{2TD}\right)^{\frac{1}{\alpha_{G}}}\right\} =\left(\frac{\epsilon}{2TD}\right)^{\frac{1}{\alpha_{G}}}
\]
Using $K=\frac{T}{\eta}$, we have
\begin{align*}
K & \leq O\left(\left(\frac{2TD}{\epsilon}\right)^{\frac{1}{\alpha_{G}}}\frac{\gamma\left(M_{\ell_{G}+\alpha_{GN}+2}^{\ell_{G}+\alpha_{GN}+2}(\tilde{p}_{k,t}+\pi)\right)\ln\left(\frac{\left(H(p_{0}|\pi)\right)}{\epsilon}\right)}{\epsilon^{\ell_{G}+\alpha_{GN}+1}}\right)\\
 & \leq O\left(\frac{D^{\frac{1}{\alpha_{G}}}\gamma^{1+\frac{1}{\alpha_{G}}}d^{\lceil\frac{\ell_{G}+\alpha_{GN}+2}{\beta}\rceil\left(\ell_{G}+\alpha_{GN}+2\right)\left(1+\frac{1}{\alpha_{G}}\right)}\ln^{\left(1+\frac{1}{\alpha_{G}}\right)}\left(\frac{\left(H(p_{0}|\pi)\right)}{\epsilon}\right)}{\epsilon^{\left(\ell_{G}+\alpha_{GN}+1\right)\left(1+\frac{1}{\alpha_{G}}\right)+\frac{1}{\alpha_{G}}}}\right).
\end{align*}
Combining with these above results for $\eta$ small enough, $M_{\ell_{G}+\alpha_{GN}+2}(\tilde{p}_{k,t}+\nu)=O\left(d^{\lceil\frac{\ell_{G}+\alpha_{GN}+2}{\beta}\rceil}\right)$
and $D=O\left(d^{\frac{\lceil\frac{4\ell_{G}}{\beta}\rceil}{2}+\frac{\lceil\frac{\left(\ell_{G}+\alpha_{GN}\right)4\alpha_{GN}}{\beta}\rceil}{2}\vee\frac{\lceil2\alpha_{GN}\rceil}{2}}\right)$,
we obtain
\[
K=O\left(\frac{\gamma^{1+\frac{1}{\alpha_{G}}}d^{\lceil\frac{\ell_{G}+\alpha_{GN}+2}{\beta}\rceil\left(\ell_{G}+\alpha_{GN}+2\right)\left(1+\frac{1}{\alpha_{G}}\right)+\frac{\lceil\frac{4\ell_{G}}{\beta}\rceil}{2}+\frac{\lceil\frac{\left(\ell_{G}+\alpha_{GN}\right)4\alpha_{GN}}{\beta}\rceil}{2}\vee\frac{\lceil2\alpha_{GN}\rceil}{2}}\ln^{\left(1+\frac{1}{\alpha_{G}}\right)}\left(\frac{\left(H(p_{0}|\nu)\right)}{\epsilon}\right)}{\epsilon^{\left(\ell_{G}+\alpha_{GN}+1\right)\left(1+\frac{1}{\alpha_{G}}\right)+\frac{1}{\alpha_{G}}}}\right)
\]

which is our desired result. \end{proof} \begin{remark} We can get
a tighter result for each specific case. For example, by choosing
$\ell_{G}=0,$ $\alpha_{GN}=\alpha_{G}=1$, $\beta\geq1.5$ we obtain
\[
K\approx\tilde{O}\left(\frac{\gamma^{2}d^{13}}{\epsilon^{5}}\right),
\]
a weaker but rather comparable to the result of \citep{erdogdu2020convergence}.
\end{remark}
\subsection{Proof of $\alpha_{H}$-mixture Hessian locally-smooth property\label{Asmooth-1}}

\begin{lemma} If potential $U:\mathbb{R}^{d}\rightarrow\mathbb{R}$
satisfies $\alpha$-mixture locally-smooth then:

\begin{align*}
\left\Vert \nabla U(x)-\nabla U(y)\right\Vert  & \leq\left(\frac{4L_{H}}{1+\alpha_{H}}\vee C_{H}\right)\left(1+\left\Vert x\right\Vert ^{\ell_{H}+\alpha_{N}}+\left\Vert y\right\Vert ^{\ell_{H}+\alpha_{N}}\right)\sum_{i=0}\left\Vert x-y\right\Vert ^{1+\alpha_{Hi}}.
\end{align*}

\end{lemma}\begin{proof}

We have
\begin{align*}
 & \left\Vert \nabla U(x)-\nabla U(y)-\nabla^{2}U(y)(x-y)\right\Vert \\
= & \left\Vert \int_{0}^{1}\left(\nabla^{2}U(y+t(x-y))(x-y)-\nabla^{2}U(y)(x-y)\right)dt\right\Vert \\
\leq & \int_{0}^{1}\left\Vert \nabla^{2}U(y+t(x-y))-\nabla^{2}U(y)\right\Vert _{op}\left\Vert x-y\right\Vert dt\\
\leq & \int_{0}^{1}\left(1+\left\Vert tx+(1-t)y\right\Vert ^{\ell_{H}}+\left\Vert y\right\Vert ^{\ell_{H}}\right)\sum_{i=1}^{N}L_{Hi}t^{\alpha_{Hi}}\left\Vert x-y\right\Vert ^{\alpha_{Hi}}\left\Vert x-y\right\Vert dt\\
\leq & \sum_{i}\frac{2L_{Hi}}{1+\alpha_{Hi}}\left(1+\left\Vert x\right\Vert ^{\ell_{H}}+\left\Vert y\right\Vert ^{\ell_{H}}\right)\left\Vert x-y\right\Vert ^{1+\alpha_{Hi}}\\
\leq & \left(\frac{4L_{H}}{1+\alpha_{H}}\vee C_{H}\right)\left(1+\left\Vert x\right\Vert ^{\ell_{H}+\alpha_{N}}+\left\Vert y\right\Vert ^{\ell_{H}+\alpha_{N}}\right)\sum_{i=0}\left\Vert x-y\right\Vert ^{1+\alpha_{Hi}},
\end{align*}
where the first line comes from Taylor expansion, the third line follows
from Cauchy-Schwarz inequality and the fourth line is due to Assumption~\ref{A0}.
From the bound for $\left\Vert \nabla^{2}U(x)\right\Vert $ we obtain:

\begin{align*}
\left\Vert \nabla U(x)-\nabla U(y)\right\Vert  & =\left\Vert \int_{0}^{1}\nabla^{2}U\left((1-t)y+tx\right)(x-y)dt\right\Vert \\
 & \leq\int_{0}^{1}\left\Vert \nabla^{2}U\left((1-t)y+tx\right)(x-y)\right\Vert dt\\
 & \leq\int_{0}^{1}\left\Vert \nabla^{2}U\left((1-t)y+tx\right)\right\Vert _{\mathrm{op}}dt\left\Vert x-y\right\Vert \\
 & \leq\int_{0}^{1}C_{H}\left(1+\left\Vert (1-t)y+tx\right\Vert ^{\ell_{H}+\alpha_{HN}}\right)dt\left\Vert x-y\right\Vert \\
 & \leq C_{H}\left(1+\left\Vert x\right\Vert ^{\ell_{H}+\alpha_{HN}}+\left\Vert y\right\Vert ^{\ell_{H}+\alpha_{HN}}\right)\left\Vert x-y\right\Vert .
\end{align*}

From that, let $y=0,$we get

\begin{align*}
\left\Vert \nabla U(x)\right\Vert  & \leq C_{H}\left(1+\left\Vert x\right\Vert ^{\ell_{H}+\alpha_{HN}}\right)\left\Vert x\right\Vert \\
 & \leq2C_{H}\left(1+\left\Vert x\right\Vert ^{\ell_{H}+\alpha_{HN}+1}\right).
\end{align*}

This gives us the desired result.
\end{proof}

\subsection{Proof of $\alpha_{H}$-mixture Hessian locally-smooth property\label{Asmooth-1-1}}

\begin{align}
 & \mathrm{\mathbb{E}}_{p_{kt}}\left\Vert \nabla U(x_{k})-\nabla U(x_{k,t})\right\Vert ^{2}\nonumber \\
 & \stackrel{_{1}}{\leq}\mathrm{\mathbb{E}}_{p_{kt}}\left(C_{H}\left(1+\left\Vert x_{k}\right\Vert ^{\ell_{H}+\alpha_{N}}\right)\left\Vert x_{k}-x_{k,t}\right\Vert +\frac{2L_{H}}{1+\alpha_{H}}\left(1+\left\Vert x_{k}\right\Vert ^{\ell_{H}}+\left\Vert x_{k,t}\right\Vert ^{\ell_{H}}\right)\sum_{i}\left\Vert x_{k}-x_{k,t}\right\Vert ^{1+\alpha_{Hi}}\right)^{2}\nonumber \\
 & \stackrel{_{2}}{\leq}2C_{H}^{2}\mathrm{\mathbb{E}}_{p_{kt}}\left(1+\left\Vert x_{k}\right\Vert ^{\ell_{H}+\alpha_{N}}\right)^{2}\left\Vert x_{k}-x_{k,t}\right\Vert ^{2}+2N\left(\frac{2L_{H}}{1+\alpha_{H}}\right)^{2}\mathrm{\mathbb{E}}_{p_{kt}}\left[\left(1+\left\Vert x_{k}\right\Vert ^{\ell_{H}}+\left\Vert x_{k,t}\right\Vert ^{\ell_{H}}\right)^{2}\sum_{i}\left\Vert x_{k}-x_{k,t}\right\Vert ^{2+2\alpha_{Hi}}\right]\\
 & \stackrel{_{3}}{\leq}4C_{H}^{2}\mathrm{\mathbb{E}}_{p_{kt}}\left(1+\left\Vert x_{k}\right\Vert ^{2\ell_{H}+2\alpha_{N}}\right)\left\Vert x_{k}-x_{k,t}\right\Vert ^{2}+6N\left(\frac{2L_{H}}{1+\alpha_{H}}\right)^{2}\mathrm{\mathbb{E}}_{p_{kt}}\left[\left(1+\left\Vert x_{k}\right\Vert ^{2\ell_{H}}+\left\Vert x_{k,t}\right\Vert ^{2\ell_{H}}\right)\sum_{i}\left\Vert x_{k}-x_{k,t}\right\Vert ^{2+2\alpha_{Hi}}\right]\nonumber \\
 & \stackrel{_{4}}{\leq}4C_{H}^{2}\sqrt{\mathrm{\mathbb{E}}_{p_{kt}}\left(1+\left\Vert x_{k}\right\Vert ^{2\ell_{H}+2\alpha_{N}}\right)^{2}}\sqrt{\mathrm{\mathbb{E}}_{p_{kt}}\left\Vert x_{k}-x_{k,t}\right\Vert ^{4}}\\
 & +6N\left(\frac{2L_{H}}{1+\alpha_{H}}\right)^{2}\sqrt{\mathrm{\mathbb{E}}_{p_{kt}}\left[\left(1+\left\Vert x_{k}\right\Vert ^{2\ell_{H}}+\left\Vert x_{k,t}\right\Vert ^{2\ell_{H}}\right)^{2}\right]}\sqrt{\mathrm{\mathbb{E}}_{p_{kt}}\left[\left(\sum_{i}\left\Vert x_{k}-x_{k,t}\right\Vert ^{2+2\alpha_{Hi}}\right)^{2}\right]}\\
 & \stackrel{_{5}}{\leq}4\sqrt{2}C_{H}^{2}\mathrm{\mathbb{E}}_{p_{kt}}\left(1+\left\Vert x_{k}\right\Vert ^{2\ell_{H}+2\alpha_{N}}\right)\sqrt{O\left(d^{\left\lfloor \frac{\left(\ell_{H}+\alpha_{HN}+1\right)4}{\beta}+1\right\rfloor }\right)\eta^{2}}\\
 & +6\sqrt{3}N\left(\frac{2L_{H}}{1+\alpha_{H}}\right)^{2}\mathrm{\mathbb{E}}_{p_{kt}}\left[\left(1+\left\Vert x_{k}\right\Vert ^{2\ell_{H}}+\left\Vert x_{k,t}\right\Vert ^{2\ell_{H}}\right)\right]\sqrt{\mathrm{\mathbb{E}}_{p_{kt}}\left[\left(\sum_{i}\left\Vert x_{k}-x_{k,t}\right\Vert ^{2+2\alpha_{Hi}}\right)^{2}\right]}\\
 & \stackrel{_{5}}{\leq}O\left(d^{\left\lfloor \frac{2\left(\ell_{H}+\alpha_{N}\right)\left(\ell_{H}+\alpha_{HN}+1\right)}{\beta}+1\right\rfloor }\right)\sqrt{O\left(d^{\left\lfloor \frac{\left(\ell_{H}+\alpha_{HN}+1\right)4}{\beta}+1\right\rfloor }\right)\eta^{2}}\\
 & +\left(O\left(d^{\left\lfloor \frac{2\ell_{H}}{\beta}+1\right\rfloor }\right)+O\left(d^{\left\lfloor \frac{\left(\ell+\alpha_{N}\right)\ell_{H}}{\beta}+1\right\rfloor \vee\left\lfloor \ell_{H}\right\rfloor +1_{2\ell_{H}\neq even}}\right)\eta^{\ell_{H}}\right)\left(\sum_{i}O\left(d^{\left\lfloor \frac{2\left(\ell_{H}+\alpha_{HN}+1\right)\left(1+\alpha_{Hi}\right)}{\beta}+1\right\rfloor }\right)\eta^{1+\alpha_{Hi}}\right)\\
 & =O\left(d^{\left\lfloor \frac{2\left(\ell_{H}+\alpha_{N}\right)\left(\ell_{H}+\alpha_{HN}+1\right)}{\beta}+1\right\rfloor +0.5\left\lfloor \frac{\left(\ell_{H}+\alpha_{HN}+1\right)4}{\beta}+1\right\rfloor }\right)\eta,\nonumber
\end{align}
where step $1$ follows from Assumption \ref{A0}, step $2$ comes
from Young inequality and normal distribution, step $3$ is because
of Lemma \ref{lem:C2} and $\eta\leq1$, step $4$ comes from choosing
$\eta$ small enough.

\subsection{Proof of Theorem \ref{thm:C4}\label{AC4-1}}

\begin{proof}

Follow similar step as in the proof above, let $A=\frac{3C}{8\gamma\left(M_{\ell_{H}+\alpha_{HN}+3}^{\ell_{H}+\alpha_{HN}+3}(\tilde{p}_{k,t}+\nu)\right)}\left(\frac{\epsilon}{2}\right)^{\ell_{H}+\alpha_{HN}+2}$.

Combining with these above results for $\eta$ small enough, $M_{\ell_{H}+\alpha_{HN}+3}(\tilde{p}_{k,t}+\nu)=O\left(d^{\lceil\frac{\ell_{H}+\alpha_{HN}+3}{\beta}\rceil}\right)$
and $D=O\left(d^{\frac{\lceil\frac{4\left(\ell_{H}+\alpha_{HN}\right)}{\beta}\rceil}{2}+\frac{\lceil\frac{\left(\ell_{H}+\alpha_{HN}+1\right)\left(4\alpha_{HN}+4\right)}{\beta}\rceil}{2}}\right)$,
we obtain

\[
K=O\left(\frac{\gamma^{1+\frac{1}{\alpha_{H}+1}}d^{\lceil\frac{\ell_{H}+\alpha_{HN}+3}{\beta}\rceil\left(\ell_{H}+\alpha_{HN}+3\right)\left(1+\frac{1}{\alpha_{H}+1}\right)+\frac{\lceil\frac{\left(4\alpha_{HN}+4\right)}{\beta}\rceil}{2}+\frac{\lceil\frac{\left(\ell_{H}+\alpha_{HN}+1\right)\left(4\alpha_{HN}+4\right)}{\beta}\rceil}{2}\vee1}\ln^{\left(1+\frac{1}{\alpha_{H}+1}\right)}\left(\frac{\left(H(p_{0}|\nu)\right)}{\epsilon}\right)}{\epsilon^{\left(\ell_{H}+\alpha_{HN}+2\right)\left(1+\frac{1}{\alpha_{H}+1}\right)+\frac{1}{\alpha_{H}+1}}}\right)
\]

which is our desired result. \end{proof}

\section{Proof under gradient Lipschitz and $\alpha_{H}$-mixture Hessian
locally-smooth}

\subsection{Proof under gradient Lipschitz and $\alpha_{H}$-mixture Hessian
locally-smooth}

\begin{align*}
 & \mathrm{\mathbb{E}}_{p_{kt}}\left\Vert \nabla U(x_{k})-\nabla U(x_{k,t})\right\Vert ^{2}\\
 & \stackrel{_{1}}{\leq}L_{G}^{2}\mathrm{\mathbb{E}}_{p_{kt}}\left[\left\Vert x_{k}-x_{k,t}\right\Vert ^{2}\right]\\
 & \stackrel{_{2}}{\leq}O\left(d^{\lceil\frac{2}{\beta}\rceil}\right)\eta,
\end{align*}
where step $1$ follows from Assumption \ref{A0}, step $2$ is because
of Lemma \ref{lem:C2} and $\eta\leq1$.

\subsection{Proof of gradient Lipschitz and $\alpha_{H}$-mixture Hessian locally-smooth}

\begin{proof} We provide proof here for completeness since we do
not use Hessian smooth condition in \citep{mou2022improved}. We follow closely
the Lemma from \citep{mou2022improved}. We decompose the difference
$y-x$ into
\begin{align*}
a_{1}(x,\ y)\  & :=\left(I\ +(t-k\eta)\nabla^{2}U(y)\right)\left(y-x+(t-k\eta)\nabla U(y)\right),\\
a_{2}(x,\ y) & :=(t-k\eta)\nabla^{2}U(y)\left(y-x+(t-k\eta)\nabla U(y)\right)\\
a_{3}(x,\ y) & :=(t-k\eta)\nabla U(y).
\end{align*}
We define the conditional expectations $I_{i}(x):=\mathbb{E}\left[a_{i}(x_{k},x_{k,t})|x_{k,t}=x\right]$
for $i=1,2,3$ a potentials via the three terms separately. From \citep{mou2022improved} Lemma 4,
\[
\mathbb{E}\left\Vert I_{1}(x_{k,t})\right\Vert ^{2}\leq(t-k\eta)^{2}\int p_{k}(x)\left\Vert \nabla\log p_{k}(x)\right\Vert ^{2}dx.
\]
In addition
\begin{align*}
{\displaystyle \frac{\left\Vert I_{2}(x)\right\Vert }{t-k\eta}} & =\left\Vert \int\nabla^{2}U(y)\left(y-x+(t-k\eta)\nabla U(y)\right)\left(2\pi(t-k\eta)\right)^{-\frac{d}{2}}\exp\left(-\frac{\left\Vert x-y-(t-k\eta)\nabla U(y)\right\Vert ^{2}}{2(t-k\eta)}\right)\frac{\hat{\pi}_{k\eta}(y)}{\hat{\pi}_{t}(x)}dy\right\Vert \\
 & =\int\nabla^{2}U(y)\left(y-x+(t-k\eta)\nabla U(y)\right)\frac{\hat{\pi}_{k\eta}(y)}{\hat{\pi}_{t}(x)}p\left(x_{k,t}=x|x_{k}=y\right)dy\\
 & =\mathbb{E}\left[\nabla^{2}U(y)\left(y-x+(t-k\eta)\nabla U(y)\right)|x_{k,t}=x\right].
\end{align*}
Plugging into the squared integral yields
\begin{align*}
{\displaystyle \mathbb{E}\left\Vert I_{2}(x_{k,t})\right\Vert ^{2}} & =(t-k\eta)^{2}\mathbb{E}\left\Vert \mathbb{E}\left[\nabla^{2}U(y)\left(y-x+(t-k\eta)\nabla U(y)\right)|x_{k,t}=x\right]\right\Vert ^{2}\\
 & \leq(t-k\eta)^{2}\mathbb{E}\mathbb{E}\left[\left\Vert \nabla^{2}U(y)\left(y-x+(t-k\eta)\nabla U(y)\right)\right\Vert ^{2}|x_{k,t}=x\right]\\
 & \leq(t-k\eta)^{2}\mathbb{E}\left[\left\Vert \nabla^{2}U(y)\left(y-x+(t-k\eta)\nabla U(y)\right)\right\Vert ^{2}\right]\\
 & \leq(t-k\eta)^{2}\mathbb{E}\left[\left\Vert \nabla^{2}U(y)\right\Vert ^{2}\left\Vert \left(y-x+(t-k\eta)\nabla U(y)\right)\right\Vert ^{2}\right]\\
 & \leq(t-k\eta)^{2}L_{G}^{2}\mathbb{E}\left[\left\Vert x_{k}+(t-k\eta)\nabla U(x_{k})-x_{k,t}\right\Vert ^{2}\right]\\
 & \leq(t-k\eta)^{2}L_{G}^{2}\left(\mathbb{E}\left[\left\Vert \int_{k\eta}^{t}dB_{s}\right\Vert ^{2}\right]\right)\\
 & \leq L_{G}^{2}d\eta^{3}
\end{align*}
The size of norm of $I_{3}$ can be controlled using
\begin{align*}
\mathbb{E}\left\Vert I_{3}(x_{k,t})\right\Vert ^{2} & =(t-k\eta)^{2}\mathbb{E}\left\Vert \mathbb{E}\left(\nabla U(x_{k})|x_{k,t}\right)\right\Vert ^{2}\\
 & \leq\eta^{2}\mathbb{E}\left\Vert \nabla U(x_{k})\right\Vert ^{2}\\
 & \leq C_{1}\eta^{2}\mathbb{E}\left[\left(1+\left\Vert x_{k}\right\Vert ^{2}\right)\right]\\
 & \leq O\left(d^{\lceil\frac{2}{\beta}\rceil}\right)\eta^{2}.
\end{align*}
We also bound the remainder term $\hat{r}_{t}$
\begin{align*}
\hat{r}_{t}(x) & =\mathrm{\mathbb{E}}\left[\int_{0}^{1}\left(\nabla^{2}U\left((1-s)x_{k,t}+sx_{k}\right)-\nabla^{2}U(x)\right)\left(x_{k}-x_{k,t}\right)ds|x_{k,t}=x\right].
\end{align*}
Taking the global expectation leads to
\begin{center}
\begin{align*}
\mathrm{\mathrm{\mathrm{\mathbb{E}}}}\left\Vert \hat{r}_{t}(x)\right\Vert ^{2} & \leq\mathrm{\mathbb{E}}\left\Vert \mathrm{\mathbb{E}}\int_{0}^{1}\left(\nabla^{2}U\left((1-s)x_{k,t}+sx_{k}\right)-\nabla^{2}U(x)\right)\left(x_{k}-x_{k,t}\right)ds|x_{k,t}=x\right\Vert ^{2}\\
 & \leq\mathrm{\mathbb{E}}\int_{0}^{1}\mathrm{\mathbb{E}}\left[\left\Vert \left(\nabla^{2}U\left((1-s)x_{k,t}+sx_{k}\right)-\nabla^{2}U(x)\right)\left(x_{k}-x_{k,t}\right)\right\Vert ^{2}|\hat{X}_{t}=x\right]ds\\
 & \leq\mathrm{\mathbb{E}}\int_{0}^{1}\mathrm{\mathbb{E}}\left[\left\Vert \left(\nabla^{2}U\left((1-s)x_{k,t}+sx_{k}\right)-\nabla^{2}U(x)\right)\right\Vert ^{2}\left\Vert x_{k}-x_{k,t}\right\Vert ^{2}|\hat{X}_{t}=x\right]ds\\
 & \leq\mathrm{\mathbb{E}}\int_{0}^{1}\mathrm{\mathbb{E}}\left[\left(\left(1+\left\Vert (1-s)x_{k,t}+sx_{k}\right\Vert ^{\ell_{H}}+\left\Vert x_{k,t}\right\Vert ^{\ell_{H}}\right)\left(\sum_{i=1}^{N}L_{Hi}s^{\alpha_{Hi}}\left\Vert x_{k}-x_{k,t}\right\Vert ^{\alpha_{Hi}}\right)\right)^{2}\left\Vert x_{k}-x_{k,t}\right\Vert ^{2}|\hat{X}_{t}=x\right]ds\\
 & \leq\mathrm{\mathbb{E}}\int_{0}^{1}\mathrm{\mathbb{E}}\left[2\left(\left(1+\left\Vert x_{k}\right\Vert ^{\ell_{H}}+\left\Vert x_{k,t}\right\Vert ^{\ell_{H}}\right)\left(\sum_{i=1}^{N}L_{Hi}s^{\alpha_{Hi}}\left\Vert x_{k}-x_{k,t}\right\Vert ^{\alpha_{Hi}}\right)\right)^{2}\left\Vert x_{k}-x_{k,t}\right\Vert ^{2}|\hat{X}_{t}=x\right]ds\\
 & \leq12N\mathrm{\mathbb{E}}\int_{0}^{1}\mathrm{\mathbb{E}}\left[\left(1+\left\Vert x_{k}\right\Vert ^{2\ell_{H}}+\left\Vert x_{k,t}\right\Vert ^{2\ell_{H}}\right)\left(\sum_{i=1}^{N}L_{Hi}^{2}s^{2\alpha_{Hi}}\left\Vert x_{k}-x_{k,t}\right\Vert ^{2\alpha_{Hi}}\right)\left\Vert x_{k}-x_{k,t}\right\Vert ^{2}|\hat{X}_{t}=x\right]ds\\
 & \leq12N\mathrm{\mathbb{E}}\mathrm{\mathbb{E}}\left[\left(1+\left\Vert x_{k}\right\Vert ^{2\ell_{H}}+\left\Vert x_{k,t}\right\Vert ^{2\ell_{H}}\right)\sum_{i=1}^{N}\frac{L_{Hi}^{2}}{2\alpha_{Hi}+1}\left\Vert x_{k}-x_{k,t}\right\Vert ^{2+2\alpha_{Hi}}|\hat{X}_{t}=x\right]\\
 & \leq12N\frac{L_{H}^{2}}{2\alpha_{H}+1}\mathrm{\mathbb{E}}\left[\left(1+\left\Vert x_{k}\right\Vert ^{2\ell_{H}}+\left\Vert x_{k,t}\right\Vert ^{2\ell_{H}}\right)\sum_{i=1}^{N}\left\Vert x_{k}-x_{k,t}\right\Vert ^{2\alpha_{Hi}+2}\right]\\
 & \leq12N\frac{L_{H}^{2}}{2\alpha_{H}+1}\mathrm{\mathbb{E}}^{\frac{1}{2}}\left[\left(1+\left\Vert x_{k}\right\Vert ^{2\ell_{H}}+\left\Vert x_{k,t}\right\Vert ^{2\ell_{H}}\right)^{2}\right]\mathrm{\mathbb{E}}^{\frac{1}{2}}\left[\left(\sum_{i=1}^{N}\left\Vert x_{k}-x_{k,t}\right\Vert ^{2\alpha_{Hi}+2}\right)^{2}\right]\\
 & \leq C_{H1}\mathrm{\mathbb{E}}^{\frac{1}{2}}\left[1+\left\Vert x_{k}\right\Vert ^{4\ell_{H}}+\left\Vert x_{k,t}\right\Vert ^{4\ell_{H}}\right]\mathrm{\mathbb{E}}^{\frac{1}{2}}\left[\sum_{i=1}^{N}\left\Vert x_{k}-x_{k,t}\right\Vert ^{4\alpha_{Hi}+4}\right]\\
 & \leq O\left(d^{\frac{\lceil\frac{4\ell_{H}}{\beta}\rceil}{2}}\right)O\left(d^{\frac{\lceil\frac{\left(4\alpha_{HN}+4\right)}{\beta}\rceil}{2}}\right)\eta^{\alpha_{Hi}+1}\\
 & \leq O\left(d^{\frac{\lceil\frac{4\ell_{H}}{\beta}\rceil+\lceil\frac{\left(4\alpha_{HN}+4\right)}{\beta}}{2}}\right)\eta^{\alpha_{H}+1}.
\end{align*}
\par\end{center}
\begin{align*}
\mathbb{E}\left[\left\Vert \nabla U(x_{k,t})-\nabla U(x_{k})\right\Vert ^{2}\right] & =\mathbb{E}\left[\left\Vert \nabla U(x_{k,t})-\mathbb{E}\left[\nabla U(x_{k})|x_{k,t}\right]\right\Vert ^{2}\right]\\
 & =\mathbb{E}\left[\left\Vert \nabla^{2}U(x)\mathbb{E}\left[x_{k}-x_{k,t}|x_{k,t}=x\right]+\hat{r}_{t}(x)\right\Vert ^{2}\right]\\
 & \leq2\mathbb{E}\left[\left\Vert \nabla^{2}U(x)\left(I_{1}+I_{2}+I_{3}\right)\right\Vert ^{2}\right]+2\mathbb{E}\left\Vert \hat{r}_{t}(x)\right\Vert ^{2}\\
 & \leq6L_{G}^{2}\mathbb{E}\left[\left\Vert I_{1}\right\Vert ^{2}+\left\Vert I_{2}\right\Vert ^{2}+\left\Vert I_{3}\right\Vert ^{2}\right]+2\mathbb{E}\left\Vert \hat{r}_{t}(x)\right\Vert ^{2}\\
 & \leq6L_{G}^{2}\eta^{2}\int p_{k}(x_{k})\left\Vert \nabla\log p_{k}(x_{k})\right\Vert ^{2}dx+3L_{G}^{2}d\eta^{3}\\
 & +O\left(d^{\left(\lceil\frac{2}{\beta}\rceil+1\right)}\right)\eta^{2}+O\left(d^{\frac{\lceil\frac{4\ell_{H}}{\beta}\rceil+\lceil\frac{\left(4\alpha_{HN}+4\right)}{\beta}}{2}}\right)\eta^{\alpha_{H}+1},
\end{align*}
From \citep{mou2022improved} Lemma 7, we have
\[
\int p_{k}(x_{k})\left\Vert \nabla\log p_{k}(x_{k})\right\Vert ^{2}dx\leq8\int p_{t}(x_{k,t})\left\Vert \nabla\log p_{t}(x_{k,t})\right\Vert ^{2}dx+32\eta^{2}d^{2}L_{2}^{2}
\]
and from \citep{chewi2021analysis} Lemma 16, it holds that
\[
\int p_{t}(x_{k,t})\left\Vert \nabla U(x_{k,t})\right\Vert ^{2}\leq I\left(p_{k,t}|\nu\right)+2dL_{G}.
\]
Combining the above inequalities and Young inequality
\begin{align*}
\int p_{t}(x_{k,t})\left\Vert \nabla\log p_{t}(x_{k,t})\right\Vert ^{2}dx & \leq2\int p_{t}(x_{k,t})\left\Vert \nabla\log\frac{p_{t}(x_{k,t})}{\nu}\right\Vert ^{2}dx+2\int p_{t}(x_{k,t})\left\Vert \nabla U(x_{k,t})\right\Vert ^{2}\\
 & \leq4I\left(p_{k,t}|\nu\right)+2dL_{G},
\end{align*}
which implies
\[
\mathbb{E}\left[\left\Vert \nabla U(x_{k,t})-\nabla U(x_{k})\right\Vert ^{2}\right]\leq24L_{G}^{2}\eta^{2}I\left(p_{k,t}|\nu\right)+12d\eta^{2}L_{G}^{3}+O\left(d^{\frac{\lceil\frac{4\ell_{H}}{\beta}\rceil+\lceil\frac{\left(4\alpha_{HN}+4\right)}{\beta}\rceil}{2}}\right)\eta^{\alpha_{H}+1}.
\]
Therefore, from \citep{vempala2019rapid} Lemma 3, the time derivative
of KL divergence along ULA is bounded by
\begin{align*}
\frac{d}{dt}H\left(p_{k,t}|\pi\right) & \leq-\frac{1}{2}I\left(p_{k,t}|\pi\right)+D\eta^{\alpha_{H}+1},
\end{align*}
where in the last inequality, we have used the definitions of $D=O\left(d^{\frac{\lceil\frac{4\ell_{H}}{\beta}\rceil+\lceil\frac{\left(4\alpha_{HN}+4\right)}{\beta}\rceil}{2}}\right)\eta^{\alpha_{H}+1}$.
\end{proof}
\subsection{Proof of Lemma \ref{lem:D2}\label{AD2-1}}
\begin{proof} From Theorem \eqref{thm:C4} we have
\[
W_{2}^{2}(\mu,\nu)\leq2M_{4}^{\frac{1}{2}}\left(\mu+\nu\right)\sqrt{\frac{1}{\gamma}}\sqrt{I\left(\mu|\nu\right)}.
\]
On the other hand,$W_{2}$ can be bounded directly again from \citep{villani2008optimal}
as
\begin{align*}smoothing
W_{2}(\mu,\nu) & \leq\left(2\int_{\mathbb{R}^{d}}\left\Vert x\right\Vert ^{2}\left|\mu(x)-\nu(x)\right|dx\right)^{\frac{1}{2}}\\
 & \leq\left(2\int_{\mathbb{R}^{d}}\left\Vert x\right\Vert ^{2}\left(\mu(x)+\nu(x)\right)dx\right)^{\frac{1}{2}}\\
 & \leq\sqrt{2}\sqrt{M_{2}\left(\mu+\nu\right)}.
\end{align*}
Since $\nu$ is Lipschitz gradient, from HWI inequality for any $s\geq4,$
we have
\begin{align*}
H\left(\mu|\nu\right) & \leq W_{2}(\mu,\nu)\sqrt{I\left(\mu|\nu\right)}+L_{G}W_{2}^{2}(\mu,\nu)\\
 & \leq\sqrt{2}M_{2}^{\frac{1}{2}}\left(\mu+\nu\right)\sqrt{I\left(\mu|\nu\right)}+2L_{G}M_{4}^{\frac{1}{2}}\left(\mu+\nu\right)\sqrt{\frac{1}{\gamma}}\sqrt{I\left(\mu|\nu\right)}\\
 & \leq\left(\sqrt{2}M_{2}^{\frac{1}{2}}\left(\mu+\nu\right)+2L_{G}M_{4}^{\frac{1}{2}}\left(\mu+\nu\right)\sqrt{\frac{1}{\gamma}}\right)\sqrt{I\left(\mu|\nu\right)}\\
 & \leq\left(\sqrt{2}+2L_{G}\sqrt{\frac{1}{\gamma}}\right)M_{4}^{\frac{1}{2}}\left(\mu+\nu\right)\sqrt{I\left(\mu|\nu\right)},
\end{align*}
where in the last step, we have used Jensen inequality for $s\geq4$.
This gives us the desired result. \end{proof}

\subsection{Proof of Lemma \ref{lem:D4}\label{AD4}}

\begin{proof}Integrating both sides of equation
from $t=0$ to $t=\eta$ we obtain
\begin{align*}
 & H(p_{k+1}|\nu)-H(p_{k}|\nu)\leq D\eta^{2+\alpha_{H}},
\end{align*}
where the inequality holds since the first term is negative. Using
discrete Grönwall inequality, we have, \textit{for any} $k\in\mathbb{N}$
\begin{align*}
H(p_{K}|\nu) & \leq H(p_{k_{0}}|\nu)+KD\eta^{2+\alpha_{H}}\\
 & \leq H(p_{k_{0}}|\nu)+TD\eta^{1+\alpha_{H}}\\
 & \leq H(p_{k_{0}}|\nu)+\frac{\epsilon}{2}.
\end{align*}
If there exists some $k<K$ such that $H(p_{k}|\nu)\leq\frac{\epsilon}{2}$
then we can choose $\eta\leq\left(\frac{\epsilon}{2TD}\right)^{\frac{1}{1+\alpha_{H}}}$
so that $H(p_{K}|\nu)\leq\epsilon$. If there is no such $k$, we
will prove for sufficiently large $K$, $H(p_{K}|\nu)\leq\epsilon$.
Let $A=\frac{3C}{8\left(M_{4}(\tilde{p}_{k,t}+\pi)\right)}\left(\frac{\epsilon}{2}\right)$,
the above expression leads to
\begin{align*}
H(p_{k+1}|\nu) & \leq H(p_{k}|\nu)\left(1-A\eta\right)+D\eta^{\alpha_{H}+2}.
\end{align*}
By iterating the process we get
\begin{align*}
H(p_{k}|\nu) & \leq H(p_{0}|\nu)\left(1-A\eta\right)^{k}+\frac{D}{A}\eta^{\alpha_{H}+1}.
\end{align*}
To get $H(p_{K}|\nu)\leq\epsilon$, for $\eta$ small enough so that
$\eta\leq\left(\frac{A\epsilon}{2D}\right)^{\frac{1}{\alpha_{H}+1}}$,
it suffices to run $K$ iterations such that
\[
\left(1-A\eta\right)^{K}\leq\frac{\epsilon}{2H(p_{0}|\nu)}.
\]
As a result, we obtain
\begin{align*}
K & =\log_{\left(1-A\eta\right)}\left(\frac{\epsilon}{\left(2H(p_{0}|\nu)\right)}\right)\\
 & =\frac{\ln\left(\frac{\left(H(p_{0}|\nu)\right)}{\epsilon}\right)}{\ln\left(\frac{1}{1-A\eta}\right)}\\
 & \leq\frac{\ln\left(\frac{\left(H(p_{0}|\nu)\right)}{\epsilon}\right)}{\frac{3C}{8\left(M_{4}(\tilde{p}_{k,t}+\pi)\right)}\left(\frac{\epsilon}{2}\right)\eta}.
\end{align*}
By plugging $T=K\eta$ and assuming without loss of generality that
$T>1$ (since we can choose $T$), we obtain
\begin{align*}
T & \leq\frac{\ln\left(\frac{\left(H(p_{0}|\nu)\right)}{\epsilon}\right)}{\frac{3C}{8\left(M_{4}(\tilde{p}_{k,t}+\pi)\right)}\left(\frac{\epsilon}{2}\right)}
\end{align*}
which is satisfied if we choose
\[
T=O\left(\frac{\left(M_{4}(\tilde{p}_{k,t}+\pi)\right)\ln\left(\frac{\left(H(p_{0}|\nu)\right)}{\epsilon}\right)}{\epsilon}\right).
\]
Without loss of generality, since $H(p_{0}|\nu)=O\left(d\right),$
we can assume that $H(p_{0}|\nu)\geq1>\epsilon.$ We have $\ln\left(\frac{\left(H(p_{0}|\nu)\right)}{\epsilon}\right)>1$.
Therefore,
\[
\eta=\min\left\{ 1,\left(\frac{A\epsilon}{2D}\right)^{\frac{1}{\alpha_{H}+1}},\left(\frac{\epsilon}{2TD}\right)^{\frac{1}{\alpha_{H}+1}}\right\} =\left(\frac{\epsilon}{2TD}\right)^{\frac{1}{\alpha_{H}+1}}
\]
Using $K=\frac{T}{\eta}$, we have
\begin{align*}
K & \leq O\left(\left(\frac{2TD}{\epsilon}\right)^{\frac{1}{\alpha_{H}+1}}\frac{M_{4}(\tilde{p}_{k,t}+\pi)\ln\left(\frac{\left(H(p_{0}|\pi)\right)}{\epsilon}\right)}{\epsilon}\right)\\
 & \leq O\left(\frac{D^{\frac{1}{\alpha_{H}+1}}d^{\lceil\frac{4}{\beta}\rceil\left(1+\frac{1}{\alpha_{H}+1}\right)}\ln^{\left(1+\frac{1}{\alpha_{H}+1}\right)}\left(\frac{\left(H(p_{0}|\pi)\right)}{\epsilon}\right)}{\epsilon^{\left(1+\frac{1}{\alpha_{H}+1}\right)+\frac{1}{\alpha_{H}+1}}}\right).
\end{align*}

Combining with these above results for $\eta$ small enough, $M_{4}(\tilde{p}_{k,t}+\pi)=O\left(d^{\lceil\frac{4}{\beta}\rceil}\right)$
and $D=O\left(d^{\frac{\lceil\frac{4\ell_{H}}{\beta}\rceil+\lceil\frac{\left(4\alpha_{HN}+4\right)}{\beta}\rceil}{2}}\right)$,smoothing
we obtain

\[
K=O\left(\frac{d^{\frac{\lceil\frac{4\ell_{H}}{\beta}\rceil+\lceil\frac{\left(4\alpha_{HN}+4\right)}{\beta}\rceil}{2\left(\alpha_{H}+1\right)}+\lceil\frac{4}{\beta}\rceil\left(1+\frac{1}{\alpha_{H}+1}\right)}\ln^{\left(1+\frac{1}{\alpha_{H}+1}\right)}\left(\frac{\left(H(p_{0}|\nu)\right)}{\epsilon}\right)}{\epsilon^{\left(1+\frac{2}{\alpha_{H}+1}\right)}}\right)
\]

which is our desired result. \end{proof} \begin{remark} We can get
a tighter result for each specific case. For example, by choosing
$\ell_{H}=0,$ $\alpha_{HN}=\alpha_{H}=1$, $\beta\simeq2$ we obtain
\[
K\approx\tilde{O}\left(\frac{d^{5}}{\epsilon^{2}}\right),
\]
a weaker but rather comparable to the result of \citep{erdogdu2020convergence}.
\end{remark}
\section{Proof of ULA algorithm via smoothing potential \label{AppC-1}}
\subsection{Proof of Lemma \ref{lem:D2}\label{AD2}}

\begin{proof} Since $\left|U_{\mu}-U\right|\leq L\mu^{1+\alpha}d^{\frac{1+\alpha}{2\wedge p}}$
and $U$ satisfies Poincaré inequality with constant $\gamma$, by
\citep{ledoux2001logarithmic}'s Lemma 1.2, $U_{\mu}$ satisfies Poincaré
inequality with constant $\gamma_{1}=\gamma e^{-4L\mu^{1+\alpha}d^{\frac{1+\alpha}{2\wedge p}}}$.
From Theorem \eqref{thm:C4} we have
\[
W_{2}^{2}(p,\pi_{\mu})\leq2M_{4}^{\frac{1}{2}}\left(p+\pi_{\mu}\right)\sqrt{\frac{1}{\gamma}}\sqrt{I\left(p|\pi_{\mu}\right)}.
\]
On the other hand,$W_{2}$ can be bounded directly again from \citep{villani2008optimal}
as
\begin{align*}
W_{2}(p,\pi_{\mu}) & \leq\left(2\int_{\mathbb{R}^{d}}\left\Vert x\right\Vert ^{2}\left|p(x)-\pi_{\mu}(x)\right|dx\right)^{\frac{1}{2}}\\
 & \leq\left(2\int_{\mathbb{R}^{d}}\left\Vert x\right\Vert ^{2}\left(p(x)+\pi_{\mu}(x)\right)dx\right)^{\frac{1}{2}}\\
 & \leq\sqrt{2}\sqrt{M_{2}\left(p+\pi_{\mu}\right)}.
\end{align*}
Since $\pi_{\mu}$ is $\frac{NL\mu^{1+\alpha}}{(1+\alpha)}d^{\frac{2}{p}\vee2}$-Lipschitz
gradient, from HWI inequality for any $s\geq4,$ we have
\begin{align*}
H\left(p|\pi_{\mu}\right) & \leq W_{2}(p,\pi_{\mu})\sqrt{I\left(p|\pi_{\mu}\right)}+\frac{NL\mu^{1+\alpha}}{(1+\alpha)}d^{\frac{2}{p}\vee2}W_{2}^{2}(p,\pi_{\mu})\\
 & \leq\sqrt{2}M_{2}^{\frac{1}{2}}\left(p+\pi_{\mu}\right)\sqrt{I\left(p|\pi_{\mu}\right)}+2\frac{NL\mu^{1+\alpha}}{(1+\alpha)}d^{\frac{2}{p}\vee2}M_{4}^{\frac{1}{2}}\left(p+\pi_{\mu}\right)\sqrt{\frac{1}{\gamma}}\sqrt{I\left(p|\pi_{\mu}\right)}\\
 & \leq\sqrt{2}M_{2}^{\frac{1}{2}}\left(p+\pi_{\mu}\right)\sqrt{I\left(p|\pi_{\mu}\right)}+2\frac{NL\mu^{1+\alpha}}{(1+\alpha)}d^{\frac{2}{p}\vee2}M_{4}^{\frac{1}{2}}\left(p+\pi_{\mu}\right)\sqrt{\frac{1}{\gamma}}\sqrt{I\left(p|\pi_{\mu}\right)}\\
 & \leq\left(\sqrt{2}M_{2}^{\frac{1}{2}}\left(p+\pi_{\mu}\right)+2\frac{NL\mu^{1+\alpha}}{(1+\alpha)}d^{\frac{2}{p}\vee2}M_{4}^{\frac{1}{2}}\left(p+\pi_{\mu}\right)\sqrt{\frac{1}{\gamma_{1}}}\right)\sqrt{I\left(p|\pi_{\mu}\right)}\\
 & \leq\left(\sqrt{2}+2\frac{NL\mu^{1+\alpha}}{(1+\alpha)}d^{\frac{2}{p}\vee2}\sqrt{\frac{1}{\gamma_{1}}}\right)M_{4}^{\frac{1}{2}}\left(p+\pi_{\mu}\right)\sqrt{I\left(p|\pi_{\mu}\right)}\\
 & \leq\left(\sqrt{2}+2\frac{NL\mu^{1+\alpha}}{(1+\alpha)}d^{\frac{2}{p}\vee2}\sqrt{\frac{1}{\gamma_{1}}}\right)M_{s}^{\frac{2}{s}}\left(p+\pi_{\mu}\right)\sqrt{I\left(p|\pi_{\mu}\right)},
\end{align*}
where in the last step, we have used Jensen inequality for $s\geq4$.
This gives us the desired result. \end{proof}

\subsection{Proof of Lemma \ref{lem:D3}\label{AD3}}

\begin{proof} We provide the proof for completeness. Recall that
by definition of $U_{\mu}$, we have $\nabla U_{\mu}(x)=\mathrm{\mathrm{\mathbb{E}}}_{\zeta}[U(x+\mu\mathrm{\zeta})]$,
where $\mathrm{\zeta}\sim N_{p}(0,I_{d})$. For $\zeta_{1}\sim N_{p}(0,I_{d})$
and it is independent of $\zeta$, clearly, $\mathbb{E}{}_{\mathrm{\mathrm{\zeta_{1}}}}[g_{\mu}(x,\zeta_{1})]=\mathbb{E}{}_{\mathrm{\mathrm{\zeta_{1}}}}\nabla U(x+\mu\zeta_{1})=\nabla\mathbb{E}{}_{\mathrm{\mathrm{\zeta_{1}}}}U(x+\mu\zeta_{1})=\nabla U_{\mu}(x)$
by exchange gradient and expectation and the definition of $U_{\mu}(x)$.
We now proceed to bound the variance of $g_{\mu}(x,\zeta_{1})$. We
have:
\begin{align*}
~ & \mathrm{\mathbb{E}}_{\mathrm{\zeta_{1}}}[\Vert\nabla U_{\mu}(x)-g_{\mu}(x,\zeta_{1})\Vert_{2}^{2}]\\
~ & \leq\mathrm{\mathbb{E}}_{\zeta_{1},\mathrm{\zeta}}[\Vert\nabla U(x+\mu\mathrm{\zeta})-\nabla U(x+\mu\mathrm{\zeta_{1}})\Vert^{2}].\\
 & \leq N\sum_{i}L_{i}^{2}\mathrm{\mathbb{E}}_{\mathrm{\zeta_{1}},\mathrm{\zeta}}[\Vert\mu(\mathrm{\zeta}-\mathrm{\zeta_{1}})\Vert^{2\alpha_{i}}\\
 & \leq N\sum_{i}L_{i}^{2}\mu^{2\alpha_{i}}\mathrm{\mathbb{E}}_{\zeta_{1},\mathrm{\zeta}}[\Vert\mathrm{\zeta}-\mathrm{\zeta_{1}}\Vert^{2\alpha_{i}}]\\
 & \leq2N\sum_{i}L_{i}^{2}\mu^{2\alpha_{i}}\left(\mathrm{\mathbb{E}}\left[\Vert\mathrm{\zeta}\Vert^{2\alpha_{i}}\right]+\mathrm{\mathbb{E}}\left[\Vert\mathrm{\zeta_{1}}\Vert^{2\alpha_{i}}\right]\right)\\
 & \leq2N\sum_{i}L_{i}^{2}\mu^{2\alpha_{i}}\left(\left(\mathrm{\mathbb{E}}\left[\Vert\mathrm{\zeta}\Vert^{2}\right]\right)^{\alpha_{i}}+\left(\mathrm{\mathbb{E}}\left[\Vert\zeta_{1}\Vert^{2}\right]\right)^{\alpha_{i}}\right)\\
 & \leq4N\sum_{i}L_{i}^{2}\mu^{2\alpha_{i}}d^{\frac{2\alpha_{i}}{p}}\\
 & \leq4N^{2}L^{2}\mu^{2\alpha}d^{\frac{2\alpha}{p}},
\end{align*}
as claimed. \end{proof}

\subsection{Proof of Lemma \ref{lem:D5}\label{AD5}}

\begin{proof} First of all, we have
\begin{align*}
 & \mathrm{\mathbb{E}}_{p_{\mu,kt\zeta}}\left\Vert \nabla U_{\mu}(x_{\mu,k})-\nabla U_{\mu}(x_{\mu,k,t})\right\Vert ^{2}\\
 & \stackrel{_{1}}{\leq}3\mathrm{\mathbb{E}}_{p_{\mu,kt\zeta}}\left[\left\Vert \nabla U_{\mu}(x_{\mu,k})-\nabla U(x_{\mu,k})\right\Vert ^{2}+\left\Vert \nabla U(x_{\mu,k})-\nabla U(x_{\mu,k,t})\right\Vert ^{2}+\left\Vert \nabla U(x_{\mu,k,t})-\nabla U_{\mu}(x_{\mu,k,t})\right\Vert ^{2}\right]\\
 & \stackrel{_{2}}{\leq}3NL^{2}\sum_{i}\mathrm{\mathbb{E}}_{p_{\mu,kt\zeta}}\left\Vert x_{\mu,k,t}-x_{\mu,k}\right\Vert ^{2\alpha_{i}}+6\left(\frac{NL\mu^{1+\alpha}}{(1+\alpha)}d^{\frac{3}{p}\lor\frac{5}{2}}\right)^{2}\\
 & \leq3NL^{2}\sum_{i}\mathrm{\mathbb{E}}_{p_{\mu,k\zeta}}\left\Vert -tg(x_{\mu,k},\zeta)+\sqrt{2t}z_{\mu,k}\right\Vert ^{2\alpha_{i}}+6\left(\frac{NL\mu^{1+\alpha}}{(1+\alpha)}d^{\frac{3}{p}\lor\frac{5}{2}}\right)^{2}\\
 & \stackrel{_{3}}{\leq}3NL^{2}\sum_{i}\left(2\eta^{2\alpha_{i}}\mathrm{\mathbb{E}}_{p_{\mu,k\zeta}}\left[\left\Vert \nabla U(x_{\mu,k})\right\Vert +\left\Vert \nabla U_{\mu}(x_{\mu,k})-\nabla U(x_{\mu,k})\right\Vert +\left\Vert \nabla U_{\mu}(x_{\mu,k})-g(x_{\mu,k},\zeta)\right\Vert \right]^{2\alpha_{i}}\right)\\
 & +3NL^{2}\sum_{i}4\eta^{\alpha_{i}}d{}^{\alpha_{i}}+6\left(\frac{NL\mu^{1+\alpha}}{(1+\alpha)}d^{\frac{3}{p}\lor\frac{5}{2}}\right)^{2}\\
 & \stackrel{_{4}}{\leq}3NL^{2}\sum_{i}\left(6\eta^{2\alpha_{i}}\mathrm{\mathbb{E}}_{p_{\mu,k\zeta}}\left[\left\Vert \nabla U(x_{\mu,k})\right\Vert ^{2\alpha_{i}}+\left\Vert \nabla U_{\mu}(x_{\mu,k})-\nabla U(x_{\mu,k})\right\Vert ^{2\alpha_{i}}+\left\Vert \nabla U_{\mu}(x_{\mu,k})-g(x_{\mu,k},\zeta)\right\Vert ^{2\alpha_{i}}\right]\right)\\
 & +3NL^{2}\sum_{i}4\eta^{\alpha_{i}}d{}^{\alpha_{i}}+6\left(\frac{NL\mu^{1+\alpha}}{(1+\alpha)}d^{\frac{3}{p}\lor\frac{5}{2}}\right)^{2}\\
 & \leq3NL^{2}\sum_{i}\left(6\eta^{2\alpha_{i}}\left[\mathrm{\mathbb{E}}_{p_{\mu,k}}\left\Vert \nabla U(x_{\mu,k})\right\Vert ^{2\alpha_{i}}+\left(\frac{NL\mu^{1+\alpha}}{(1+\alpha)}d^{\frac{3}{p}\lor\frac{5}{2}}\right)^{2\alpha_{i}}+\left(\mathrm{\mathbb{E}}_{p_{\mu,k\zeta}}\left\Vert \nabla U_{\mu}(x_{\mu,k})-g(x_{\mu,k},\zeta)\right\Vert ^{2}\right)^{\alpha_{i}}\right]\right)\\
 & +3NL^{2}\sum_{i}4\eta^{\alpha_{i}}d{}^{\alpha_{i}}+6\left(\frac{NL\mu^{1+\alpha}}{(1+\alpha)}d^{\frac{3}{p}\lor\frac{5}{2}}\right)^{2}\\
 & \stackrel{_{5}}{\leq}3NL^{2}\sum_{i}\left(6\eta^{2\alpha_{i}}\left[O\left(d^{\lceil\frac{\left(\ell_{G}+\alpha_{GN}\right)2\alpha_{i}}{\beta}\rceil}\right)+\left(\frac{NL\mu^{1+\alpha}}{(1+\alpha)}d^{\frac{3}{p}\lor\frac{5}{2}}\right)^{2\alpha_{i}}+\left(8N^{2}L^{2}\mu^{2\alpha}d^{\frac{2\alpha}{p}}\right)^{\alpha_{i}}\right]\right)\\
 & +3NL^{2}\sum_{i}4\eta^{\alpha_{i}}d{}^{\alpha_{i}}+6\left(\frac{NL\mu^{1+\alpha}}{(1+\alpha)}d^{\frac{3}{p}\lor\frac{5}{2}}\right)^{2}\\
 & \leq\left(O\left(d^{\lceil\frac{\alpha_{GN}2\alpha_{i}}{\beta}\rceil}\right)+6\left(\frac{NL\mu}{(1+\alpha)}d^{\frac{3}{p}\lor\frac{5}{2}}\right)^{2}\right)\eta^{\alpha}\\
 & +3NL^{2}\sum_{i}\left(6\left[\left(\frac{NL\mu^{1+\alpha}}{(1+\alpha)}d^{\frac{3}{p}\lor\frac{5}{2}}\right)^{2\alpha_{i}}+\left(8N^{2}L^{2}d^{\frac{2\alpha}{p}}\right)^{\alpha_{i}}\right]+4d{}^{\alpha_{i}}\right)\eta^{\alpha},
\end{align*}
where step $1$ follows from Assumption \ref{A0}, step $2$ comes
from Young inequality and triangle inequality, step $3$ comes from
triangle inequality and normal distribution, step 4 is due to Young
inequality, step 5 comes from Lemma \ref{lem:D3}, and in the last
step, we have used Lemma \ref{lem:D4}. On the other hand, by choosing
$\mu=\sqrt{\eta}$, we also have
\begin{align*}
 & \mathrm{\mathbb{E}}_{p_{kt\zeta}}\left\Vert \nabla U_{\mu}(x_{\mu,k,t})-g(x_{\mu,k},\zeta)\right\Vert ^{2}\\
 & \stackrel{_{1}}{\leq}2\left[\mathrm{\mathbb{E}}_{p_{kt\zeta}}\left\Vert \nabla U_{\mu}(x_{\mu,k,t})-\nabla U_{\mu}(x_{\mu,k})\right\Vert ^{2}+\left\Vert \nabla U_{\mu}(x_{\mu,k})-g(x_{k},\zeta)\right\Vert ^{2}\right]\\
 & \stackrel{_{2}}{\leq}2\mathrm{\mathbb{E}}_{p_{kt\zeta}}\left\Vert \nabla U_{\mu}(x_{\mu,k,t})-\nabla U_{\mu}(x_{\mu,k})\right\Vert ^{2}+8N^{2}L^{2}\mu^{2\alpha}d^{\frac{2\alpha}{p}}\\
 & \leq\left(O\left(d^{\lceil\frac{\left(\ell_{G}+\alpha_{GN}\right)2\alpha_{i}}{\beta}\rceil}\right)+12\left(\frac{NL\mu}{(1+\alpha)}d^{\frac{3}{p}\lor\frac{5}{2}}\right)^{2}\right)\eta^{\alpha}\\
 & +6NL^{2}\sum_{i}\left(6\left[\left(\frac{NL\mu^{1+\alpha}}{(1+\alpha)}d^{\frac{3}{p}\lor\frac{5}{2}}\right)^{2\alpha_{i}}+\left(8N^{2}L^{2}d^{\frac{2\alpha}{p}}\right)^{\alpha_{i}}\right]+4d{}^{\alpha_{i}}\right)\eta^{\alpha}+8N^{2}L^{2}d^{\frac{2\alpha}{p}}\eta^{\alpha}\\
 & \stackrel{}{\leq}O\left(d^{\lceil\frac{2\alpha_{GN}^{2}}{\beta}\rceil}\right)\eta^{\alpha},
\end{align*}
where step $1$ follows from Young inequality, step $2$ is because
of Lemma \ref{lem:D3} and $\eta\leq1$, and the last step comes from
$\eta$ small enough. Therefore, from Lemma \ref{thm:D6}, the time
derivative of KL divergence along ULA is bounded by
\begin{align*}
\frac{d}{dt}H\left(p_{\mu,k,t}|\pi_{\mu}\right) & \leq-\frac{3}{4}I\left(p_{\mu,k,t}|\pi_{\mu}\right)+\mathrm{\mathbb{E}}_{p_{kt\zeta}}\left\Vert \nabla U(x_{\mu,k,t})-g(x_{\mu,k},\zeta)\right\Vert ^{2}\\
 & \leq-\frac{3}{4}I\left(p_{\mu,k,t}|\pi_{\mu}\right)+D_{\mu}\eta^{\alpha_{G}},
\end{align*}
where $D_{\mu}=O\left(d^{\lceil\frac{2\alpha_{GN}^{2}}{\beta}\rceil}\right)$,
as desired. \end{proof}

\subsection{Proof of Theorem \ref{thm:D6}\label{AD6}}

\begin{proof} Follow the same steps as in Theorem \ref{thm:C4},
we will get $H(p_{K}|\pi)\leq\epsilon$ after
\[
K=O\left(\frac{\gamma^{1+\frac{1}{\alpha_{G}}}d^{\lceil\frac{2\alpha_{GN}^{2}}{\beta}\rceil\frac{1}{\alpha_{G}}+\lceil\frac{\alpha_{GN}+2}{\beta}\rceil\left(\alpha_{GN}+2\right)\left(1+\frac{1}{\alpha_{G}}\right)}\ln^{\left(1+\frac{1}{\alpha_{G}}\right)}\left(\frac{\left(H(p_{0}|\pi)\right)}{\epsilon}\right)}{\epsilon^{\left(\alpha_{GN}+1\right)\left(1+\frac{1}{\alpha_{G}}\right)+\frac{1}{\alpha_{G}}}}\right).
\]
By replacing $\delta_{1}=\frac{1}{2}$ and $\delta_{2}=\frac{2}{s}$
for $s>4$, we have
\[
K\approx\tilde{O}\left(\frac{{\displaystyle d^{\lceil\frac{2\alpha_{GN}^{2}}{\beta}\rceil\frac{1}{\alpha_{G}}+\lceil\frac{\alpha_{GN}+2}{\beta}\rceil\left(\alpha_{GN}+2\right)\left(1+\frac{1}{\alpha_{G}}\right)}}}{\epsilon^{\left(\alpha_{GN}+1\right)\left(1+\frac{1}{\alpha_{G}}\right)+\frac{1}{\alpha_{G}}}}\right).
\]
From \citep{nguyen2021unadjusted}'s Lemma 3.4, by choosing $\mu=\sqrt{\eta}$
small enough so that $W_{\beta}(\pi,\ \pi_{\mu})\leq3\sqrt{NLE_{2}}\eta^{\frac{\alpha}{2}}d^{\frac{1}{p}}\leq\frac{\epsilon}{2}$.
Since $\pi$ satisfies Poincaré inequality, by triangle inequality
we also get
\begin{align*}
W_{\beta}(p_{k},\ \pi) & \leq W_{\beta}(p_{k},\ \pi_{\mu})+W_{\beta}(\pi,\ \pi_{\mu})\\
 & \leq2\inf_{\tau}\left[\tau\left(1.5+\log\int e^{\tau\left\Vert x\right\Vert ^{\beta}}\pi(x)dx\right)\right]^{\frac{1}{\beta}}\left(H(p_{k}|\pi)^{\frac{1}{\beta}}+H(p_{k}|\pi)^{\frac{1}{2\beta}}\right)+W_{\beta}(\pi,\ \pi_{\mu})\\
 & \leq2\left[\frac{a}{4\beta}\left(1.5+\tilde{d}+\tilde{\mu}\right)\right]^{\frac{1}{\beta}}\left(H(p_{k}|\pi)^{\frac{1}{\beta}}+H(p_{k}|\pi)^{\frac{1}{2\beta}}\right)+3\sqrt{NLE_{2}}\eta^{\frac{\alpha}{2}}d^{\frac{1}{p}}
\end{align*}
To have $W_{\beta}(p_{K},\ \pi)\leq\epsilon$, it is sufficient to
choose $H(p_{k}|\pi)^{\frac{1}{2\beta}}=\tilde{O}\left(\epsilon d^{\frac{-1}{\beta}}\right)$,
which in turn implies $H(p_{k}|\pi)=\tilde{O}\left(\epsilon^{2\beta}d^{-2}\right).$
By replacing this in the bound above, we obtain the number of iteration
for $L_{\beta}$-Wasserstein distance is $\tilde{O}\left(\frac{{\displaystyle d^{\frac{2}{\beta}\left(\lceil\frac{2\alpha_{GN}^{2}}{\beta}\rceil\frac{1}{\alpha_{G}}+\lceil\frac{\alpha_{GN}+2}{\beta}\rceil\left(\alpha_{GN}+2\right)\left(1+\frac{1}{\alpha_{G}}\right)\right)+2+\frac{4}{\alpha}}}}{\gamma_{1}^{\left(1+\frac{1}{\alpha_{G}}\right)}\epsilon^{\left(\alpha_{GN}+1\right)\left(1+\frac{1}{\alpha_{G}}\right)+\frac{1}{\alpha_{G}}}}\right)$
where $\gamma_{1}=\gamma e^{-4L\mu^{1+\alpha}d^{\frac{1+\alpha}{2\wedge p}}}.$
Given $\epsilon>0$, if we further assume
\[
\eta=\min\left\{ 1,\left(\frac{\epsilon}{2TD_{\mu}}\right)^{\frac{1}{\alpha_{G}}},\left(\frac{\epsilon}{9\sqrt{NLE_{2}}d^{\frac{1}{p}}}\right)^{\frac{2}{\alpha_{G}}}\right\}
\]
and for $\mu$ is small enough, then the above inequality implies
for
\[
K\ge\tilde{O}\left(\frac{{\displaystyle d^{\frac{2}{\beta}\left(\lceil\frac{2\alpha_{GN}^{2}}{\beta}\rceil\frac{1}{\alpha_{G}}+\lceil\frac{\alpha_{GN}+2}{\beta}\rceil\left(\alpha_{GN}+2\right)\left(1+\frac{1}{\alpha_{G}}\right)\right)+2+\frac{4}{\alpha}}}}{\gamma_{1}^{\left(1+\frac{1}{\alpha_{G}}\right)}\epsilon^{\left(\alpha_{GN}+1\right)\left(1+\frac{1}{\alpha_{G}}\right)+\frac{1}{\alpha_{G}}}}\right),
\]
we have $W_{\beta}(p_{K},\ \pi)\le\frac{\epsilon}{2}+\frac{\epsilon}{2}=\epsilon$,
as desired. \end{proof}
\section{Extended result}
\subsection{Proof of Theorem \ref{thm:E3}\label{AE1}}
\begin{proof} Using Lemma 2, there exists $\breve{U}\left(x\right)\in C^{1}(R^{d})$
with its Hessian exists everywhere on $R^{d}$, and $\breve{U}$ is
convex on $R^{d}$ such that
\begin{equation}
\sup\left(\breve{U}(\ x)-U(\ x)\right)-\inf\left(\breve{U}(\ x)-U(\ x)\right)\leq2\sum_{i}L_{i}R^{1+\alpha_{i}}.
\end{equation}
We now prove that $U$ satisfies a Poincaré inequality with constant
$\frac{1}{32C_{K}^{2}d\left(\frac{a+b+2aR^{2}+3}{a}\right)}e^{-4\left(2\sum_{i}L_{i}R^{1+\alpha_{i}}\right)}.$
Since $\breve{U}$ is convex, by Theorem 1.2 of \citep{bobkov1999isoperimetric},
$\breve{U}$ satisfies Poincaré inequality with constant
\begin{align*}
\gamma & \geq\frac{1}{4C_{K}^{2}\int\left\Vert x-E_{\pi}(x)\right\Vert ^{2}\pi\left(x\right)dx}\\
 & \stackrel{_{1}}{\geq}\frac{1}{8C_{K}^{2}\left(E_{\pi}\left(\left\Vert x\right\Vert ^{2}\right)+\left\Vert E_{\pi}(x)\right\Vert ^{2}\right)}\\
 & \stackrel{}{\geq}\frac{1}{16C_{K}^{2}E_{\pi}\left(\left\Vert x\right\Vert ^{2}\right)},
\end{align*}
where $C_{K}$ is a universal constant, step $1$ follows from Young
inequality and the last line is due to Jensen inequality. In addition,
for $\left\Vert x\right\Vert >R+2\epsilon+\delta$ from $\beta-$dissipative
assumption, we have for some $a,$ $b>0,\left\langle \nabla\breve{U}(x),x\right\rangle =\left\langle \nabla U(x),x\right\rangle \geq a\left\Vert x\right\Vert ^{\beta}-b$,
while for $\left\Vert x\right\Vert \leq R+2\epsilon+\delta$ by convexity
of $\breve{U}$
\begin{align*}
\left\langle \nabla\breve{U}(x),x\right\rangle  & \geq0\\
 & \geq a\left\Vert x\right\Vert ^{\beta}-a\left(R+2\epsilon+\delta\right)^{2}\\
 & \geq a\left\Vert x\right\Vert ^{\beta}-2aR^{2}.
\end{align*}
so for every $x\in\mathbb{R}^{d},$
\[
\left\langle \nabla\breve{U}(x),x\right\rangle \geq a\left\Vert x\right\Vert ^{\beta}-\left(b+2aR^{2}\right).
\]
Therefore, $\breve{U}(\mathrm{x})$ also satisfies $\beta-$dissipative,
which implies
\[
E_{\breve{\pi}}\left(\left\Vert x\right\Vert ^{2}\right)\leq2d\left(\frac{a+b+2aR^{2}+3}{a}\right),
\]
so the Poincaré constant satisfies
\[
\gamma\stackrel{}{\geq}\frac{1}{32C_{K}^{2}d\left(\frac{a+b+2aR^{2}+3}{a}\right)}.
\]
From \citep{ledoux2001logarithmic}'s Lemma 1.2, we have $U$ satisfies
Poincaré inequality with constant
\[
\gamma\geq\frac{1}{32C_{K}^{2}d\left(\frac{a+b+2aR^{2}+3}{a}\right)}e^{-4\left(2\sum_{i}L_{i}R^{1+\alpha_{i}}\right)}.
\]
Now, applying Theorem \eqref{thm:C4}, we derive for $\alpha_{G}$-mixture
locally smooth with $\ell_{G}=0$, ULA converges in
\[
K\approx\tilde{O}\left(\frac{\left(32C_{K}^{2}d\left(\frac{a+b+2aR^{2}+3}{a}\right)e^{4\left(2\sum_{i}L_{i}R^{1+\alpha_{i}}\right)}\right){}^{1+\frac{1}{\alpha_{G}}}d^{\lceil\frac{\alpha_{GN}+2}{\beta}\rceil\left(\alpha_{GN}+2\right)\left(1+\frac{1}{\alpha_{G}}\right)+\frac{\lceil2\alpha_{GN}\rceil}{2}}}{\epsilon^{\frac{\alpha_{GN}^{2}+2\alpha_{GN}+2}{\alpha_{G}}}}\right)
\]
which is the desired result. Similarly, we have the convergence rate
are
\[
\tilde{O}\left(\frac{\left(32C_{K}^{2}d\left(\frac{a+b+2aR^{2}+3}{a}\right)e^{4\left(2\sum_{i}L_{i}R^{1+\alpha_{i}}\right)}\right)^{2}d^{2\lceil\frac{\alpha_{HN}+3}{\beta}\rceil\left(\alpha_{HN}+3\right)+\frac{\lceil\frac{4\alpha_{HN}}{\beta}\rceil}{2}+\frac{\lceil\frac{\left(\alpha_{HN}+1\right)\left(4\alpha_{HN}+4\right)}{\beta}\rceil}{2}}}{\epsilon^{2\alpha_{HN}+5}}\right)
\]
and
\[
K=\tilde{O}\left(\frac{\left(32C_{K}^{2}d\left(\frac{a+b+2aR^{2}+3}{a}\right)e^{4\left(2\sum_{i}L_{i}R^{1+\alpha_{i}}\right)}\right)^{2}d^{\frac{\lceil\frac{8}{\beta}\rceil}{2}+2\lceil\frac{4}{\beta}\rceil}}{\epsilon^{2}}\right)
\]
respectively if the potential is $\alpha_{H}$-mixture locally Hessian
smooth with $\ell_{H}=0$ or if the potential is $1$-smooth and $1$-Hessian
smooth.
\end{proof}
\section{Useful lemmas}
\subsection{Proof of Lemma~\ref{lem:F1}\label{AF1}}
\begin{lemma}\label{lem:F4}
 {[}\citet{erdogdu2020convergence}' Lemma 34{]}
The function $\left\Vert x\right\Vert ^{\alpha-2}x$ is $\alpha-1$-H\"{o}lder
for $1<\alpha<2$.
\end{lemma}
\begin{lemma}\label{lem:F4-1} The function $\left\Vert x\right\Vert ^{\alpha}$
is $\frac{\alpha-1}{n}$-locally smooth for $1\leq n-1<\alpha-1\leq n$.
\end{lemma} \begin{proof} Without loss of generality, assume $\left\Vert y\right\Vert \leq\left\Vert x\right\Vert $
which implies $\left\Vert y\right\Vert -\left\Vert x\right\Vert \leq\left\Vert x-y\right\Vert \leq\left\Vert x\right\Vert +\left\Vert y\right\Vert $.
Therefore,
\begin{align*}
 & {\displaystyle \left\Vert \nabla f(x)-\nabla f(y)\right\Vert }\\
 & \leq\left\Vert \left\Vert x\right\Vert ^{\alpha-2}x-\left\Vert y\right\Vert ^{\alpha-2}y\right\Vert \\
 & \leq\left\Vert \left\Vert x\right\Vert ^{\alpha-2}x-\left\Vert x\right\Vert ^{\alpha-1}\frac{y}{\left\Vert y\right\Vert }+\left\Vert x\right\Vert ^{\alpha-1}\frac{y}{\left\Vert y\right\Vert }-\left\Vert y\right\Vert ^{\alpha-2}y\right\Vert \\
 & \leq\left\Vert x\right\Vert ^{\alpha-1}\left\Vert \frac{x}{\left\Vert x\right\Vert }-\frac{y}{\left\Vert y\right\Vert }\right\Vert +\left|\left\Vert x\right\Vert ^{\alpha-1}-\left\Vert y\right\Vert ^{\alpha-1}\right|\\
 & \stackrel{_{1}}{\leq}\left\Vert x\right\Vert ^{\alpha-1}\left\Vert \frac{x}{\left\Vert x\right\Vert }-\frac{y}{\left\Vert x\right\Vert }+\frac{y}{\left\Vert x\right\Vert }-\frac{y}{\left\Vert y\right\Vert }\right\Vert +\left\Vert x-y\right\Vert ^{\frac{\alpha-1}{n}}\left(\left\Vert x\right\Vert ^{\frac{\left(n-1\right)\left(\alpha-1\right)}{n}}+\ldots+\left\Vert y\right\Vert ^{\frac{\left(n-1\right)\left(\alpha-1\right)}{n}}\right)\\
 & \leq\left\Vert x\right\Vert ^{\alpha-2}\left\Vert x-y\right\Vert +\left\Vert x\right\Vert ^{\alpha-1}\left\Vert \frac{y}{\left\Vert y\right\Vert }\left(\frac{\left\Vert y\right\Vert }{\left\Vert x\right\Vert }-1\right)\right\Vert +\left\Vert x-y\right\Vert ^{\alpha-1}\\
 & \leq\left\Vert x-y\right\Vert ^{\frac{\alpha-1}{n}}\left(2\left\Vert x\right\Vert ^{\alpha-2}\left\Vert x-y\right\Vert ^{1-\frac{\alpha-1}{n}}+\left\Vert x\right\Vert ^{\frac{\left(n-1\right)\left(\alpha-1\right)}{n}}+\ldots+\left\Vert y\right\Vert ^{\frac{\left(n-1\right)\left(\alpha-1\right)}{n}}\right)\\
 & \leq n\left\Vert x-y\right\Vert ^{\frac{\alpha-1}{n}}\left(2\left\Vert x\right\Vert ^{\alpha-2}\left\Vert x\right\Vert ^{1-\frac{\alpha-1}{n}}+2\left\Vert x\right\Vert ^{\alpha-2}\left\Vert y\right\Vert ^{1-\frac{\alpha-1}{n}}+\left\Vert x\right\Vert ^{\frac{\left(n-1\right)\left(\alpha-1\right)}{n}}+\ldots+\left\Vert y\right\Vert ^{\frac{\left(n-1\right)\left(\alpha-1\right)}{n}}\right)\\
 & \leq\left(n+5\right)\left\Vert x-y\right\Vert ^{\frac{\alpha-1}{n}}\left(1+\left\Vert x\right\Vert ^{\frac{\left(n-1\right)\left(\alpha-1\right)}{n}}+\left\Vert y\right\Vert ^{\frac{\left(n-1\right)\left(\alpha-1\right)}{n}}\right).
\end{align*}
This is the desired result. \end{proof}
\begin{lemma}\label{lem:F4-1-1} The function $\left\Vert x\right\Vert ^{\alpha}$
is $\alpha-2$-locally Hessian smooth for $2<\alpha\leq3$. \end{lemma}
\begin{proof} Without loss of generality, assume $\left\Vert y\right\Vert \leq\left\Vert x\right\Vert $
which implies $\left\Vert y\right\Vert -\left\Vert x\right\Vert \leq\left\Vert x-y\right\Vert \leq\left\Vert x\right\Vert +\left\Vert y\right\Vert \leq2\left\Vert x\right\Vert $,
which in turn implies $\left\Vert x\right\Vert ^{\alpha-3}\leq2^{3-\alpha}\left\Vert x-y\right\Vert ^{\alpha-3}$.
Therefore,
\begin{align*}
 & {\displaystyle \left\Vert \nabla^{2}f(x)-\nabla^{2}f(y)\right\Vert _{\mathrm{op}}}\\
 & \leq\alpha\left\Vert \left\Vert x\right\Vert ^{\alpha-2}I+\left(\alpha-2\right)\left\Vert x\right\Vert ^{\alpha-3}\frac{xx^{T}}{\left\Vert x\right\Vert }-\left\Vert y\right\Vert ^{\alpha-2}I-\left(\alpha-2\right)\frac{yy^{T}}{\left\Vert y\right\Vert }\left\Vert y\right\Vert ^{\alpha-3}\right\Vert _{\mathrm{op}}\\
 & \leq\left\Vert y\right\Vert ^{\alpha-2}-\left\Vert x\right\Vert ^{\alpha-2}+\left(\alpha-2\right)\left\Vert \left\Vert x\right\Vert ^{\alpha-4}xx^{T}-\left\Vert x\right\Vert ^{\alpha-4}yx^{T}+\left\Vert x\right\Vert ^{\alpha-4}yx^{T}-yy^{T}\left\Vert y\right\Vert ^{\alpha-4}\right\Vert _{\mathrm{op}}\\
 & \leq\left\Vert y\right\Vert ^{\alpha-2}-\left\Vert x\right\Vert ^{\alpha-2}+\left(\alpha-2\right)\left\Vert x\right\Vert ^{\alpha-3}\left\Vert x-y\right\Vert +\left(\alpha-2\right)\left\Vert \left\Vert x\right\Vert ^{\alpha-4}yx^{T}-\left\Vert x\right\Vert ^{\alpha-4}yy^{T}+\left\Vert x\right\Vert ^{\alpha-4}yy^{T}-yy^{T}\left\Vert y\right\Vert ^{\alpha-4}\right\Vert _{\mathrm{op}}\\
 & \leq\left\Vert y\right\Vert ^{\alpha-2}-\left\Vert x\right\Vert ^{\alpha-2}+\left(\alpha-2\right)\left\Vert x\right\Vert ^{\alpha-3}\left\Vert x-y\right\Vert +\left(\alpha-2\right)\left\Vert x\right\Vert ^{\alpha-4}\left\Vert y\right\Vert \left\Vert x-y\right\Vert +\left(\alpha-2\right)\left(\left\Vert x\right\Vert ^{\alpha-4}-\left\Vert y\right\Vert ^{\alpha-4}\right)\left\Vert y\right\Vert ^{2}\\
 & \stackrel{}{\leq}\left\Vert y\right\Vert ^{\alpha-2}-\left\Vert x\right\Vert ^{\alpha-2}+2\left(\alpha-2\right)\left\Vert x\right\Vert ^{\alpha-3}\left\Vert x-y\right\Vert +\left(\alpha-2\right)\left\Vert x\right\Vert ^{\alpha-4}\left\Vert x-y\right\Vert ^{2}\\
 & +\left(\alpha-2\right)\left(\left\Vert y\right\Vert ^{\alpha-2}-\left\Vert x\right\Vert ^{\alpha-2}\right)+\left(\alpha-2\right)\left\Vert x\right\Vert ^{\alpha-4}\left(\left\Vert y\right\Vert ^{2}-\left\Vert x\right\Vert ^{2}\right)\\
 & \stackrel{}{\leq}\left(\alpha-1\right)\left(\left\Vert y\right\Vert ^{\alpha-2}-\left\Vert x\right\Vert ^{\alpha-2}\right)+2\left(\alpha-2\right)\left\Vert x\right\Vert ^{\alpha-3}\left\Vert x-y\right\Vert +\left(\alpha-2\right)\left\Vert x\right\Vert ^{\alpha-4}\left\Vert x-y\right\Vert ^{2}\\
 & +\left(\alpha-2\right)\left\Vert x\right\Vert ^{\alpha-4}\left(2\left\Vert x\right\Vert +\left\Vert x-y\right\Vert \right)\left(\left\Vert x-y\right\Vert \right)\\
 & \stackrel{}{\leq}\left(\alpha-1\right)\left(\left\Vert y\right\Vert ^{\alpha-2}-\left\Vert x\right\Vert ^{\alpha-2}\right)+4\left(\alpha-2\right)\left\Vert x\right\Vert ^{\alpha-3}\left\Vert x-y\right\Vert +2\left(\alpha-2\right)\left\Vert x\right\Vert ^{\alpha-4}\left\Vert x-y\right\Vert ^{2}\\
 & \leq\left(\alpha-1+\left(\alpha-2\right)2^{6-\alpha}\right)\left\Vert x-y\right\Vert ^{\alpha-2},
\end{align*}
where the last inequality follows from power expansion and triangle
inequality. This is the desired result. \end{proof}
\begin{lemma}\label{lem:F4-1-1-1} The function $\left\Vert x\right\Vert ^{\alpha}$
is $\frac{\alpha-1}{n}$-locally smooth for $1\leq n-1<\alpha-2\leq n$.
\end{lemma} \begin{proof} Without loss of generality, assume $\left\Vert y\right\Vert \leq\left\Vert x\right\Vert $
which implies $\left\Vert y\right\Vert -\left\Vert x\right\Vert \leq\left\Vert x-y\right\Vert \leq\left\Vert x\right\Vert +\left\Vert y\right\Vert \leq2\left\Vert x\right\Vert $.
Therefore,
\begin{align*}
 & {\displaystyle \left\Vert \nabla^{2}f(x)-\nabla^{2}f(y)\right\Vert _{\mathrm{op}}}\\
 & \leq\left(\alpha-1\right)\left|\left\Vert y\right\Vert ^{\alpha-2}-\left\Vert x\right\Vert ^{\alpha-2}\right|+4\left(\alpha-2\right)\left\Vert x\right\Vert ^{\alpha-3}\left\Vert x-y\right\Vert +2\left(\alpha-2\right)\left\Vert x\right\Vert ^{\alpha-4}\left\Vert x-y\right\Vert ^{2}\\
 & \leq\left(\alpha-1\right)\left\Vert x-y\right\Vert ^{\frac{\alpha-2}{n}}\left(\left\Vert x\right\Vert ^{\frac{\left(n-1\right)\left(\alpha-2\right)}{n}}+\ldots+\left\Vert y\right\Vert ^{\frac{\left(n-1\right)\left(\alpha-2\right)}{n}}\right)\\
 & +4\left(\alpha-2\right)\left\Vert x\right\Vert ^{\alpha-3}\left\Vert x-y\right\Vert +2\left(\alpha-2\right)\left\Vert x\right\Vert ^{\alpha-4}\left\Vert x-y\right\Vert ^{2}\\
 & \leq\left(n\left(\alpha-1\right)+6\left(\alpha-2\right)\right)\left\Vert x-y\right\Vert ^{\frac{\alpha-2}{n}}\left(1+\left\Vert x\right\Vert ^{\frac{\left(n-1\right)\left(\alpha-2\right)}{n}}+\left\Vert y\right\Vert ^{\frac{\left(n-1\right)\left(\alpha-2\right)}{n}}\right).
\end{align*}
This is the desired result. \end{proof}
\begin{lemma}\label{lem:F1} Suppose $\pi=e^{-U}$ satisfies $\alpha$-mixture
weakly smooth. Let $p_{\mu,0}=N(0,\frac{1}{L}I)$. Then $H(p_{\mu,0}|\pi_{\mu})\le U(0)+\frac{NL\mu^{1+\alpha}}{(1+\alpha)}d^{\frac{2}{2\wedge p}}-\frac{d}{2}\log\frac{2\Pi e}{L}+\frac{Nd}{1+\alpha}=O(d).$
\end{lemma} \begin{proof} \label{BInitial}Since $U$ is mixture
weakly smooth, for all $x\in\mathbb{R}^{d}$ we have
\begin{align*}
U_{\mu}(x) & \le U(0)+\langle\nabla U(0),x\rangle+\frac{L}{1+\alpha}\sum_{i}\left\Vert x\right\Vert ^{1+\alpha_{i}}+\frac{NL\mu^{1+\alpha}}{(1+\alpha)}d^{\frac{2}{2\wedge p}}\\
 & \leq U(0)+\frac{L}{1+\alpha}\sum_{i}\left\Vert x\right\Vert ^{1+\alpha_{i}}+\frac{NL\mu^{1+\alpha}}{(1+\alpha)}d^{\frac{2}{2\wedge p}}.
\end{align*}
Let $X\sim\rho=N(0,\frac{1}{L}I)$. Then
\begin{align*}
\mathbb{E}_{\rho}[U(X)] & \le U(0)+\frac{NL\mu^{1+\alpha}}{(1+\alpha)}d^{\frac{2}{2\wedge p}}+\frac{L}{1+\alpha}\sum_{i}\mathbb{E}_{\rho}\left(\left\Vert x\right\Vert ^{1+\alpha_{i}}\right)\\
 & \leq U(0)+\frac{NL\mu^{1+\alpha}}{(1+\alpha)}d^{\frac{2}{2\wedge p}}+\frac{L}{1+\alpha}\sum_{i}\mathbb{E}_{\rho}\left(\left\Vert x\right\Vert ^{2}\right)^{\frac{1+\alpha_{i}}{2}}\\
 & \leq U(0)+\frac{NL\mu^{1+\alpha}}{(1+\alpha)}d^{\frac{2}{2\wedge p}}+\frac{L}{1+\alpha}\sum_{i}\left(\frac{d}{L}\right)^{\frac{1+\alpha_{i}}{2}}\\
 & \leq U(0)+\frac{NL\mu^{1+\alpha}}{(1+\alpha)}d^{\frac{2}{2\wedge p}}+\frac{Nd}{1+\alpha}.
\end{align*}
Recall the entropy of $\rho$ is $H(\rho)=-\mathbb{E}_{\rho}[\log\rho(X)]=\frac{d}{2}\log\frac{2\Pi e}{L}$.
Therefore, the KL divergence is
\begin{align*}
\mathbb{E}(\rho|\pi) & =\int\rho\left(\log\rho+U\right)dx\\
 & =-H(\rho)+\mathbb{E}_{\rho}[U]\\
 & \le U(0)+\frac{NL\mu^{1+\alpha}}{(1+\alpha)}d^{\frac{2}{2\wedge p}}-\frac{d}{2}\log\frac{2\Pi e}{L}+\frac{Nd}{1+\alpha}\\
 & =O(d).
\end{align*}
This is the desired result. \end{proof}
\subsection{Proof of Lemma \ref{lem:F2a}\label{AF2}}
\begin{lemma}\label{lem:F2a} Assume $\pi_{\mu}=e^{-U_{\mu}(x)}$
then
\[
\mathbb{E}_{\pi_{\mu}}\left[\left\Vert \nabla U(x)\right\Vert ^{2}\right]\le\frac{2NL}{\mu^{1-\alpha}}d^{\frac{2}{p}}d^{\frac{2}{p}}+2\left(\frac{NL\mu^{1+\alpha}}{(1+\alpha)}d^{\frac{2}{2\wedge p}}\right)^{2},
\]
for $d$ sufficiently large.\end{lemma} \begin{proof} Since $\pi_{\mu}$
is stationary distribution, we have
\begin{align*}
\mathbb{E}_{\pi_{\mu}}\left[\left\Vert \nabla U_{\mu}(x)\right\Vert ^{2}\right] & =\mathbb{E}_{\pi_{\mu}}\left(\triangle U_{\mu}\left(x\right)\right)\\
 & \stackrel{}{\leq}\frac{NL}{\mu^{1-\alpha}}d^{\frac{2}{p}},
\end{align*}
where the last step comes from Lemma \ref{lem:B1}that $\nabla U_{\mu}\left(x\right)$
is $\frac{NL}{\mu^{1-\alpha}}d^{\frac{2}{p}}$-Lipschitz, $\nabla^{2}U_{\mu}\left(x\right)\preceq\frac{NL}{\mu^{1-\alpha}}d^{\frac{2}{p}}I$.
In addition,
\begin{align*}
\mathbb{E}_{\pi_{\mu}}\left[\left\Vert \nabla U(x)\right\Vert ^{2}\right] & \le2\mathbb{E}_{\pi_{\mu}}\left[\left\Vert \nabla U_{\mu}(x)\right\Vert ^{2}+\left\Vert \nabla U_{\mu}(x)-\nabla U(x)\right\Vert ^{2}\right]\\
 & \le\frac{2NL}{\mu^{1-\alpha}}d^{\frac{2}{p}}d^{\frac{2}{p}}+2\left(\frac{NL\mu^{1+\alpha}}{(1+\alpha)}d^{\frac{2}{2\wedge p}}\right)^{2},
\end{align*}
where the last step follows from Lemma \ref{lem:B1}. This gives the
desired result. \end{proof}
\subsection{Proof of lemma \ref{lem:F3}\label{AF3}}
\begin{lemma} If $U$ satisfies Assumptions \ref{A0}and \ref{A3},
then
\begin{center}
\begin{equation}
U(x)\geq\frac{a}{2\beta}\Vert x\Vert^{\beta}+U(0)-\frac{L}{\alpha+1}\sum_{i}R^{\alpha_{i}+1}-\frac{b}{\beta}.
\end{equation}
\par\end{center}
\label{lem:F3} \end{lemma}
\subsection{Proof of Lemma \ref{lem:F4}\label{AF4}}
\begin{lemma} \label{lem:F4}Assume that $U$ satisfies Assumptions
\ref{A0} and \ref{A3}, then for $\pi=e^{-U}$ and any distribution
$p$, we have for $\beta>0$,
\[
W_{\beta}^{\beta}(p,\ \pi)\leq\frac{4a}{\beta}\left(1.5+\tilde{d}+\tilde{c}_{\mu}\right)H(p_{\mu,k}|\pi_{\mu})+\frac{4a}{\beta}\left(1.5+\tilde{d}+\tilde{c}_{\mu}\right),
\]
where
\begin{align}
\tilde{c}_{\mu} & =\frac{1}{2}\log(\frac{2}{\beta})+\frac{L}{\alpha+1}\sum_{i}\left(\frac{2b}{a}\right)^{\frac{\alpha_{i}+1}{\beta}}+\frac{b}{\beta}+|U(0)|+\frac{NL\mu^{1+\alpha}}{(1+\alpha)}d^{\frac{2}{2\wedge p}},\\
\tilde{d} & =\frac{d}{\beta}\left[\frac{\beta}{2}log\left(\Pi\right)+\log\left(\frac{4\beta}{a}\right)+(1-\frac{\beta}{2})\log(\frac{d}{2e})\right].
\end{align}
\end{lemma}
\subsection{Proof of Lemma \ref{lem:F5}\label{AF5}}

\begin{lemma}\label{lem:F5} If the potential $U$ satisfies $\beta$-dissipative
and $\alpha$-mixture weakly smooth with $2\alpha_{N}\leq\beta$,
then let $p_{k}$ be the distribution of $x_{k}$ of ULA with a step
size satisfying ${\displaystyle \eta\leq\frac{1}{2}\left(1\wedge\frac{a}{2N^{2}L^{2}}\right)}$,
we have for any even integer $s\geq2$,
\begin{center}
\[
\mathrm{M}_{s}(p_{k}+\pi)\leq\mathrm{M}_{s}(p_{0}+\pi)+C_{s}k\eta,
\]
\par\end{center}
where
\begin{center}
$C_{s}\stackrel{\triangle}{=}{\displaystyle \left(\frac{3a+2b+3}{1\wedge a}\right)^{\frac{s-2}{\beta}+1}s^{s}d^{\frac{s-2}{\beta}+1}},$
\par\end{center}
\[
\mathrm{M}_{s}(p_{0}+\pi)\leq2(\frac{3a+b+3}{a})^{s/\beta}s^{s/\beta}d^{s/\beta}.
\]
\end{lemma} \begin{proof} Since $U$ satisfies $\alpha$-mixture
weakly smooth, we have
\begin{align*}
\left\Vert \nabla U(x)\right\Vert  & \leq\sum_{i}L_{i}\left\Vert x\right\Vert ^{\alpha_{i}}\\
 & \leq\sum_{i}L\left\Vert x\right\Vert ^{\alpha_{i}}\\
 & \leq\sum_{i}L\left(\left\Vert x\right\Vert ^{\alpha_{N}}+1\right)\\
 & \leq NL\left(\left\Vert x\right\Vert ^{\alpha_{N}}+1\right)
\end{align*}
where $2\alpha_{N}\leq\beta$ by our assumption. Moreover, $U$ also
satisfies $\beta$-dissipative. From \citep{erdogdu2020convergence}
Proposition 2 and Lemma 22, we obtain the desired result. \end{proof}
\begin{lemma}\label{lem:F6} {[}\citep{nguyen2021unadjusted} Lemma
F.16{]} If $\xi\sim N_{p}\left(0,I_{d}\right)$ then $d^{\left\lfloor \frac{n}{p}\right\rfloor }\leq E(\left\Vert \xi\right\Vert _{p}^{n})\leq\left[d+\frac{n}{2}\right]^{\frac{n}{p}}$where$\left\lfloor x\right\rfloor $
denotes the largest integer less than or equal to $x.$ If $n=kp,$
then $E(\left\Vert \xi\right\Vert _{p}^{n})=d..(d+k-1)$. \end{lemma}
\begin{lemma}\label{lem:F7} {[}\citep{nguyen2021unadjusted} Lemma
C.2{]} Assume $\pi=e^{-U(x)}$ is $\alpha$-mixture weakly smooth.
Then
\[
\mathbb{E}_{\pi}\left[\left\Vert \nabla U(x)\right\Vert ^{2}\right]\le2NL^{2}d^{\frac{3}{p}},
\]
for $d$ sufficiently large.\end{lemma} \begin{lemma}\label{lem:F0}
{[}\citep{nguyen2021unadjusted} Lemma 2.1{]} If potential $U:\mathbb{R}^{d}\rightarrow\mathbb{R}$
satisfies an $\alpha$-mixture weakly smooth for some $0<\alpha=\alpha_{1}<...<\alpha_{N}\leq1$,
$i=1,..,N$ $0<L_{i}<\infty$, then:
\begin{equation}
U(y)\leq U(x)+\left\langle \nabla U(x),\ y-x\right\rangle +\sum_{i}\frac{L_{i}}{1+\alpha_{i}}\Vert y-x\Vert^{1+\alpha_{i}}.\label{eq:4}
\end{equation}
\end{lemma}

\section*{Acknowledgements}
This research was funded in part by the University of Mississippi summer grant.

\bibliographystyle{apalike}

\end{document}